%% file: Filling_of_vacuum_for_BE.tex
\documentclass[12pt]{amsart}  

\usepackage{graphicx}
\usepackage{here} 
\usepackage{array} 

\usepackage{vmargin}
\usepackage{amssymb}
\usepackage{mathrsfs}
\usepackage[all]{xy}
\usepackage[usenames,dvipsnames]{color}
\usepackage{soul}
\RequirePackage[colorlinks,linkcolor=blue,citecolor=LimeGreen,urlcolor=red]{hyperref} 
\usepackage{amsmath}

\newtheorem{theorem}{Theorem}[section]

\newtheorem{cor}[theorem]{Corollary}

\newtheorem{defi}[theorem]{Definition}

\newtheorem{lemma}[theorem]{Lemma}

\newtheorem{prop}[theorem]{Proposition}
\newtheorem{remark}[theorem]{Remark}

\numberwithin{equation}{section}

\newcommand{\R}{\mathbb{R}}

\newcommand{\N}{\mathbb{N}}

\newcommand{\T}{\mathbb{T}}

\newcommand{\function}[5]{\begin{array}[t]{lrcl}
#1 : & #2 & \longrightarrow & #3 \\
     & #4 & \longmapsto & #5
\end{array}}
\newcommand{\func}[3]{#1 : #2 \longrightarrow #3}

\newcommand{\abs}[1]{\left|#1\right|}
\newcommand{\eps}{\varepsilon}
\newcommand{\norm}[1]{\left\|#1\right\|}

\renewcommand{\leq}{\leqslant}
\renewcommand{\geq}{\geqslant}

\renewcommand{\bar}{\overline}

\renewcommand{\tilde}{\widetilde}

\def\signmb{\bigskip \begin{center} {\sc
Marc Briant\par\vspace{3mm}
University of Cambridge\par
DPMMS, Centre for Mathematical Sciences\par
Wilberforce Road,
Cambridge CB3 0WA,
UK\par\vspace{3mm}
e-mail:} \tt{briant.maths@gmail.com} \end{center}}

\begin{document} 

\title[INSTANTANEOUS FILLING OF THE VACUUM FOR THE FULL BOLTZMANN EQUATION]{INSTANTANEOUS FILLING OF THE VACUUM FOR THE FULL BOLTZMANN EQUATION IN CONVEX DOMAINS}
\author{M. Briant}

\maketitle

\begin{abstract}
We prove the immediate appearance of a lower bound for mild solutions to the full Boltzmann equation in the torus or a $C^2$ convex domain with specular boundary conditions, under the sole assumption of continuity away from the grazing set of the solution. These results are entirely constructive if the domain is $C^3$ and strictly convex. We investigate a wide range of collision kernels, some satisfying Grad's cutoff assumption and others not. We show that this lower bound is exponential, independent of time and space with explicit constants depending only on the \textit{a priori} bounds on the solution. In particular, this lower bound is Maxwellian in the case of cutoff collision kernels. A thorough study of characteristic trajectories, as well as a geometric approach of grazing collisions against the boundary are derived.
\end{abstract}

\vspace*{10mm}



\smallskip
\textbf{Acknowledgements:} I would like to acknowledge my supervisor, Cl\'{e}ment Mouhot, for suggesting me this problem and for all the fruitful discussions and advices he offered me. I also would like to thank Alexandre Boritchev, Amit Einav and Sara Merino for fruitful discussions.
\\ This work was supported by the UK Engineering and Physical Sciences Research Council (EPSRC) grant EP/H023348/1 for the University of Cambridge Centre for Doctoral Training, the Cambridge Centre for Analysis.
\tableofcontents

\input{introduction}

\input{mainresults}

\input{boundCutoff_1ststep}

\input{boundCutoff_far}

\input{boundCutoff_grazing}

\input{boundCutoff_finalproof}

\input{boundNonCutoff}


\appendix

\input{annexe1_transport}

\bibliographystyle{acm}
\bibliography{bibliography}


\signmb

\end{document}

%% file: introduction.tex
\section{Introduction} \label{sec:intro}

This paper deals with the Boltzmann equation, which rules the behaviour of rarefied gas particles moving in a domain $\Omega$ of $\R^d$ with velocities in $\R^d$ ($d \geq 2$) when the only interactions taken into account are binary collisions. More precisely, the Boltzmann equation describes the time evolution of $f(t,x,v)$, the distribution of particles in position and velocity, starting from an initial distribution $f_0(x,v)$ .
\par We investigate the case where $\Omega$ is either the torus or a $C^2$ convex bounded domain. The Boltzmann equation reads

\begin{eqnarray}
\forall t \geq 0 &,& \:\forall (x,v) \in \Omega \times \R^d,\quad  \partial_t f + v\cdot \nabla_x f = Q(f,f),\label{BE}
\\ && \:\forall (x,v) \in \bar{\Omega} \times \R^d,\quad f(0,x,v) = f_0(x,v), \nonumber
\end{eqnarray}
with $f$ being periodic in the case of $\Omega = \T^d$, the torus, or with $f$ satisfying the specular reflections boundary condition if $\Omega$ is a $C^2$ convex bounded domain:
\begin{equation}\label{specularboundary}
\forall (x,v) \in \partial\Omega\times\R^d,\quad  f(t,x,v) = f(t,x,\mathcal{R}_x(v)).
\end{equation}

$\mathcal{R}_x$, for $x$ on the boundary of $\Omega$, stands for the specular reflection at that point of the boundary. One can compute, denoting by $n(x)$ the outward normal at a point $x$ on $\partial\Omega$,
$$\forall v \in \R^d,\quad \mathcal{R}_x(v) = v - 2(v\cdot n(x))n(x).$$

\bigskip
The quadratic operator $Q(f,f)$ is local in time and space and is given by 
$$Q(f,f) =  \int_{\R^d\times \mathbb{S}^{d-1}}B\left(|v - v_*|,\mbox{cos}\:\theta\right)\left[f'f'_* - ff_*\right]dv_*d\sigma,$$
where $f'$, $f_*$, $f'_*$ and $f$ are the values taken by $f$ at $v'$, $v_*$, $v'_*$ and $v$ respectively. Define:
$$\left\{ \begin{array}{rl} \displaystyle{v'} & \displaystyle{= \frac{v+v_*}{2} +  \frac{|v-v_*|}{2}\sigma} \vspace{2mm} \\ \vspace{2mm} \displaystyle{v' _*}&\displaystyle{= \frac{v+v_*}{2}  -  \frac{|v-v_*|}{2}\sigma} \end{array}\right., \: \mbox{and} \quad \mbox{cos}\:\theta = \langle \frac{v-v_*}{\abs{v-v_*}},\sigma\rangle .$$

\bigskip
The collision kernel $B \geq 0$ contains all the information about the interaction between two particles and is determined by physics (see \cite{Ce} or \cite{Ce1} for a formal derivation for the hard sphere model of particles). In this paper we shall only be interested in the case of $B$ satisfying the following product form
\begin{equation}\label{assumptionB}
B\left(|v - v_*|,\mbox{cos}\:\theta\right) = \Phi\left(|v - v_*|\right)b\left( \mbox{cos}\:\theta\right),
\end{equation}
which is a common assumption as it is more convenient and also covers a wide range of physical applications.
Moreover, we shall assume that $\Phi$ satisfies either
\begin{equation}\label{assumptionPhi}
\forall z \in \R,\quad c_\Phi \abs{z}^\gamma \leq \Phi(z) \leq C_\Phi \abs{z}^\gamma
\end{equation}
or a mollified assumption
\begin{equation}\label{assumptionPhimol}
\left\{\begin{array}{rl}\displaystyle{\forall \abs{z} \geq 1 \in \R,}&\displaystyle{\quad c_\Phi \abs{z}^\gamma \leq \Phi(z) \leq C_\Phi \abs{z}^\gamma} \vspace{2mm} \\\vspace{2mm} \displaystyle{\forall \abs{z} \leq 1 \in \R,}&\displaystyle{\quad c_\Phi \leq \Phi(z) \leq C_\Phi,} \end{array}\right.
\end{equation}
$c_\Phi$ and $C_\Phi$ being strictly positive constants and $\gamma$ in $(-d,1]$. The collision kernel is said to be ``hard potential" in the case of $\gamma >0$, ``soft potential" if  $\gamma < 0$ and ``Maxwellian" if $\gamma = 0$.
\par Finally, we shall consider $b$ to be a continuous function on $\theta$ in $(0,\pi]$, strictly positive near $\theta \sim \pi/2$, which satisfies
\begin{equation}\label{assumptionb}
b\left(\mbox{cos}\:\theta\right)\mbox{sin}^{d-2}\theta\mathop{\sim}\limits_{\theta \to 0^+}b_0 \theta^{-(1+\nu)}
\end{equation}
for $b_0 >0$ and $\nu$ in $(-\infty,2)$. The case when $b$ is locally integrable, $\nu < 0$, is referred to by the Grad's cutoff assumption (first introduce in \cite{Gr1}) and therefore $B$ will be said to be a cutoff collision kernel. The case $\nu \geq 0$ will be designated by non-cutoff collision kernel.


\subsection{Motivations and comparison with previous results} \label{subsec:previously}
The aim of this article is to show and to quantify the strict positivity of the solutions to the Boltzmann equation when  the gas particles move in a bounded domain. This issue has been investigated for a long time since it not only presents a great physical interest but also appears to be of significant importance for the mathematical study of the Boltzmann equation.
\par Moreover, our results only require some regularity on the solution and no further assumption on its local density, which was assumed to be uniformly bounded from below in previous studies (which is equivalent of assuming \textit{a priori} either that there is no vacuum or that the solution is strictly positive).
\par More precisely, we shall prove that solutions to the Boltzmann equation in a $C^2$ convex bounded domain or the torus that which have uniformly bounded energy satisfy an immediate exponential lower bound:
$$\forall t_0>0, \: \exists K,\:C_1,\:C_2 >0, \: \forall t \geq t_0,\:\forall (x,v) \in \Omega\times\R^d, \quad f(t,x,v) \geq C_1e^{-C_2\abs{v}^{2+K}},$$
with $K=0$ (Maxwellian lower bound) in the case of a collision kernel with angular cutoff.
\par We emphasize here that the present results only require solutions to the Boltzmann equation to be continuous away from the grazing set
\begin{equation}\label{grazingset}
\Lambda_0 = \left\{(x,v)\in \partial\Omega\times\R^d,\:\quad n(x)\cdot v =0\right\},
\end{equation}
which is a property that is known to hold in the case of specular reflection boundary conditions \cite{Gu6}(see also \cite{Gu5} or \cite{HwVel} for boundary value problems for mean-field equations). We also note that more physically relevant boundary conditions are a combination of specular reflections with some diffusion process at the boundary. The same kind of exponential lower bounds, with the same assumptions on the solution, have recently been obtained by the author in the case of Maxwellian diffusion boundary conditions \cite{Bri5}.

\bigskip
The strict positivity of the solutions to the Boltzmann equation standing in the form of an exponential lower bound was already noticed by Carleman in \cite{Ca} for the spatially homogeneous equation. In his article he proved that such a lower bound is created immediately in time in the case of hard potential kernels with cutoff in dimension $3$. More precisely, the radially symmetric solutions he constructed in \cite{Ca} satisfies an almost Maxwellian lower bound,
$$\forall t \geq t_0, \forall v \in \R^3, \quad f(t,v) \geq C_1 e^{-C_2\abs{v}^{2+\eps}},$$
$C_1,C_2 >0$ for all $t_0 >0$ and $\eps >0$. His breakthrough was to notice that a part $Q^+$ of the Boltzmann operator $Q$ satisfies a spreading property, roughly speaking
$$Q^+(\mathbf{1}_{B(\bar{v},r)},\mathbf{1}_{B(\bar{v},r)}) \geq C_+\mathbf{1}_{B\left(\bar{v},\sqrt{2}r\right)},$$
with $C_+<1$ (see Lemma $\ref{lem:Q+spread}$ for an exact statement). 
\par The spreading strategy was used by Pulvirenti and Wennberg in \cite{PW} to extend the latter inequality to solutions to the spatially homogeneous Boltzmann equation with hard potential and cutoff in dimension $3$ with more general initial data. Their contribution was to get rid of the inital boundedness suggested in \cite{Ca} by Carleman thanks to the use of an iterative regularity property of the $Q^+$ operator. This property allowed them to immediately create an ``upheaval point" that they then spread with the method of Carleman. Moreover, they obtain an exact Maxwellian lower bound of the form by controlling the decay of $C_+^n$
$$\forall t \geq t_0, \forall v \in \R^3, \quad f(t,v) \geq C_1 e^{-C_2\abs{v}^{2}},$$
for all $t_0>0$.
\par Finally, Mouhot in \cite{Mo2} dealt with the full Boltzmann equation in the torus. He derived a spreading method that is invariant under the flow of the characteristics, obtaining lower bounds uniformly in space  as long as the solution has uniformly bounded density, energy and entropy (for the hard potential case) together with uniform bounds on higher moments (for the soft and Maxwellian potentials case). However, he also implicitly assumed that the initial data had to be bounded from below uniformly in space. He also derived (\cite{Mo2}) the same kind of results in the non-cutoff case in the torus, the immediate appearance of an exponential lower bound of the form
$$\forall t \geq t_0, \forall (x,v) \in \T^d\times\R^d, \quad f(t,v) \geq C_1(\eps) e^{-C_2(\eps)\abs{v}^{K+\eps}},$$
for all $t_0>0$, all $\eps >0$ and $K=K(\nu)$ with $K(0)=2$ (thus recovering the cutoff case in the limit). His idea was to split further the $Q$ operator into a cutoff part and a non-cutoff part that is seen as a small perturbation of his original spreading method.
\par Our results extend those in \cite{Mo2} in the case of $C^2$ bounded convex domain. Our main contribution is the derivation of a spreading method that remains invariant under the characteristics flow that, unlike the torus case, changes the direction of velocities at the boundary. Moreover, we emphasize here that the existence of boundaries implies the existence of grazing collisions against the latter, where the strategies developped in \cite{PW} and \cite{Mo2} fail. We therefore derive a geometrical approach to those problematic trajectories. 
\par Furthermore, we do not assume any uniform boundedness on the initial data but we require the continuity of the solution to the Boltzmann equation. However, if we keep the assumptions made in \cite{Mo2} and further assume that the domain is $C^3$ and strictly convex then our proofs are constructive. 

\bigskip
The quantification of the strict positivity, and above all the appearance of an exponential lower bound, has been seen to be of great mathematical interest thanks to the development of the entropy-entropy production method. This method (see \cite{Vi2}, Chapter $3$, and \cite{Vi1}) provides a useful way of investigating the long-time behaviour of solutions to kinetic equations. Indeed, it has been successfully used to prove convergence to the equilibrium in non-pertubative cases for the Fokker-Planck equation, \cite{DV}, and the full Boltzmann equation in the torus or in $C^1$ bounded connected domains with specular reflections, \cite{DV1}. This entropy-entropy production method requires (see Theorem $2$ in \cite{DV1}) uniform boundedness on moments and Sobolev norms for the solutions to the Boltzmann equation but also an \textit{a priori} exponential lower bound of the form
$$f(t,x,v) \geq C_1e^{-C_2\abs{v}^q},$$
with $q\geq 2$.
\par Therefore, the present paper allows us to prove that the latter \textit{a priori} assumption is in fact satisfied for a lot of different cases (see \cite{Mo2}, Section $5$ for an overview). We also emphasize here that the assumption of continuity of the solution we have made does not reduce the range of applications since a lot more regularity is usually asked for the entropy-entropy production method. Moreover, our method, unlike the ones developed in \cite{Mo2} and \cite{PW}, does not require a uniform bound on the local density of solutions, which is not a requirement for the entropy-entropy production method either (see \cite{DV1}, Theorem $2$).

\bigskip
To conclude we note that our investigations require a deep and detailed understanding of the geometry and properties of characteristic trajectories for the free transport equation. In particular, a geometric approach of grazing collisions against the boundary is derived and is the key ingredient to study the strict positivity of solutions to the Boltzmann equation. The existing strategies as well as our improvements are discussed in the next section.


\subsection{Our strategy} \label{subsec:strategy}

Our strategy to tackle this issue will follow the method introduced by Carleman \cite{Ca} together with the idea of Mouhot \cite{Mo2} to find a spreading method that will be invariant along the characteristic trajectories. Roughly speaking we shall built characteristics in a $C^2$ bounded convex domain, create an ``upheaval point" (as in \cite{PW} and \cite{Mo2}) that we spread and expand uniformly along the characteristics. Finally, once the lower bound can be compared to an exponential one we reach the expected result.
\par However, the existence of rebounds against the boundary leads to difficulties. We describe them below and point out how we shall overcome them.

\bigskip
Creating an ``upheaval point" was achieved, in \cite{PW} and \cite{Mo2}, by using an iterated Duhamel formula and a regularity property of the collision operator relying on a uniform lower bound of the local density of the function. But the use of this property requires a uniform control along the characteristics of the density, the energy and the entropy of the solutions to the Boltzmann equation which is natural in the homogeneous case but made Mouhot consider initial datum bounded from below uniformly in space. Our way of dealing with the appearance of the ``upheaval point" is rather different but includes more general initial datum. We make the assumption of continuity of solutions to the Boltzmann equation and by compactness arguments we can construct a partition of our phase space where initial localised lower bounds exist, i.e., localised ``upheaval points".
\par The case on the torus studied by Mouhot tells us that an exponential lower bound should arise immediately and therefore we expect the same to happen as long as the characteristic trajectory is a straight line. Unfortunately, the possibility for a trajectory to remain a line depends on the distance from the boundary of the starting point, which can be as short as one wants. This thought is the basis of our means for spreading the initial lower bound. We divided our trajectories into two categories, the ones which always stay close to the boundary (grazing collisions) and the others. For the latter we can spread our lower bound uniformly as noticed in \cite{Mo2}. The key contribution of our proof is a thorough investigation of the geometry of grazing collisions. We show that their velocity does not evolve a lot along time and mix it with the spreading property of the collision operator. Notice here that the convexity of $\Omega$ is needed for the study of grazing trajectories.
\par The last behaviour to notice is the fact that specular reflections completely change velocities but preserve their norm. Therefore, the existence of rebounds against the boundary prevents us from obtaining a uniform spreading method straight from the ''upheaval point" unless it is depending only on the norm of the velocity. Our strategy is to spread the lower bound created at the ``upheaval points" independently for grazing and non-grazing collisions up to the point when the lower bound we obtain depends only on the norm of the velocity. Roughly, our lower bounds will be balls in velocity that can be centred away from the origin and we shall grow them up finitely many times to balls containing the origin and finally be able to generate a uniform spreading method.

\bigskip
Collision kernels satisfying a cutoff property as well as collision kernels with a non-cutoff property will be treated following the strategy described above. The only difference is the decomposition of the Boltzmann bilinear operator $Q$ we consider in each case. In the case of a non-cutoff collision kernel, we shall divide it into a cutoff collision kernel and a remainder. The cutoff part will already be dealt with and a careful control of the $L^\infty$-norm of the remainder will give us the expected lower bound, smaller than a Maxwellian lower bound.

\bigskip
A preliminary to our study (left in appendix) is to be able to construct the characteristic trajectories associated to the Boltzmann equation with specular reflections in a $C^2$ bounded convex domain. These trajectories are merely those of the free transport and so can be seen as the movement of a billiard ball inside the boundary of our domain.
\par Such a free transport in a convex domain has been studied in \cite{JM} (see also \cite{Ta}, \cite{Ta1} or \cite{Po} for geometrical properties) and has been used in kinetic theory by Guo, \cite{Gu5},  or Hwang, \cite{Hw}, for instance. Yet, the common feature in \cite{JM}, \cite{Gu5}, \cite{Gu6} and \cite{Hw} is that their assumptions on the boundary always lead to clear rebounds of the characteristic trajectories. That is to say, the absoption phenomenon of \cite{JM}, the electromagnetic field in \cite{Gu5} and \cite{Hw} or the smooth strict convexity assumption used in \cite{Gu6}, prevent the characteristics to roll on the boundary which is one of the possible behaviour we have to take into account in our general settings. As briefly mentionned in the introduction of \cite{Ta1}, the behaviour at some specific boundary points is mathematically quite unexpected, even if that is of no physical relevance. We thus classify all the possible outcomes of a rebound against the boundary and study them carefully to analytically build the characteristics for the free transport equation in our domain $\Omega$.
 \par Finally, we need to control the number of rebounds that can happen in a finite time. In \cite{Ta}, Tabachnikov focuses on the footprints on the boundary of the trajectories of billiard balls and shows that the initial conditions leading to infinitely many rebounds on the boudary is a set of measure $0$. We extend this to the whole trajectory (see Appendix $\ref{subsec:rebounds}$, Proposition $\ref{prop:nbrebounds}$), not only its footprints on the boundary, allowing us to consider only finitely many rebounds in finite time and to have an analytic formula for the characteristics which we shall use throughout the article.
\par Notice that all this study of the free transport equation will be done in the case of a merely $C^1$ bounded domain, which extends the results of \cite{Gu6}.


\subsection{Organisation of the paper}\label{subsec:organization}

Section $\ref{sec:mainresults}$ is dedicated to the statement and the description of the main results proved in this article. It contains four different parts
\par Section $\ref{subsec:notations}$ defines all the notations which will be used throughout the article.
\par As mentioned above, we shall investigate in detail the characteristics and the free transport equation in a $C^1$ bounded domain. Section $\ref{subsec:maintransport}$ mathematically formulates the intuitive ideas of trajectories.
\par The last subsections, $\ref{subsec:maincutoff}$ and $\ref{subsec:mainnoncutoff}$, are dedicated to a mathematical formulation of the results related to the lower bound in, respectively, the cutoff case and the non-cutoff case, described in Section $\ref{subsec:strategy}$. It also defines the concept of mild solutions to the Boltzmann equation in each case.

\bigskip
Sections $\ref{sec:cutoff_1ststep}$ to $\ref{sec:cutoff_finalproof}$ focuse on the Maxwellian lower bound in the cutoff case. It is divided into the four main arguments of the proof.
\par Following our strategy, Section $\ref{sec:cutoff_1ststep}$ creates the localised ``upheaval points" whereas Section $\ref{sec:cutoff_far}$ and Section $\ref{sec:cutoff_grazing}$ spread them along non-grazing and grazing trajectories respectively.
\par Section $\ref{sec:cutoff_finalproof}$ concludes by describing the immediate appearance of a lower bound depending only on the norm of the velocity ( Proposition $\ref{prop:centredball}$) as well as proving the immediate Maxwellian lower bound (proof of Theorem $\ref{theo:boundcutoff}$).

\bigskip
Finally, we deal with non-cutoff collision kernels in Section $\ref{sec:noncutoff}$ where we prove the immediate appearance of an exponential lower bound (Theorem $\ref{theo:boundnoncutoff}$). The proof follows exactly the same steps as in the case of cutoff kernels and is thus divided into Section $\ref{subsec:centredballNCO}$, where we construct a lower bound only depending on the norm of the velocity, and Section $\ref{subsec:prooftheononcutoff}$, where we derive the exponential lower bound.

\bigskip
As mentioned before, we need to study the free transport equation and the different important properties of the characteristics. Appendix $\ref{appendix:transport}$ formulates these issues, investigates all the different behaviours  of rebounds against the boundary (Section $\ref{subsec:rebounds}$), builds the characteristics and derives their properties (Section $\ref{subsec:characteristics}$) and solves the free transport equation (Section $\ref{subsec:transport}$).

%% file: mainresults.tex
\section{Main results} \label{sec:mainresults}

We begin with the notations we shall use all along this article.

\subsection{Notations}\label{subsec:notations}

We denote $\langle \cdot \rangle = \sqrt{1 + \abs{\cdot}^2}$ and $y^+ = \max\{0,y\}$, the positive part of $y$.
\\This study will hold in specific functional spaces regarding the $v$ variable that we describe here and use throughout the sequel. Most of them are based on natural Lebesgue spaces $L^p_v = L^p\left(\R^d\right)$ with a weight:
\begin{itemize}
\item for $p\in[1,\infty]$ and $q \in \R$, $L^p_{q,v}$ is the Lebesgue space with the following norm $$\norm{f}_{L^p_{q,v}} = \norm{\langle v \rangle^q f}_{L^p_v},$$
\item for $p \in [1,\infty]$ and $k \in \N$ we use the Sobolev spaces $W^{k,p}_v$ by the norm $$\norm{f}_{W^{k,p}_v} = \left[\sum\limits_{\abs{s}\leq k}\norm{\partial^sf(v)}^p_{L^p_v}\right]^{1/p},$$
with the usual convention $H^k_v = W^{k,2}_v$.
\end{itemize}

\bigskip
In what follows, we are going to need bounds on some physical observables of solution to the Boltzmann equation $\eqref{BE}$.
\par We consider here a function $f(t,x,v) \geq 0$ defined on $[0,T)\times \Omega\times \R^d$ and we recall the definitions of its local hydrodynamical quantities.

\bigskip
\begin{itemize}
\item its local energy  $$e_f (t,x) = \int_{\R^d}\abs{v}^2f(t,x,v)dv,$$
\item its local weighted energy  $$e'_f (t,x) = \int_{\R^d}\abs{v}^{\tilde{\gamma}}f(t,x,v)dv,$$ where $\tilde{\gamma} = (2+\gamma)^+$,
\item its local $L^p$ norm ($p \in [1,+\infty)$) $$l^p_f(t,x) = \norm{f(t,x,\cdot)}_{L^p_v},$$
\item its local $W^{2,\infty}$ norm $$w_f(t,x) = \norm{f(t,x,\cdot)}_{W^{2,\infty}_v}.$$
\end{itemize}
\bigskip

Our results depend on uniform bounds on those quantities and therefore, to shorten calculations we will use the following

\begin{eqnarray*}
E_f =  \sup\limits_{(t,x)\in [0,T)\times \Omega}e_f(t,x) &,& E'_f = \sup\limits_{(t,x)\in [0,T)\times \Omega}e'_f(t,x) ,
\\ L^p_f = \sup\limits_{(t,x)\in [0,T)\times \Omega} l^p_f(t,x) &,& W_f = \sup\limits_{(t,x)\in [0,T)\times \Omega} w_f(t,x).
\end{eqnarray*}

\bigskip
In our theorems we are giving a priori lower bound results for solutions to $\eqref{BE}$ satisfying some properties about their local hydrodynamical quantities. Those properties will differ depending on which case of collision kernel we are considering. We will take them as assumptions in our proofs and they are the following.
\begin{itemize}
\item In the case of hard or Maxwellian potentials with cutoff ($\gamma\geq 0$ and $\nu <0$):
\begin{equation}\label{assumption1}
E_f < +\infty.
\end{equation}
\item In the case of a singularity of the kinetic collision kernel ($\gamma \in (-d,0)$) we shall make the additional assumption
\begin{equation}\label{assumption2}
L^{p_\gamma}_f < +\infty,
\end{equation}
where $p_\gamma  > d/(d+\gamma)$.
\item In the case of a singularity of the angular collision kernel ($\nu \in [0,2)$) we shall make the additional assumption
\begin{equation}\label{assumption3}
W_f < +\infty, \:\: E'_f < + \infty.
\end{equation}
\end{itemize}
As noticed in \cite{Mo2}, in some cases several assumptions might be redundant.
\par Furthermore, in the case of the torus with periodic conditions or the case of bounded domain with specular boundary reflections, solutions to $\eqref{BE}$ also satisfy the following conservation laws (see \cite{Ce}, \cite{Ce1} or \cite{Vi2} for instance) for the total mass and the total energy:

\begin{equation}\label{conservations}
\exists \mbox{M},\:\mbox{E} \geq 0,\:\forall t \in \R^+, \quad \left\{\begin{array}{l} \displaystyle{\int_{\Omega}\int_{\R^d} f(t,x,v)\:dxdv = M,} \vspace{2mm}\\\vspace{2mm} \displaystyle{\int_{\Omega}\int_{\R^d} \abs{v}^2f(t,x,v)\:dxdv = E.} \end{array}\right.
\end{equation}

 
\subsection{Results about the free transport equation}\label{subsec:maintransport}

Our investigations start with the study of the characteristics of the free transport equation. We only focus on the case where $\Omega$ is not the torus (the characteristics in the torus being merely straight lines) but we will use the same notations in both cases. This is achieved by the following theorem.

\bigskip
\begin{theorem}\label{theo:transport}
Let $\Omega$ be an open, bounded and $C^1$ domain in $\R^d$.
\\Let $\func{u_0}{\bar{\Omega} \times \R^d}{\R}$ be $C^1$ in $x \in \Omega $ and in $L^2_{x,v}$.
\par The free transport equation with specular reflections reads
\begin{eqnarray}
\forall t \geq 0 &,& \:\forall (x,v) \in \Omega \times \R^d,\quad  \partial_t u(t,x,v) + D_x(v)(u)(t,x,v) = 0,  \label{transporteq}
\\ && \:\forall (x,v) \in \bar{\Omega} \times \R^d,\quad u(0,x,v) = u_0(x,v), \label{transportCI}
\\  && \: \forall (x,v) \in \partial \Omega \times \R^d,\quad  u(t,x,v) = u(t,x,\mathcal{R}_x(v)), \label{transportBC}
\end{eqnarray}
where $\mathcal{R}_x$ stands for the specular reflection at a point $x$ and $D_x(v)$ is the directional derivative at $x$ in the direction of $v$.
\\Then this equation has a unique solution $\func{u}{\R^+\times\bar{\Omega} \times \R^d}{\R}$ which is $C^1$ in time, admits a directional derivative in space in the direction of $v$ and is in $L^2_{x,v}$.
\par Moreover, for all $(t,x,v)$ in $\R^+\times\bar{\Omega} \times \R^d$, there exists $x_{fin}(t,x,v)$, $v_{fin}(t,x,v)$ and $t_{fin}(t,x,v)$ (see Definition $\ref{def:solution}$) such that
$$u(t,x,v) = u_0\left(x_{fin}-(t-t_{fin})v_{fin},v_{fin}\right).$$
\end{theorem}
\bigskip

This part of the article provides a thorough study of the characteristics of our system. However, it is independent of the rest of the work (apart for building solid grounds for trajectories) and therefore is left to Appendix $\ref{appendix:transport}$.

 
\subsection{Maxwellian lower bound for cutoff collision kernels}\label{subsec:maincutoff}

The final theorem we prove in the case of cutoff collision kernel is the immediate appearance of a uniform Maxwellian lower bound. We use, in that case, the Grad's splitting for the bilinear operator $Q$ such that the Boltzmann equation reads
\begin{eqnarray*}
Q(g,h) &=&  \int_{\R^d\times \mathbb{S}^{d-1}}\Phi\left(|v - v_*|\right)b\left( \mbox{cos}\theta\right)\left[h'g'_* - hg_*\right]dv_*d\sigma
\\     &=& Q^+(g,h) - Q^-(g,h),
\end{eqnarray*}
where we used the following definitions
\begin{eqnarray}
Q^+(g,h) &=& \int_{\R^d\times \mathbb{S}^{d-1}}\Phi\left(|v - v_*|\right)b\left( \mbox{cos}\theta\right)h'g'_*\: dv_*d\sigma,\nonumber
\\ Q^-(g,h) &=& n_b \left(\Phi * g (v)\right)h = L[g](v)h, \label{gradsplitting}
\end{eqnarray}
where
\begin{equation}\label{nb}
n_b = \int_{\mathbb{S}^{d-1}}b\left(\mbox{cos}\:\theta\right)d\sigma = \left|\mathbb{S}^{d-2}\right|\int_0^\pi b\left(\mbox{cos}\:\theta\right) \mbox{sin}^{d-2}\theta \:d\theta.
\end{equation}

\bigskip
In Appendix $\ref{appendix:transport}$ we prove that we are able to construct the characteristics $(X_t(x,v),V_t(x,v))$, for all $(t,x,v)$ in $\R^+\times \bar{\Omega}\times \R^d$, of the transport equation (Proposition $\eqref{prop:characteristics}$). Thanks to this Proposition we can define a mild solution of the Boltzmann equation in the cutoff case. This weaker form of solutions is actually the key point for our result and also gives a more general statement.

\bigskip
\begin{defi}\label{def:mildcutoff}
Let $f_0$ be a measurable function, non-negative almost everywhere on $\bar{\Omega}\times \R^d$.
\\ A measurable function $f = f(t, x, v)$ on $[0, T)\times\bar{\Omega}\times \R^d$ is a mild solution of the Boltzmann equation associated to the initial datum $f_0(x, v)$ if 
\begin{enumerate}
\item $f$ is non-negative on $\bar{\Omega}\times \R^d$,
\item for every $(x, v)$ in $\Omega\times \R^d$:
$$t \longmapsto L[f(t,X_t(x,v),\cdot)](V_t(x,v)), \:\: t \longmapsto Q^+[f(t,X_t(x,v),\cdot), f(t,X_t(x,v),\cdot)](V_t(x,v))$$
are in $L^1_{loc}([0,T))$,
\item and for each $t\in [0,T)$, for all $x \in \Omega$ and $v\in \R^d$,
\end{enumerate}
\begin{eqnarray}
&&f(t,X_t(x,v),V_t(x,v)) = f_0(x,v)\emph{\mbox{exp}}\left[-\int_0^t L[f(s,X_s(x,v),\cdot)](V_s(x,v))\:ds\right] \label{mildCO}
\\ &&\quad + \int_0^t \emph{\mbox{exp}}\left(-\int_s^t L[f(s',X_{s'}(x,v),\cdot)](V_{s'}(x,v))\:ds'\right)\nonumber
\\ && \quad\quad\quad\quad Q^+[f(s,X_s(x,v),\cdot), f(s,X_s(x,v),\cdot)](V_s(x,v))\: ds.\nonumber
\end{eqnarray}
\end{defi}

\bigskip
Now we state our result.

\bigskip
\begin{theorem}\label{theo:boundcutoff}
Let $\Omega$ be $\T^d$ or a $C^2$ open convex bounded domain in $\R^d$ with nowhere null normal vector and let $f_0$ be a non-negative continuous function on $\bar{\Omega} \times \R^d$. 
Let $B=\Phi b$ be a collision kernel satisfying $\eqref{assumptionB}$, with $\Phi$ satisfying $\eqref{assumptionPhi}$ or $\eqref{assumptionPhimol}$ and $b$ satisfying $\eqref{assumptionb}$ with $\nu < 0$. Let $f(t,x,v)$ be a mild solution of the Boltzmann equation in $\bar{\Omega}\times \R^d$ on some time interval $[0,T)$, $T \in (0,+\infty]$, which satisfies
\begin{itemize}
\item $f$ is continuous on $[0,T) \times \left(\bar{\Omega} \times \R^d-\Lambda_0\right)$ ($\Lambda_0$ grazing set defined by $\eqref{grazingset}$), $f(0,x,v) = f_0(x,v)$, $M>0$ and $E < \infty$ in $\eqref{conservations}$;
\item if $\Phi$ satisfies $\eqref{assumptionPhi}$ with $\gamma \geq 0$ or if $\Phi$ satisfies $\eqref{assumptionPhimol}$, then $f$ satisfies $\eqref{assumption1}$;
\item if $\Phi$ satisfies $\eqref{assumptionPhi}$ with $\gamma < 0$, then $f$ satisfies $\eqref{assumption1}$ and $\eqref{assumption2}$.
\end{itemize}
Then for all $\tau \in (0,T)$ there exists $\rho >0$ and $\theta > 0$, depending on $\tau$, $E_f$ (and $L^{p_\gamma}_f$ if  $\Phi$ satisfies $\eqref{assumptionPhi}$ with $\gamma < 0$), such that for all $t \in [\tau,T)$ the solution $f$ is bounded from below, almost everywhere, by a global Maxwellian distribution with density $\rho$ and temperature $\theta$, i.e.
$$\forall t \in [\tau,T),\:\forall (x,v) \in \bar{\Omega}\times \R^d, \quad f(t,x,v) \geq \frac{\rho}{(2\pi \theta )^{d/2}}e^{-\frac{\abs{v}^2}{2\theta }}. $$
\end{theorem}
\bigskip

If we add the assumptions of uniform boundedness of $f_0$ and of the mass and entropy of the solution $f$ we can use the arguments originated in \cite{PW} to construct explicitely the initial ``upheaval point'', without any compactness argument (see Section $\ref{subsec:constructiveupheavalpoint}$). Moreover, if we further suppose that $\Omega$ is $C^3$ and strictly convex, the use of tools developed by Guo \cite{Gu6} yields a constructive method to control grazing collisions (see Remark $\ref{rem:constructivegrazing}$). We thus have the following corollary.

\bigskip
\begin{cor}\label{cor:constructivecutoff}
Suppose that conditions of Theorem $\ref{theo:boundcutoff}$ are satisfied (the continuity assumption on $f_0$ can be dropped) and further assume that $\Omega$ is $C^3$ and strictly convex, \mbox{i.e.} there exists $\func{\xi}{\R^d}{\R}$ to be $C^3$ such that 
$$\Omega = \{x\in\R^d,\quad \xi (x)<0\}$$
and such that $\nabla\xi \neq 0$ on $\partial\Omega$ and there exists $C_\xi>0$ such that 
$$\partial_{ij}\xi (x)v_iv_j \geq C_\xi \norm{v}^2$$
for all $x$ in $\bar{\Omega}$ and all $v$ in $\R^d$. Further assume that $f_0$ is uniformly bounded from below
$$\forall (x,v) \in \Omega\times\R^d, \quad f_0(x,v) \geq \varphi(v) > 0,$$
and that $f$ has a bounded  local mass and entropy
\begin{eqnarray*}
R_f &=& \inf\limits_{(t,x)\in [0,T)\times \Omega}\int_{\R^d}f(t,x,v)\:dv >0
\\H_f &=& \sup\limits_{(t,x)\in [0,T)\times \Omega}\abs{\int_{\R^d}f(t,x,v)\emph{\mbox{log}}f(t,x,v)\:dv} <+\infty.
\end{eqnarray*}
Then conclusion of Theorem $\ref{theo:boundcutoff}$ holds true with the constants $\rho$ and $\theta$ being explicitely constructed in terms of $\tau$, $E_f$, $H_f$, $L^{p_\gamma}_f$ and upper and lower bounds on $\abs{\nabla\xi}$  and $\abs{\nabla^2\xi}$on $\partial\Omega$.
\end{cor}

\bigskip
As stated in Subsection $\ref{subsec:strategy}$, the main result to reach Theorem $\ref{theo:boundcutoff}$ is the construction of an immediate lower bound only depending on the norm of the velocity:

\bigskip
\begin{prop}\label{prop:centredball}
Let $f$ be the mild solution of the Boltzmann equation described in Theorem $\ref{theo:boundcutoff}$.
\\For all $0 < \tau < T$ there exists $r_V$, $a_0(\tau)>0$ such that
$$\forall t \in [\tau/2,\tau],\:\forall(x,v) \in \bar{\Omega}\times \R^d, \quad f(t,x,v) \geq a_0(\tau)\mathbf{1}_{B(0,r_V)}(v), $$
$r_V$ and $a_0(\tau)$ only depending on $\tau$, $E_f$ (and $L^{p_\gamma}_f$ if  $\Phi$ satisfies $\eqref{assumptionPhi}$ with $\gamma < 0$).
\end{prop}
\bigskip

 
\subsection{Exponential lower bound for non-cutoff collision kernels}\label{subsec:mainnoncutoff}

In the case of non-cutoff collision kernels ($0 \leq \nu <2$ in $\eqref{assumptionb}$), Grad's splitting does not make sense anymore and so we have to find a new way to define mild solutions to the Boltzmann equation $\eqref{BE}$. The splitting we are going to use is a standard one and it reads
\begin{eqnarray*}
Q(g,h) &=&  \int_{\R^d\times \mathbb{S}^{d-1}}\Phi\left(|v - v_*|\right)b\left( \mbox{cos}\theta\right)\left[h'g'_* - hg_*\right]dv_*d\sigma
\\     &=& Q^1_b(g,h) - Q^2_b(g,h),
\end{eqnarray*}
where we used the following definitions
\begin{eqnarray}
Q^1_b(g,h) &=& \int_{\R^d\times \mathbb{S}^{d-1}}\Phi\left(|v - v_*|\right)b\left( \mbox{cos}\theta\right)g'_*\left(h'-h\right)\: dv_*d\sigma,\nonumber
\\ Q^2_b(g,h) &=&  -\left(\int_{\R^d\times \mathbb{S}^{d-1}}\Phi\left(|v - v_*|\right)b\left( \mbox{cos}\theta\right)\left[g'_*-g_*\right]\:dv_*d\sigma\right)h \label{splittingQ1Q2}
\\ &=& S[g](v)h.\nonumber
\end{eqnarray}

We would like to use the properties we derived in the study of collision kernels with cutoff. Therefore we will consider additional splitting of $Q$.
\par For $\eps$ in $(0,\pi/4)$ we define a cutoff angular collision kernel
$$b^{CO}_\eps\left( \mbox{cos}\theta\right) = b\left( \mbox{cos}\theta\right)\mathbf{1}_{\abs{\theta} \geq \eps}$$
and a non-cutoff one
$$b^{NCO}_\eps\left( \mbox{cos}\theta\right) = b\left( \mbox{cos}\theta\right)\mathbf{1}_{\abs{\theta} \leq \eps}.$$

Considering the two collision kernels $B^{CO}_\eps = \Phi b^{CO}_\eps$ and $B^{NCO}_\eps = \Phi b^{NCO}_\eps$, we can combine Grad's splitting $\eqref{gradsplitting}$ applied to  $B^{CO}_\eps$ with the non-cutoff splitting $\eqref{splittingQ1Q2}$ applied to  $B^{NCO}_\eps$. This yields the splitting we shall use to deal with non-cutoff collision kernels,
\begin{equation}\label{noncutoffsplitting}
Q = Q^+_\eps - Q^-_\eps + Q^1_\eps - Q^2_\eps,
\end{equation}
where we use the shortened notations $Q^{\pm}_\eps = Q^{\pm}_{b^{CO}_\eps}$ and $Q^{i}_\eps = Q^{i}_{b^{NCO}_\eps}$, for $i=1,2$.

\bigskip
Thanks to the splitting $\eqref{noncutoffsplitting}$ and the study of characteristics mentionned in Section $\ref{subsec:maintransport}$, we are able to define mild solutions to the Boltzmann equation with non-cutoff collision kernels. This is obtained by considering the Duhamel formula associated to the splitting $\eqref{noncutoffsplitting}$ along the characteristics (as in the cutoff case).

\bigskip
\begin{defi}\label{def:mildnoncutoff}
Let $f_0$ be a measurable function, non-negative almost everywhere on $\bar{\Omega}\times \R^d$.
\\ A measurable function $f = f(t, x, v)$ on $[0, T)\times\bar{\Omega}\times \R^d$ is a mild solution of the Boltzmann equation with non-cutoff angular collision kernel associated to the initial datum $f_0(x, v)$ if there exists $0 < \eps_0 <\pi/4$ such that for all $0<\eps<\eps_0$:
\begin{enumerate}
\item $f$ is non-negative on $\bar{\Omega}\times \R^d$,
\item for every $(x, v)$ in $\Omega\times \R^d$:
$$t \longmapsto L_\eps[f(t,X_t(x,v),\cdot)](V_t(x,v)), \:\: t \longmapsto Q^+_\eps[f(t,X_t(x,v),\cdot), f(t,X_t(x,v),\cdot)](V_t(x,v))$$
$$t \longmapsto S_\eps[f(t,X_t(x,v),\cdot)](V_t(x,v)), \:\: t \longmapsto Q^1_\eps[f(t,X_t(x,v),\cdot), f(t,X_t(x,v),\cdot)](V_t(x,v))$$
are in $L^1_{loc}([0,T))$,
\item and for each $t\in [0,T)$, for all $x \in \Omega$ and $v\in \R^d$,
\end{enumerate}

\begin{eqnarray}
&&\label{mildNCO}
\\&&f(t,X_t(x,v),V_t(x,v)) = f_0(x,v)\emph{\mbox{exp}}\left[-\int_0^t \left(L_\eps + S_\eps\right)[f(s,X_s(x,v),\cdot)](V_s(x,v))\:ds\right]\nonumber 
\\ &&\quad\quad + \int_0^t \emph{\mbox{exp}}\left(-\int_s^t \left(L_\eps + S_\eps\right)[f(s',X_{s'}(x,v),\cdot)](V_{s'}(x,v))\:ds'\right)\nonumber
\\ && \quad\quad\quad\quad\quad \left(Q^+_\eps + Q^1_\eps\right)[f(s,X_s(x,v),\cdot), f(s,X_s(x,v),\cdot)](V_s(x,v))\: ds.\nonumber
\end{eqnarray}
\end{defi}

\bigskip
Now we state our result.

\bigskip
\begin{theorem}\label{theo:boundnoncutoff}
Let $\Omega$ be $\T^d$ or a $C^2$ open convex bounded domain in $\R^d$ with nowhere null normal vector and $f_0$ be a continuous function on $\bar{\Omega} \times \R^d$. 
Let $B=\Phi b$ be a collision kernel satisfying $\eqref{assumptionB}$, with $\Phi$ satisfying $\eqref{assumptionPhi}$ or $\eqref{assumptionPhimol}$ and $b$ satisfying $\eqref{assumptionb}$ with $\nu$ in $[0,2)$. Let $f(t,x,v)$ be a mild solution of the Boltzmann equation in $\bar{\Omega}\times \R^d$ on some time interval $[0,T)$, $T \in (0,+\infty]$, which satisfies
\begin{itemize}
\item $f$ is continuous on $[0,T) \times \left(\bar{\Omega} \times \R^d-\Lambda_0\right)$ ($\Lambda_0$ grazing set defined by $\eqref{grazingset}$) and $f(0,x,v) = f_0(x,v)$, $M>0$ and $E < \infty$ in $\eqref{conservations}$;
\item if $\Phi$ satisfies $\eqref{assumptionPhi}$ with $\gamma \geq 0$ or if $\Phi$ satisfies $\eqref{assumptionPhimol}$, then $f$ satisfies $\eqref{assumption1}$ and $\eqref{assumption3}$;
\item if $\Phi$ satisfies $\eqref{assumptionPhi}$ with $\gamma < 0$, then $f$ satisfies $\eqref{assumption1}$, $\eqref{assumption2}$ and $\eqref{assumption3}$.
\end{itemize}
Then for all $\tau \in (0,T)$ and for any exponent $K$ such that
$$K > 2\frac{\emph{\mbox{log}}\left(2+\frac{2\nu}{2-\nu}\right)}{\emph{\mbox{log}2}},$$

there exists $C_1, C_2 >0$, depending on $\tau$, $K$, $E_f$, $E'_f$, $W_f$ (and $L^{p_\gamma}_f$ if  $\Phi$ satisfies $\eqref{assumptionPhi}$ with $\gamma < 0$), such that
$$\forall t \in [\tau,T),\:\forall (x,v) \in \bar{\Omega}\times \R^d, \quad f(t,x,v) \geq C_1e^{-C_2\abs{v}^K}. $$
Moreover, in the case $\nu=0$, one can take $K=2$ (Maxwellian lower bound).
\end{theorem}
\bigskip

We emphasize here that, in the same spirit as in the cutoff case, the main part of the proof will rely on the establishment of an equivalent to Proposition $\ref{prop:centredball}$ for non-cutoff collision kernels.

\bigskip
\begin{cor}\label{cor:constructivenoncutoff}
As for Corollary $\ref{cor:constructivecutoff}$, if $f_0$ is bounded uniformly from below as well as the local mass of $f$, the local entropy of $f$ is uniformly bounded from above and $\Omega$ is $C^3$ and strictly convex then the conclusion of Theorem $\ref{theo:boundnoncutoff}$ holds true with constants being explicitely constructed in terms of $\tau$, $K$, $E_f$, $E'_f$, $W_f$, $H_f$, $L^{p_\gamma}_f$ and upper and lower bounds on $\abs{\nabla\xi}$  and $\abs{\nabla^2\xi}$on $\partial\Omega$.
\end{cor}

\bigskip
\begin{remark}
Throughout the paper, we are going to deal with the case where $\Omega$ is a $C^2$ convex bounded domain since it is the case where the most important difficulties arise. However, if $\Omega = \T^d$, we can follow the same proofs by letting the first time of collision with the boundary to be $+\infty$ (see Appendix $\ref{appendix:transport}$) and by making the definition that the distance to the boundary (which does not exist) is $+\infty$ (which rules out the case of grazing trajectories).
\end{remark}
\bigskip

%% file: boundCutoff_1ststep.tex
\section{The cutoff case: localized ``upheaval points''}\label{sec:cutoff_1ststep}

In this section and the next three we are going to prove a Maxwellian lower bound for a solution to the Boltzmann equation $\eqref{BE}$ in the case where the collision kernel satisfies a cutoff property.
\par The strategy to tackle this result follows the main idea used in \cite{Mo2} and \cite{PW} which relies on finding an ``upheaval point'' (a first minoration uniform in time and space but localised in velocity) and spreading this bound, thanks to the spreading property of the $Q^+$ operator, in order to include larger and larger velocities.

\bigskip
We gather here two lemmas, proven in \cite{Mo2}, that we will frequently use in this section. We remind the reader that we are using Grad's splitting $\eqref{gradsplitting}$. Let us first give an $L^\infty$ bound on the loss term (Corollary $2.2$ in \cite{Mo2}).

\bigskip
\begin{lemma}\label{lem:L}
Let $g$ be a measurable function on $\R^d$. Then
$$\forall v \in \R^d,\quad \abs{L[g](v)} \leq C_g^L\langle v \rangle^{\gamma^+},$$
where $C_g^L$ is defined by:
\begin{enumerate}
\item If $\Phi$ satisfies $\eqref{assumptionPhi}$ with $\gamma \geq 0$ or if $\Phi$ satisfies $\eqref{assumptionPhimol}$, then
$$C^L_g = \emph{\mbox{cst}}\: n_b C_\Phi e_g.$$
\item If $\Phi$ satisfies $\eqref{assumptionPhi}$ with $\gamma \in (-d,0)$, then
$$C^L_g = \emph{\mbox{cst}}\: n_b C_\Phi \left[e_g+ l^p_g\right],\quad p > d/(d+\gamma).$$
\end{enumerate}
\end{lemma}
\bigskip

The spreading property of $Q^+$ is given by the following lemma (Lemma $2.4$ in \cite{Mo2}), where we define
\begin{equation}\label{lb}
l_b = \inf_{\pi/4 \leq \theta \leq 3\pi/4}b\left(\mbox{cos}\:\theta\right).
\end{equation}

\bigskip
\begin{lemma}\label{lem:Q+spread}
Let $B=\Phi b$ be a collision kernel satisfying $\eqref{assumptionB}$, with $\Phi$ satisfying $\eqref{assumptionPhi}$ or $\eqref{assumptionPhimol}$ and $b$ satisfying $\eqref{assumptionb}$ with $\nu \leq 0$. Then for any $\bar{v} \in \R^d$, $0<r \leq R$, $\xi \in (0,1)$, we have

$$Q^+(\mathbf{1}_{B(\bar{v},R)},\mathbf{1}_{B(\bar{v},r)}) \geq \emph{\mbox{cst}}\: l_b c_\Phi r^{d-3} R^{3+\gamma} \xi^{\frac{d}{2}-1}\mathbf{1}_{B\left(\bar{v},\sqrt{r^2 + R^2}(1-\xi)\right)}.$$

As a consequence in the particular quadratic case $\delta = r = R$, we obtain

$$Q^+(\mathbf{1}_{B(\bar{v},\delta)},\mathbf{1}_{B(\bar{v},\delta)}) \geq \emph{\mbox{cst}}\: l_b c_\Phi \delta^{d+\gamma} \xi^{\frac{d}{2}-1}\mathbf{1}_{B\left(\bar{v},\delta\sqrt{2}(1-\xi)\right)},$$
for any $\bar{v} \in \R^d$ and $\xi \in (0,1)$.
\end{lemma}
\bigskip

The case of the torus, studied in \cite{Mo2}, indicates that without rebounding the expected minoration is created after time $t=0$ as quickly as one wants. Therefore we expect  the same kind of bound to arise on each characteristic trajectory before its first rebound. However, in the case of a bounded domain, rebounds against the boundary can occur very close to the time $t=0$ and a rebound preserves only the norm of the velocity. Therefore, we will fail finding a uniformly (in space) small time where a uniform bound arises. Nevertheless, the convexity  and the smoothness of the domain implies that grazing collisions against the boundary do not change the velocity very much. 
\par Thus our study will be split in three parts, which are the next three sections. The first step will be to partition the position and velocity spaces so that we have an immediate appearance of an ``upheaval point" in each of those partitions. The second one is to obtain a uniform lower bound which will depend only on the norm of the velocity. Then the final part will use the standard spreading method used in \cite{Mo2} and \cite{PW} which will allow us to deal with large velocities and derive the exponential lower bound uniformly.


\subsection{Partition of the phase space and first localised lower bounds}\label{subsec:upheaval}

In this section we use the continuity of $f$ together with the conservation laws $\eqref{conservations}$ to obtain a point in the phase space where $f$ is strictly positive. Then, thanks to the continuity of $f$, its Duhamel representation $\eqref{mildCO}$ and the spreading property of the $Q^+$ operator (Lemma $\ref{lem:Q+spread}$) we extend this positivity to high velocities at that particular point (Lemma $\ref{lem:positivity1}$). Finally, the free transport part of the solution $f$ will imply the immediate appearance of the localised lower bounds (Proposition $\ref{prop:upheaval}$).
\par Moreover we define constants that we will use in the next two subsections in order to have a uniform lower bound. 

\bigskip
We define some shorthand notations. For $x$ in $\bar{\Omega}$, $v$ in $\R^d$ and $s,t \geq 0$ we denote the point at time $s$ of the forward characteristic passing through $(x,v)$ at time $t$ by
$$\begin{array}{rl} \displaystyle{X_{s,t}(x,v) }&\displaystyle{= X_s(X_t(x,-v),-V_t(x,-v))} \vspace{2mm} \\ \vspace{2mm} \displaystyle{V_{s,t}(x,v) }&\displaystyle{= V_s(X_t(x,-v),-V_t(x,-v)),}\end{array}$$
which has been derived from $\eqref{bijectioncharacteristics}$.

\bigskip
We start by the strict positivity of our function at one point for all velocities:

\bigskip
\begin{lemma}\label{lem:positivity1}
Let $f$ be the mild solution of the Boltzmann equation described in Theorem $\ref{theo:boundcutoff}$.
\\ Then there exists $(x_1,v_1)$ in $\Omega \times \R^d$ and $\Delta >0$ such that for all $n \in \N$ and all $t$ in $[0,\Delta]$, there exists $r_n >0$, depending only on $n$, and $\alpha_n(t)>0$ such that
$$\forall x \in B\left(x_1,\frac{\Delta}{2^n}\right), \:\forall v \in \R^d, \quad f(t,x,v) \geq \alpha_n(t)\mathbf{1}_{B(v_1,r_n)}(v),$$
with $\alpha_0 >0$ independent of $t$ and the induction formula
$$\alpha_{n+1}(t) = C_Q\frac{r_n^{d+\gamma}}{4^{d/2-1}}\int_0^{\min\left(t,\Delta/(2^{n+1}(2r_n+\norm{v_1})\right)}e^{-s C_L  \langle 2 r_n+\norm{v_1}\rangle^{\gamma^+}}\alpha^2_n(s)\:ds$$
where $C_Q=cst\: l_b c_\Phi$ is defined in Lemma $\ref{lem:Q+spread}$ and $C_L = cst\: n_b C_\Phi E_f$ (or $C_L = cst\: n_b C_\Phi (E_f+L^p_f)$) is defined in Lemma $\ref{lem:L}$, and
$$r_0 = \Delta, \quad r_{n+1} = \frac{3\sqrt{2}}{4}r_n.$$
\end{lemma}
\bigskip

\bigskip
\begin{proof}[Proof of Lemma $\ref{lem:positivity1}$]
The proof is an induction on $n$.

\bigskip
\textbf{Step $1$: Initialization}. We recall the conservation laws satisfied by a solution to the Boltzmann equation, $\eqref{conservations}$,
$$\forall t \in \R^+, \quad \int_{\Omega}\int_{\R^d} f(t,x,v)\:dxdv = M, \quad \int_{\Omega}\int_{\R^d} \abs{v}^2f(t,x,v)\:dxdv = E,$$
with $M>0$ and $E<\infty$.
\\ Since $\Omega$ is bounded, and so is included in, say, $B(0,R_X)$, we also have that
$$\forall t \in \R^+,\quad \int_{\Omega}\int_{\R^d} \left(\abs{x}^2+ \abs{v}^2\right) f(t,x,v)\:dxdv \leq  \alpha = M R_X^2 + E < +\infty.$$
Therefore if we take $t =0$ and $R_{min} = \sqrt{2\alpha/M}$, we have the following

$$\int_{B(0,R_{min})}\int_{B(0,R_{min})}f_0(x,v)\:dxdv \geq \frac{M}{2} > 0.$$

Therefore we have that there exists $x_1$ in $\Omega$ and $v_1$ in $B(0,R_{min})$ such that
$$f_0(x_1,v_1) \geq \frac{\mbox{M}}{4\mbox{Vol}(B(0,R_{min}))^2} >0.$$

The first step of the induction is then due to the continuity of $f$ at $(0,x_1,v_1)$. Indeed, there exists $\delta_T,\delta_X,\delta_V >0$ such that

$$\forall t \in [0,\delta_T], \:\forall x \in B(x_1,\delta_X), \: \forall v \in B(v_1,\delta_V), \quad f(t,x,v) \geq \frac{\mbox{M}}{8\mbox{Vol}(B(0,R_{min}))^2}.$$
and we define $\Delta = \min(\delta_T,\delta_X,\delta_V)$.

\bigskip
\textbf{Step $2$: Proof of the induction}. We assume the conjecture is valid for $n$.
\\ Let $x$ be in $B(x_1,\Delta/2^{n+1})$, $v$ in $B(0,\norm{v_1} + 2r_n)$ and $t$ in $[0,\Delta]$.

\bigskip
We use the fact that $f$ is a mild solution to write $f(t,x,v)$ under its Duhamel form $\eqref{mildCO}$. The control we have on the $L$ operator, Lemma $\ref{lem:L}$, allows us to bound from above the second integral term (the first term is positive).  Moreover, this bound on $L$ is independent on $t$, $x$ and $v$ since it only depends on an upper bound on the energy $e_{f(t,x,\cdot)}$ (and its local $L^p$ norm $l^p_{f(t,x,\cdot)}$) which is uniformly bounded by $E_f$ (and by $L^p_f$). This yields, for $\tau_n(t) = \min\left(t,\Delta/(2^{n+1}(2r_n+\norm{v_1}))\right)$

\begin{equation}\label{inductionpositivity}
f(t,x,v) \geq  \int_{0}^{\tau_n(t)}e^{-s C_L \langle \norm{v_1}+ 2r_n \rangle^{\gamma^+}}  Q^+ \left[f(s,X_{s,t}(x,v),\cdot),f(s,X_{s,t}(x,v),\cdot)\right]\left(V_{s,t}(x,v)\right)\:ds,
\end{equation}
where $C_L = cst\: n_b C_\Phi E_f$ (or $C_L = cst\: n_b C_\Phi (E_f+L^p_f)$), see Lemma $\ref{lem:L}$, and we used $\norm{V_{s,t}(x,v)} = \norm{v} \leq 2r_n + \norm{v_1}$.

\bigskip
Besides, we have that $B(x_1,\Delta) \subset \Omega$ and also
$$\forall s \in \left[0,\frac{\Delta}{2^{n+1}(2r_n+\norm{v_1})}\right], \:\forall v_* \in B(0,\norm{v_1} + 2r_n), \quad \norm{x_1- (x + sv_*)} \leq \frac{\Delta}{2^{n}}$$
which, by definition of the characteristics (see Appendix $\ref{subsec:characteristics}$), yields
$$\forall s \in \left[0,\tau_n(t)\right],\:\forall v_* \in B(0,\norm{v_1}+2r_n), \quad \left\{ \begin{array}{rl} \displaystyle{X_{s,t}(x,v_*) }&\displaystyle{= x + sv_* \in B\left(x_1,\frac{\Delta}{2^{n}}\right)} \vspace{2mm} \\\vspace{2mm} \displaystyle{V_{s,t}(x,v_*)}&\displaystyle{= v_*.}\end{array}\right.$$

\bigskip
Therefore, by calling $v_*$ the integration parametre in the operator $Q^+$ we can apply the induction property to $f(s,X_{s,t}(x,v),v_*)$ which implies, in $\eqref{inductionpositivity}$,

$$f(t,x,v) \geq  \int_{0}^{\tau_n(t)}e^{-s C_L \langle \norm{v_1}+2r_n \rangle^{\gamma^+}}\alpha_n^2(s)  Q^+\left[\mathbf{1}_{B(v_1,r_n)},\mathbf{1}_{B(v_1,r_n)}\right]\:ds (v).$$

\bigskip
Applying the spreading property of $Q^+$, Lemma $\ref{lem:Q+spread}$, with $\xi = 1/4$ gives us the expected result for the step $n+1$ since $B(v_1,r_{n+1}) \subset B(0,\norm{v_1}+2r_n)$.
\end{proof}
\bigskip

We now have all the tools to prove the next proposition which is the immediate appearance of localised ``upheaval points".

\bigskip
\begin{prop}\label{prop:upheaval}
Let $f$ be the mild solution of the Boltzmann equation described in Theorem $\ref{theo:boundcutoff}$.
\\Then there exists $\Delta >0$ such that for all $0<\tau_0 \leq \Delta$, there exists $\delta_T(\tau_0)$, $\delta_X(\tau_0)$, $\delta_V(\tau_0)$, $R_{min}(\tau_0)$, $a_0(\tau_0) > 0$ such that for all $N$ in $\N$ there exists $N_X$ in $\N^*$ and $x_1,\dots,x_{N_X}$ in $\Omega$ and $v_1,\dots,v_{N_X}$ in $B(0,R_{min}(\tau_0))$ and
\begin{itemize}
\item $\bar{\Omega} \subset \bigcup\limits_{1 \leq i \leq N_X}B\left(x_i,\delta_X(\tau_0)/2^{N}\right)$;
\item $\forall t \in [\tau_0,\delta_T(\tau_0)], \:\forall x \in B(x_i,\delta_X(\tau_0)),\forall v \in \R^d,$ 
$$f(t,x,v) \geq a_0(\tau_0)\mathbf{1}_{B\left(v_i,\delta_V(\tau_0)\right)}(v),$$
with $B\left(v_i,\delta_V(\tau_0)\right) \subset B(0,R_{min}(\tau_0))$.
\end{itemize}
\end{prop}
\bigskip

\bigskip
\begin{proof}[Proof of Proposition $\ref{prop:upheaval}$]
We are going to use the free transport part of the Duhamel form of $f$ $\eqref{mildCO}$, to create localised lower bounds out of Lemma $\ref{lem:positivity1}$.

\bigskip
We take $0< \tau_0 \leq \Delta$, where $\Delta$ is defined in Lemma $\ref{lem:positivity1}$.
\\$\Omega$ is bounded so let us denote its diameter by $d_{\Omega}$. Let $n$ be big enough such that $r_n \geq 2 d_{\Omega}/\tau_0+\norm{v_1}$ and define $R_{min}(\tau_0) = 2 d_{\Omega}/\tau_0$.
\par Thanks to Lemma $\ref{lem:positivity1}$ applied to this particular $n$ we have that
\begin{equation}\label{initialbound}
\forall t \in \left[\frac{\tau_0}{2},\Delta\right],\:\forall x \in B(x_1,\Delta/2^n), \quad f(t,x,v) \geq \alpha_n\left(\frac{\tau_0}{2}\right)\mathbf{1}_{B(v_1,r_n)}(v),
\end{equation}
where we used the fact that $\alpha_n(t)$ is an increasing function.

\bigskip
Define
$$a_0(\tau_0) = \frac{1}{2}\alpha_n\left(\frac{\tau_0}{2}\right) e^{-\frac{\tau_0}{2} C_L \langle \frac{2 d_{\Omega}}{\tau_0} \rangle^{\gamma^+}}.$$

\bigskip
\textbf{Definition of the constants.}
We notice that for all $x$ in $\partial\Omega$ we have that $n(x)\cdot (x-x_1) >0$, because $\Omega$ has nowhere null normal vector by hypothesis. But the function 
$$x \longmapsto n(x)\cdot \frac{x-x_1}{\norm{x-x_1}}$$
 is continuous (since $\Omega$ is $C^2$) on the compact $\partial{\Omega}$ and therefore has a minimum that is atteined at a certain $X(x_1)$ on $\partial\Omega$.
\par Hence, 
\begin{equation}\label{lambdax1}
\forall x \in \partial\Omega, \quad n(x)\cdot \frac{x-x_1}{\norm{x-x_1}} \geq n(X(x_1))\cdot \frac{X(x_1)-x_1}{\norm{X(x_1)-x_1}}= 2 \lambda(x_1) >0.\end{equation}

\bigskip
To shorten following notations, we define on $\bar{\Omega} \times \left(\R^d-\{0\}\right)$ the function
\begin{equation}\label{Phi}
\Phi(x,v) = n\left(x+t\left(x,\frac{v}{\norm{v}}\right)\frac{v}{\norm{v}}\right),
\end{equation}
where we defined $t(x,v) = \min\{t\geq 0: x+tv \in \partial\Omega\}$, the first time of contact against the boundary of the forward characteristic $(x+sv)_{s\geq 0}$ defined for $v\neq 0$ and continuous on $\bar{\Omega}\times\left(\R^d-\{0\}\right)$ (see Lemma $5.2$).

\bigskip
We denote $d_1$ to be half of the distance from $x_1$ to $\partial\Omega$. We define two sets included in $[0,\Delta]\times\bar{\Omega}\times\R^d$:
$$\Lambda^{(1)} = [0,\Delta]\times B(x_1,d_1) \times \R^d$$
and
$$\Lambda^{(2)} = \left\{(t,x,v)\notin \Lambda^{(1)},\quad \norm{v}\geq \frac{d_1}{\tau_0} \quad\mbox{and} \quad \Phi(x,v)\cdot \frac{v}{\norm{v}} \geq \lambda(x_1)\right\}$$

\bigskip
By continuity of $t(x,v)$ and of $n$ (on $\partial\Omega$), we have that
$$\Lambda = \Lambda^{(1)} \cap \Lambda^{(2)}$$
is compact and does not intersect the grazing set $[0,\Delta]\times \Lambda_0$ defined by $\eqref{grazingset}$. Therefore, $f$ is continuous in $\Lambda$ and thus is uniformly continuous on $\Lambda$. Hence, there exist $\delta_T'(\tau_0)$, $\delta'_X(\tau_0)$, $\delta'_V(\tau_0) >0$ such that
$$\forall (t,x,v),\:(t',x',v')\in\Lambda,\quad |t-t'|\leq \delta'_T(\tau_0),\: \norm{x-x'}\leq \delta'_X(\tau_0),\: \norm{v-v'}\leq \delta'_V(\tau_0),$$
\begin{equation}\label{uniformcontinuity}
\abs{f(t,x,v)- f(t',x',v')} \leq a_0(\tau_0).
\end{equation}

\bigskip
The map $\Phi$ (defined by $\eqref{Phi}$) is uniformly continuous on the compact $[0,\Delta]\times\bar{\Omega}\times \mathbb{S}^{d-1}$ and therefore there exist $\delta_T''(\tau_0)$, $\delta''_X(\tau_0)$, $\delta''_V(\tau_0) >0$ such that
$$\forall (t,x,v),\:(t',x',v')\in \Lambda^{(2)},\quad |t-t'|\leq \delta''_T(\tau_0),\: \norm{x-x'}\leq \delta''_X(\tau_0),\: \norm{v-v'}\leq \delta''_V(\tau_0),$$
\begin{equation}\label{uniformlambda2}
\abs{ \Phi(x,v)-  \Phi(x',v')} \leq \frac{\lambda(x_1)}{2}.
\end{equation}

\bigskip
We conclude our definitions by taking 
\begin{eqnarray*}
\delta_T(\tau_0)&=& \min\left(\Delta,\: \tau_0 +\delta_T'(\tau_0),\: \tau_0+\delta_T''(\tau_0)\right),
\\ \delta_X(\tau_0)&=& \min\left(\frac{\Delta}{2^n},\:\delta_X'(\tau_0),\:\delta_X''(\tau_0),\:d_1/2\right),
\\ \delta_V(\tau_0) &=&\min\left(r_n,\:\delta_V'(\tau_0),\:\frac{d_1}{2\tau_0}\delta_V''(\tau_0),\:\frac{\lambda(x_1)}{2}\right).
\end{eqnarray*}

\bigskip
\textbf{Proof of the lower bounds.}
We take $N \in \N$ and notice that $\bar{\Omega}$ is compact and therefore there exists $x_{1},\dots,x_{N_{X}}$ in $\Omega$ such that $\bar{\Omega} \subset \bigcup\limits_{1 \leq i \leq N_{X}}B\left(x_i,\delta_X(\tau_0)/2^{N}\right)$. Moreover, we construct them such that $x_{1}$ is the one defined in Lemma $\ref{lem:positivity1}$ and we then take $v_1$ to be the one defined in Lemma $\ref{lem:positivity1}$. We define 
$$\forall i \in \{2,\dots,N_X\}, \quad v_i =\frac{2}{\tau_0}(x_i - x_1).$$
Because $\Omega$ is convex we have that 
\begin{eqnarray*}
X_{\tau_0/2,\tau_0}(x_i,v_i) &=& x_1,
\\ V_{\tau_0/2,\tau_0}(x_i,v_i) &=& v_i.
\end{eqnarray*}

\bigskip
Using the fact that $f$ is a mild solution of the Boltzmann equation, we write it under its Duhamel form $\eqref{mildCO}$ and we drop the last term which is positive. As in the proof of Lemma $\ref{lem:positivity1}$ we can control the $L$ operator appearing in the first term in the right-hand side of $\eqref{mildCO}$ (corresponding to the free transport). Thus, we use the Duhamel form $\eqref{mildCO}$ between $\tau_0$ and $\tau_0/2$. This yields

\begin{eqnarray*}
f(\tau_0,x_i,v_i) &\geq& f\left(\frac{\tau_0}{2},x_1,v_i\right)e^{-\frac{\tau_0}{2}C_L \langle \frac{2}{\tau_0}(x_i-x_1)\rangle^{\gamma^+}}
\\ &\geq& \alpha_n\left(\frac{\tau_0}{2}\right)e^{-\frac{\tau_0}{2}C_L \langle \frac{2d_{\Omega}}{\tau_0}\rangle^{\gamma^+}}\mathbf{1}_{B(v_1,r_n)}(v_i)
\\ &\geq& 2 a_0(\tau_0)\mathbf{1}_{B(v_1,r_n)}(v_i),
\end{eqnarray*}
where we used $\eqref{initialbound}$ for the second inequality. We see here that $v_i$ belongs to $B(0,R_{min}(\tau_0))$ and that $B(0,R_{min}(\tau_0)) \subset B(v_1,r_n)$ and therefore

\begin{equation}\label{finpositivity}
f(\tau_0,x_i,v_i) \geq 2a_0(\tau_0).
\end{equation}

\bigskip
We first notice that $(\tau_0,x_i,v_i)$ belongs to $\Lambda$ since either $x_i$ belongs to $B(x_1,d_1)$ or $\norm{x_1-x_i} \geq d_1$ but by definition of $v_i$ and $\lambda(x_1)$ (see $\eqref{lambdax1}$),
$$n\left(x_i+t\left(x_i,\frac{v_i}{\norm{v_i}}\right)\frac{v_i}{\norm{v_i}}\right)\cdot \frac{v_i}{\norm{v_i}} \geq 2\lambda(x_1)$$
and
$$\norm{v_i} = \frac{2}{\tau_0}\norm{x_i-x_1} \geq \frac{2}{\tau_0}d_1.$$

\bigskip
We take $t$ in $[\tau_0,\delta_T(\tau_0)]$, $x$ in $B(x_i,\delta_X(\tau_0))$ and $v$ in $B(v_i,\delta_V(\tau_0))$ and we will prove that $(t,x,v)$ also belongs to $\Lambda$.
\par If $x_i$ belongs to $B(x_1,d_1/2)$ then since $\delta_X(\tau_0)\leq d_1/2$,
$$\norm{x-x_1} \leq \frac{d_1}{2} + \norm{x-x_i} \leq d_1$$
and $(t,x,v)$ thus belongs to $\Lambda^{(1)} \subset \Lambda$.
\par In the other case where $\norm{x_1-x_i}\geq d_1/2$ we first have that 
$$\norm{v_i} = \frac{2}{\tau_0}\norm{x_i-x_1} \geq \frac{d_1}{\tau_0}.$$
And also
$$\norm{\frac{v}{\norm{v}}-\frac{v_i}{\norm{v_i}}}\leq \frac{2}{\norm{v_i}}\norm{v-v_i}= \frac{\tau_0}{\norm{x_i-x_1}}\delta_V(\tau_0)\leq \frac{2\tau_0}{d_1}\delta_V(\tau_0)\leq \delta_V''(\tau_0).$$
The latter inequality combined with $\eqref{uniformlambda2}$ and that $\abs{t-\tau_0}\leq \delta_T''(\tau_0)$ and $\norm{x-x_i}\leq \delta''_X(\tau_0)$ yields
$$\abs{\Phi(x,v)- \Phi(x_i,v_i)} \leq \frac{\lambda(x_1)}{2},$$
which in turn implies
\begin{eqnarray*}
\Phi(x,v) \cdot \frac{v}{\norm{v}} &\geq& \Phi(x_i,v_i)\cdot\frac{v_i}{\norm{v_i}} + \Phi(n,v)\cdot(v-v_i) + \left(\Phi(x,v)-\Phi(x_i,v_i)\right)\cdot v_i
\\&\geq& 2\lambda(x_1)-\norm{v-v_i} - \abs{\Phi(x,v)-\Phi(x_i,v_i)}
\\&\geq& \lambda(x_1),
\end{eqnarray*}
so that $(t,x,v)$ belongs to $\Lambda^{(2)}$.

\bigskip
We can now conclude the proof.
\par We proved that $(\tau_0,x_i,v_i)$ belongs to $\Lambda$ and that for all  $t$ in $[\tau_0,\delta_T(\tau_0)]$, $x$ in $B(x_i,\delta_X(\tau_0))$ and $v$ in $B(v_i,\delta_V(\tau_0))$, $(t,x,v)$ belongs to $\Lambda$. By definition of the constants, $(t-\tau_0,x-x_i,v-v_i)$ satisfies the inequality of the uniform continuity of $f$ on $\Lambda$ $\eqref{uniformcontinuity}$. Combining this inequality with $\eqref{finpositivity}$, the lower bound at $(\tau_0,x_i,v_i)$, we have that 
$$f(t,x,v) \geq a_0(\tau_0).$$
\end{proof}
\bigskip

\begin{remark}
This last  proposition tells us that localised lower bounds appear immediately, that is to say after any time $\tau_0 >0$. The exponential lower bound we expect will appear immediately after those initial localised lower bounds, i.e. for all $\tau_1>\tau_0$. Therefore, to shorten notation and lighten our  presentation, we are going to study the case of solution to the Boltzmann equation which satisfies Proposition $\ref{prop:upheaval}$ at $\tau_0=0$. Then we will immediatly create the exponential lower bound after $0$ and apply this result to $F(t,x,v) = f(t+\tau_0,x,v)$.
\end{remark}


\subsection{A constructive approach to the initial lower bound, Corollary $\ref{cor:constructivecutoff}$}\label{subsec:constructiveupheavalpoint}

The initial lower bounds we just derived relies on compactness arguments and their construction is therefore not explicit. However, as mentioned in Section $\ref{subsec:maincutoff}$, a few more assumptions on $f_0$ and $f$ suffice to obatin a completely constructive approach for the ``upheaval point". This method is based on a property of the iterated $Q^+$ operator discovered by Pulvirenty and Wennberg \cite{PW} and reformulated by Mouhot (\cite{Mo2} Lemma $2.3$) as follows.

\bigskip
\begin{lemma}\label{lem:Q+Q+}
Let $B=\Phi b$ be a collision kernel satisfying $\eqref{assumptionB}$, with $\Phi$ satisfying $\eqref{assumptionPhi}$ or $\eqref{assumptionPhimol}$ and $b$ satisfying $\eqref{assumptionb}$ with $\nu \leq 0$. Let $g(v)$ be a nonnegative function on $\R^d$ with bounded energy $e_g$ and entropy $h_g$ and a mass $\rho_g$ such that $0 < \rho_g < +\infty$. Then there exist $R_0 \:, \delta_0\:, \eta_0 >0$ and $\bar{v} \in B(0,R_0)$ such that
$$Q^+\left(Q^+\left(g\mathbf{1}_{B(0,R_0)},g\mathbf{1}_{B(0,R_0)}\right),g\mathbf{1}_{B(0,R_0)}\right) \geq \eta_0 \mathbf{1}_{B(\bar{v},\delta_0)},$$
with $R_0 \:, \delta_0\:, \eta_0$ being constructive in terms on $\rho_g$, $e_g$ and $h_g$.
\end{lemma}
\bigskip

We now suppose that $0<\rho_{f_0}<+\infty$, $h_{f_0}<+\infty$ and that
$$\forall (x,v) \in \Omega\times\R^d, \quad f_0(x,v) \geq \varphi(v) > 0$$
and we follow the argument used in\cite{Mo2}.

\bigskip
By the Duhamel definition $\eqref{mildCO}$ of $f$ being a mild solution and Lemma $\ref{lem:L}$ we have
\begin{equation}\label{construct1}
f(t,X_t(x,v),V_t(x,v)) \geq f_0(x,v)e^{-tC_L\langle v \rangle^{\gamma^+}}
\end{equation}
and
$$f(t,x,v) \geq  \int_{0}^{t}e^{-(t-s) C_L \langle v \rangle^{\gamma^+}}  Q^+ \left[f(s,X_{s,t}(x,v),\cdot),f(s,X_{s,t}(x,v),\cdot)\right]\left(V_{s,t}(x,v)\right)\:ds.$$
Define $t(x,v)>0$ the time of first contact with $\partial\Omega$ of  the trajectory $x+sv$ (see rigorous definition in Proposition $\ref{prop:tmin}$). For all $t$ in $[0,t(x,v)]$ we have 
\begin{eqnarray*}
X_{0,t}(x,v) &=& x+ tv,
\\V_{0,t}(x,v) &=& v.
\end{eqnarray*}
Thus, for all $0\leq t \leq t(x,v)$,
$$f(t,x,v) \geq  \int_{0}^{t}e^{-(t-s) C_L \langle v \rangle^{\gamma^+}}  Q^+ \left[f(s,x+sv,\cdot),f(s,x+sv,\cdot)\right]\left(v\right)\:ds,$$
and we can iterate the latter inequality
\begin{equation}\label{construct2}
\begin{split}
&f(t,x,v)\geq \int_{0}^{t}e^{-(t-s) C_L \langle v \rangle^{\gamma^+}}
\\& Q^+\left[ \int_{0}^se^{-(s-s') C_L \langle v \rangle^{\gamma^+}}Q^+\left(f(s,x+s'v,\cdot),f(s,x+s'v,\cdot)\right)(\cdot)ds',f(t,x+sv,\cdot)\right]\left(v\right)ds.
\end{split}
\end{equation}

$\eqref{construct1}$ and $\eqref{construct2}$ are exactly the same bounds than the ones obtained in \cite{Mo2}, Step $1$ of proof of Proposition $3.2$, and we can therefore conclude the same way with Lemma $\ref{lem:Q+Q+}$
$$f(t,x,v) \geq a_0(\tau_0)\mathbf{1}_{B(\bar{v},\delta_0),}$$
as long as $v$ is in $B(0,R_0)$ and $0\leq t \leq \tau_0$. 
\par The only difference with \cite{Mo2} is the fact that we need $\tau_0$ to be in $[0,t(x,v)]$, giving local lower bounds instead of a global one.


\subsection{A lower bound depending only on the norm of the velocity: strategy of the proof of Proposition $\ref{prop:centredball}$}\label{subsec:centredball}

As stated in the introduction, the spreading property of the bilinear operator $Q^+$ cannot be used (at least uniformly in time and space) when we are really close to the boundary due to the lack of control over the rebounds. However, if we have a lower bound depending only on the norm of the velocity then the latter bound will not take into account rebounds as they preserve the norm, allowing us to spread this minoration up to an exponential one.

\bigskip
The next two sections are dedicated to the creation of such a uniform lower bound depending solely on the norm of the velocity. In order to do so we restrain the problem without taking into account large velocities and divide the study to two cases:  if the trajectory stays close to the boundary or if it does not. In both cases we will start from the localised ``upheaval points" constructed in Section $\ref{subsec:upheaval}$ and spread them to the point where one gets a lower bound depending only on the norm of the velocity. 
\par The next sections tackle each of these points. We first study the case when a characteristic reaches a point far from the boundary and finally we focus on the case of grazing characteristics. We fix $\delta_T$, $\delta_X$, $\delta_V$, $R_{min}$ and $a_0$ to be the ones described in Proposition $\ref{prop:upheaval}$ at time $\tau_0 =0$.
\par The result we will derive out of those studies is Proposition $\ref{prop:centredball}$ and from now on, dependencies on physical observables of $f$ ($E_f$ and $L^{p_\gamma}_f$) will be mentionned but will not be explicitly written everytime.

%% file: boundCutoff_far.tex
\section{The cutoff case: characteristics passing by a point far from the boundary}\label{sec:cutoff_far}

In this section we manage to spread the lower bounds created in Proposition $\ref{prop:upheaval}$ up to a ball in velocity centred at zero as long as the trajectory we look at reaches a point far enough from the boundary.

\bigskip
First, we pick $N$ in $\N^*$ and cover $\bar{\Omega}$ with $\bigcup_{1 \leq i \leq N_X}B(x_i,\delta_X/2^{N})$ as in Proposition $\ref{prop:upheaval}$.
\\Then for $l\geq 0$ we define

\begin{equation}\label{Ul}
\Omega_l = \left\{x \in \Omega:\:d(x,\partial\Omega)\geq l\right\},
\end{equation}
where $d(x,\partial\Omega)$ is the distance from $x$ to the boundary of $\Omega$.
\par For any $R>0$ we define two sequences in $\R^+$ by induction, for all $\tau \geq 0$ and $l \geq 0$,

\begin{eqnarray}
&&\left\{\begin{array}{rl} \displaystyle{r_0} &\displaystyle{= \delta_V} \vspace{2mm}\\\vspace{2mm} \displaystyle{r_{n+1} }&\displaystyle{= \frac{3\sqrt{2}}{4}r_n} \end{array}\right. \label{rn}
\\\mbox{and} \nonumber
\\&&\left\{\begin{array}{rl} \displaystyle{a_0(l,\tau) }&\displaystyle{= a_0} \vspace{2mm}\\\vspace{2mm} \displaystyle{a_{n+1}(l,\tau) }&\displaystyle{= C_Q\frac{r_n^{d+\gamma}}{4^{d/2-1}}\frac{l}{2^{n+3}R}e^{-\tau C_L  \langle R\rangle^{\gamma^+}}a^2_n\left(\frac{l}{8},\tau\right),}\end{array}\right.\label{an}
\end{eqnarray}
where $C_Q$ and $C_L$ were defined in Lemma $\ref{lem:positivity1}$.
\par We express the spreading of the lower bound in the following proposition.

\bigskip
\begin{prop}\label{prop:spreadfar}
Let $f$ be the mild solution of the Boltzmann equation described in Theorem $\ref{theo:boundcutoff}$ and suppose that $f$ satisfies Proposition $\ref{prop:upheaval}$ with $\tau_0=0$.
\\Consider $0<\tau\leq \delta_T$ and $N$ in $\N$. Let $(x_i)_{i \in \{1,\dots,N_X\}}$ and $(v_i)_{i \in \{1,\dots,N_X\}}$ be given as in Proposition $\ref{prop:upheaval}$ with $\tau_0=0$.
\\Then for all $n$ in $\{0,\dots,N\}$ we have that the following holds: for all $0<l\leq \delta_X$, and $R>0$ such that $l/R < \tau$, for all $t$ in $[l/(2^n R),\tau]$, and for all $x \in \bar{\Omega}$ and $v \in B(0,R)$, if there exists $t_1 \in [0,t-l/(2^n R)]$ such that $X_{t_1,t}(x,v)$ belongs to $\Omega_l \cap B(x_i,\delta_X/2^n)$ then

$$f(t,x,v) \geq a_n(l,\tau) \mathbf{1}_{B(v_i,r_n)}(V_{t_1,t}(x,v)),$$
where $(r_n)$ and $(a_n)$ are defined by $\eqref{rn}$-$\eqref{an}$.
\end{prop}
\bigskip

\bigskip
\begin{proof}[Proof of Proposition $\ref{prop:spreadfar}$]
This Proposition will be proved by induction on $n$.

\bigskip
\textbf{Step $1$: Initialization}. The initialisation is simply Proposition $\ref{prop:upheaval}$.

\bigskip
Indeed, we use the definition of $f$ being a mild solution to write $f(t,x,v)$ under its Duhamel form $\eqref{mildCO}$ starting at $t_1$ where both parts are positive. The control we have on the $L$ operator, Lemma $\ref{lem:L}$, allows us to bound from above the first term.  Moreover, this bound on $L$ is independent on $x$ and $v$ (see proof of Lemma $\ref{lem:positivity1}$). This gives

\begin{equation}\label{1ststepCO}
f(t,x,v) \geq e^{-(t-t_1)C_L\langle R \rangle^{\gamma^+}}f\left(t_1,X_{t_1,t}(x,v),V_{t_1,t}(x,v)\right). 
\end{equation}

\bigskip
Finally, Proposition $\ref{prop:upheaval}$ applied to $f(t_1,X_{t_1,t}(x,v),V_{t_1,t}(x,v))$ gives us the property for $n=0$.

\bigskip
\textbf{Step $2$: Proof of the induction}. We consider the case where the proposition is true for $n$.
\\Given $l \in (0,\delta_X]$, $t \in [l/(2^{n+1}R),\tau]$, $x \in \bar{\Omega}$ and $v\in B(0,R)$.
\par We suppose now that there exists $t_1 \in [0,t-l/(2^{n+1}R)]$ such that $X_{t_1,t}(x,v) \in \Omega_l \cap B(x_i,\delta_X/2^{n+1})$.

\bigskip
Similar to what we did in the first step of the induction, but concentrating on the second part of the Duhamel formula $\eqref{mildCO}$ we conclude that

\begin{eqnarray}
&& \quad\quad f(t,x,v) \geq \label{inductionCO}
\\ &&  e^{-C_L \tau \langle R \rangle^{\gamma^+}} \left(\int_{t_1+\frac{l}{2^{n+3}R}}^{t_1 + \frac{l}{2^{n+2}R}}  Q^+ \left[f(s,X_{s,t}(x,v),\cdot),f(s,X_{s,t}(x,v),\cdot)\right]\:ds\right)\left(V_{t_1,t}(x,v)\right). \nonumber
\end{eqnarray}
\bigskip

The goal is now to apply the induction to the triplet $(s,X_{s,t}(x,v),v_*)$, where $v_*$ is the integration parametre inside the $Q^+$ operator, with $\norm{v_*} \leq R$.
\par One easily shows that $X_{s,t}(x,v) = X_{t_1,t}(x,v) + (s-t_1)V_{t_1,t}(x,v)$, for $s$ in $[t_1+\frac{l}{2^{n+3}R},t_1 + \frac{l}{2^{n+2}R}]$, and therefore we have that

\begin{equation}\label{geo1}
\norm{X_{t_1,t}(x,v)-X_{s,t}(x,v)} \leq \frac{l}{2^{n+2}},
\end{equation}
and so that $X_{s,t}(x,v)$ belongs to $\Omega_{l-l/2^{n+2}}$.

\bigskip
Finally, we have to find a point on the characteristic trajectory of $(s,X_{s,t}(x,v),v_*)$ that is in $\Omega_{l'}$ for some $l'$. This is achieved at the time $t_1$ (see Fig.$\ref{fig:spreadfar}$).

\bigskip
\begin{figure}[!h]
\begin{center}
\includegraphics[scale=0.68]{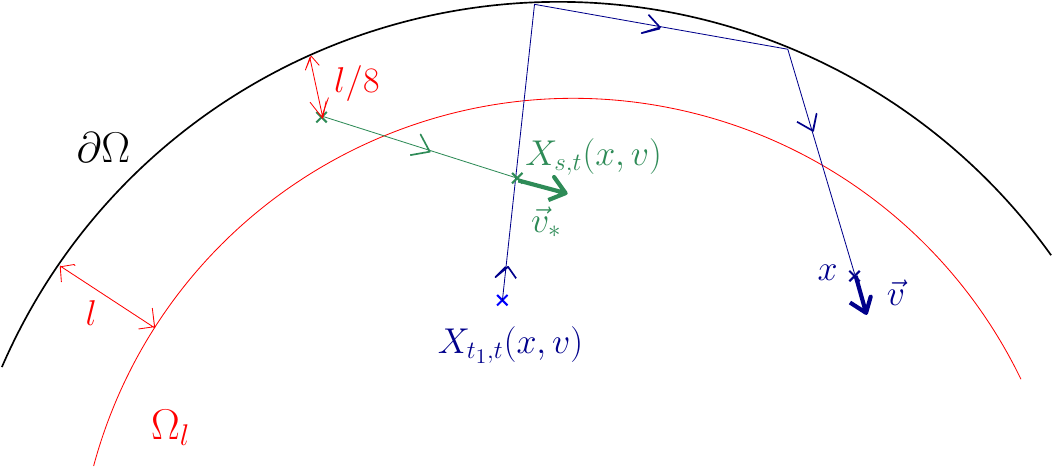}
\end{center}
\caption{\footnotesize Study of  $(s,X_{s,t}(x,v),v_*)$ far from the boundary}
\label{fig:spreadfar}
\end{figure}
\bigskip

Indeed, we have $s$ in $[t_1+l/(2^{n+3}R),t_1+l/(2^{n+2}R)]$ so, for $\norm{v_*} \leq R$

\begin{equation}\label{geo2}
\forall s' \in [t_1,s], \quad\norm{X_{s,t}(x,v) - \left(X_{s,t}(x,v) - (s-s')v_*\right)} \leq \frac{l}{2^{n+2}}.
\end{equation}

\bigskip 
 This gives us the characteristics trajectory backward starting from $s$, since $X_{s,t}(x,v) - (s-s')v_*$ remains in $\Omega$, and therefore
 
$$\forall s' \in [t_1,s], \quad \left\{\begin{array}{rl} \displaystyle{X_{s',s}\left(X_{s,t}(x,v),v_*\right) }&\displaystyle{= X_{s,t}(x,v) - (s-s')v_*} \vspace{2mm} \\ \vspace{2mm} \displaystyle{V_{s',s}\left(X_{s,t}(x,v),v_*\right) }&\displaystyle{= v_*.} \end{array} \right.$$

\bigskip
To conclude we just need to gather the upper bounds we found about the trajectories reaching $(X_{s,t}(x,v),v_*)$ in  a time $s$ in $[t_1+l/(2^{n+3}R),t_1+l/(2^{n+2}R)]$, equations $\eqref{geo1}$ and $\eqref{geo2}$

\begin{eqnarray*}
\norm{X_{t_1,t}(x,v)- X_{t_1,s}\left(X_{s,t}(x,v),v_*\right)} \leq \frac{l}{2^{n+1}}.
\end{eqnarray*}

\bigskip
We have that $X_{t_1,t}(x,v)$ belongs to $\Omega_l \cap B(x_i, \delta_X/(2^{n+1}))$ and therefore we have that for all $s$ in $[t_1+l/(2^{n+3}R),t_1+l/(2^{n+2}R)]$, $X_{t_1,s}\left(X_{s,t}(x,v),v_*\right)$ belongs to $\Omega_{l/2} \cap B(x_i,\delta_X/2^n)$.
\par Finally, if $s$ belongs to $[t_1+l/(2^{n+3}R),t_1+l/(2^{n+2}R)]$ we have that $(l/8)/(2^n R)\leq s \leq \tau$ and $t_1$ is in $[0,s-(l/8)/(2^n R)]$.
\par We can therefore apply the induction assumption for $l'=l/8$ inside the $Q^+$ operator in $\eqref{inductionCO}$, recalling that $V_{t_1,s}(X_{s,t}(x,v),v_*)=v_*$.

$$f(t,x,v) \geq a_n\left(\frac{l}{8},\tau\right)^2 e^{-C_L \tau \langle R \rangle^{\gamma^+}} \left(\int_{t_1+\frac{l}{2^{n+3}R}}^{t_1 + \frac{l}{2^{n+2}R}}  Q^+ \left[\mathbf{1}_{B(v_i,r_n)},\mathbf{1}_{B(v_i,r_n)}\right]\:ds\right)\left(V_{t_1,t}(x,v)\right).$$

\bigskip
Applying the spreading property of $Q^+$, Lemma $\ref{lem:Q+spread}$, with $\xi = 1/4$ gives us the expected result for the step $n+1$.
\end{proof}
\bigskip

One easily notices that $(r_n)_{n\in\N}$ is a strictly increasing sequence. Moreover, for all $N$ in $\N$ we have that for all $1\leq i \leq N_X$, $v_i$ belongs to $B(0,R_{min})$. Therefore, by taking $N$ big enough (bigger than $N_1$ say) we have that 

$$\forall i \in \{1,\dots,N_X\}, \quad B(0,2R_{min}) \subset B(v_i,r_{N}).$$

\bigskip
This remark leads directly to the following corollary which stands for Proposition $\ref{prop:centredball}$ in the case when a point on the trajectory is far from the boundary of $\Omega$.

\bigskip
\begin{cor}\label{cor:centredballfar}
Let $f$ be the mild solution of the Boltzmann equation described in Theorem $\ref{theo:boundcutoff}$ and suppose that $f$ satisfies Proposition $\ref{prop:upheaval}$ with $\tau_0=0$.
\\Let $\Delta_T$ be in $(0,\delta_T]$ and take $\tau_1$ in $(0,\Delta_T]$.
\\ Then for all $0<l\leq \delta_X$, there exists $a(l,\tau_1,\Delta_T)>0$ and $0<\tilde{t}(l,\tau_1,\Delta_T)<\tau_1$ such that for all $t$ in $[\tau_1,\Delta_T]$, and every $(x,v)$ in $\bar{\Omega} \times \R^d$: if there exists $t_1 \in [0,t-\tilde{t}(l,\tau_1,\Delta_T)]$ such that $X_{t_1,t}(x,v)$ belongs to $\Omega_l$ then

$$f(t,x,v) \geq a(l,\tau_1,\Delta_T) \mathbf{1}_{B(0,2R_{min})}(v).$$
\end{cor}
\bigskip

\bigskip
\begin{proof}[Proop of Corollary $\ref{cor:centredballfar}$]
This is a direct consequence of Proposition $\ref{prop:spreadfar}$.
\par Indeed, take $0< l \leq \delta_X$, $0<\tau_1 \leq \Delta_T$ and $R=R(\Delta_T)>0$ such that $R\geq 3R_{min}$ and $l/R\leq \Delta_T$. Then  take $N_2 \geq N_1$ big enough such that $l/(2^{N_2}R) < \tau_1$. We emphasize here that $N_2$ depends on to $\tau_1$ so we write $N_2(\tau_1)$.
\par Now apply Proposition $\ref{prop:spreadfar}$ with $N = N_2(\tau_1)$ and for $t$ in $[\tau_1,\Delta_T]$. We obtain exactly Corollary $\ref{cor:centredballfar}$ (since $B(0,2R_{min}) \subset B(v_i,r_{N})$ for all $i$ and $R\geq 3R_{min}$) with

$$a(l,\tau_1,\Delta_T) = a_{N_2(\tau_1)}(l,\Delta_T) \quad \mbox{and} \quad \tilde{t}(l,\tau_1,\Delta_T) =\frac{l}{2^{N_2}R(\Delta_T)},$$
and the fact that $\bigcup\limits_{1 \leq i \leq N_X}B\left(x_i,\delta_x/2^{N}\right)$ covers $\bar{\Omega}$.
\end{proof}
\bigskip

%% file: boundCutoff_grazing.tex
\section{The cutoff case: geometry and grazing trajectories}\label{sec:cutoff_grazing}

We now turn to the case when the characteristic trajectory never escapes a small distance from the boundary of our convex domain $\Omega$.
\par Intuitively, by considering the case where $\Omega$ is a circle, one can see that such a behaviour is possible only when the angles of collisions with the boundary remain small (which corresponds in high dimension to the scalar product of the velocity with the outside normal being close to zero), or the angle is important but the norm of the velocity or the time of motion is small. Thus, by using the spreading property of the $Q^+$ operator we may be able to create larger balls in between two rebounds against the boundary because the latters should not change the velocity too much.
\par The study of grazing collisions will follow this intuition. First of all Section $\ref{subsec:geometry}$ proves a geometric lemma dealing with the fact that if the velocities are bounded from below and above, then for short times, the possibility for a trajectory to stay very close to the boundary implies that the velocity do not change a lot over time. Then Section $\ref{subsec:spreadgrazing}$ spreads a lower bound, in the same spirit as the last subsection, up to the point when this lower bound covers a centred ball in velocity. Notice that the geometric property forces us to work with velocities whose norm is bounded from below and so we shall have to take into account the speed of the spreading.


\subsection{Geometric study of grazing trajectories}\label{subsec:geometry}
The key point of the study of grazing collisions is the following geometric lemma. We emphasize here that this is the only part of the article where we need the fact that $\Omega$ is $C^2$.

\bigskip
\begin{prop}\label{prop:grazinggeo}
Let $\Omega$ be an open convex bounded $C^2$ domain in $\R^d$ and let $0<v_m<v_M$.
\\Then, for all $\eps >0$ there exists $t_\eps(v_M)$ such that for all $0<\tau_2\leq t_\eps(v_M)$ there exists $l_\eps(v_m,\tau_2) > 0$ such that for all $x$ in $\bar{\Omega}$ and all $v$ in $\R^d$ with $v_m \leq \norm{v} \leq v_M$,

$$\left(\forall s \in [0,\tau_2],\: X_{s}(x,v)\notin \Omega_{l_\eps(v_m,\tau_2)}\right) \Longrightarrow \left(\forall s \in [0,t_\eps(v_M)],\: \norm{V_s(x,v)-v}\leq \eps\right).$$

Furthermore, $l_\eps(v_m,\cdot)$ is an increasing function.
\end{prop}
\bigskip

The following is dedicated to the proof of Proposition $\ref{prop:grazinggeo}$.

\bigskip
We recall that for $x$ in $\bar{\Omega}$ and $v$ in $\R^d$ we define, see Appendix $\ref{appendix:transport}$, $t_{min}(x,v)$ to be the time of the first proper rebound when we start from $x$ with a velocity $-v$. This means that $t_{min}(x,v)$ does not take into account the case where a ball rolls on the boundary. This implies that one cannot hope to get continuity of the function $t_{min}$ because changing the velocity slightly may lead to a proper rebound instead of a rolling movement.
\par This being said, we define a time of collision against the boundary which will not take into account the possibility of rolling along the boundary of $\Omega$. This will not be too restrictive as we are considering a $C^2$ convex domain and therefore a trajectory that stays on the boundary will only reach a stopping point which happens only on a set of measure zero in the phase space (see Appendix $\ref{appendix:transport}$). Therefore we define for $x$ in $\bar{\Omega}$ and $v$ in $\R^d$, the first  forward  contact with the boundary, $t(x,v)$. It exists by the same arguments as for $t_{min}$. Notice that if $x$ is on $\partial\Omega$ then for all $v \neq 0$ we have that $t(x,v) =0$ if and only if $n(x)\cdot v \geq 0$, with $n(x)$ being the outward normal to $\partial\Omega$ at the point $x$.
\par We have the following Lemma dealing with the continuity of the outward normal to $\partial\Omega$ at the first forward contact point which will be of great interest for proving the crucial Proposition $\ref{prop:grazinggeo}$.

\bigskip
\begin{lemma}\label{lem:grazinggeo}
Let $\Omega$ be an open convex bounded $C^1$ domain in $\R^d$.
\\Then $\func{t}{(x,v)}{t(x,v)}$ is continuous from $\bar{\Omega}\times \left(\R^d-\{0\}\right)$ to $\R^+$.
\end{lemma}
\bigskip

\bigskip
\begin{proof}[Proof of Lemma $\ref{lem:grazinggeo}$]
Let suppose that $t$ is not continuous at $(x_0,v_0)$ in $\Omega\times \left(\R^d-\{0\}\right)$. Then

$$\exists \eps >0,\: \forall N\geq 1,\:\exists (x_N,v_N), \quad \left\{\begin{array}{l}\displaystyle{\norm{x_0-x_N}\leq 1/N} \vspace{2mm} \\\vspace{2mm} \displaystyle{\norm{v_0-v_N}\leq 1/N} \end{array}\right. \mbox{and}\quad \abs{t(x_0,v_0)-t(x_N,v_N)}>\eps.$$

\bigskip
If we still denote by $d_{\Omega}$ the diameter of $\Omega$, we obviously have that for all $N$, $0\leq t(x_N,v_N) \leq d_{\Omega}/\norm{v_N}$. Thus, $\left(t(x_N,v_N)\right)_{N\in\N}$ is a bounded sequence of $\R$ and we can extract a converging subsequence $\left(t(x_{\phi(N)},v_{\phi(N)})\right)$ such that $T = \lim\limits_{N\to+\infty}t(x_{\phi(N)},v_{\phi(N)})$.
\par By construction (see Appendix $\ref{appendix:transport}$) we have that for all $N$ in $\N$, $x_{\phi(N)} + t(x_{\phi(N)},v_{\phi(N)})v_{\phi(N)}$ belongs to $\partial\Omega$ which is closed. Moreover, this sequence converges to $x_0+T v_0$ which therefore is on $\partial\Omega$.
\par Finally we have that $\abs{t(x_0,v_0) - T}\geq \eps$. Since $\bar{\Omega}$ is convex,  the segment $[x_0,x_0 + \max(t(x_0,v_0),T)v_0]$ stays in $\bar{\Omega}$ and intersect the boundary at least at two distinct points. By convexity of the domain, this implies that the extreme points of the latter segment have to be on the boundary
which means that $x_0$ belongs to $\partial\Omega$ which is a contradiction.

\bigskip
Therefore, $t$ is continuous in $\Omega\times \left(\R^d-\{0\}\right)$. By the definition of $t(x,v)$ we have its continuity at the boundary. Indeed, $n(x)\cdot v \geq 0$ means we came from inside the domain to reach that point and we have
$$\abs{t(x',v)-t(x,v)} \leq \frac{\norm{x-x'}}{\norm{v}}.$$

\end{proof}
\bigskip

We are now ready to prove the geometric Proposition $\ref{prop:grazinggeo}$.

\bigskip
\begin{proof}[Proof of Proposition $\ref{prop:grazinggeo}$]
Consider $\eps >0$ and $0<v_m<v_M$.

\bigskip
\textbf{Step $1$: the case of segments}. The first step is to understand that if a whole trajectory stays close to the boundary, then the angle made by the velocity with respect to the normal at the point of collision is close to $\pi/2$ for dimension $d=2$. The same behaviour in higher dimensions is described by the scalar product of the direction of the trajectory and the normal being close to zero. One has to remember that controlling $\norm{V_s(x,v)-v}$ is the same as controlling the scalar products of the trajectory and the normal on the boundary at each collision point (see definition of $V_s(x,v)$ in Appendix $\ref{appendix:transport}$).
\par Let $x$ be on $\partial\Omega$ and $p$ in $\N^*$. We define
$$\Gamma_p(x) = \left\{\abs{n(x)\cdot v} : \: v\in\mathbb{S}^{d-1}\:\mbox{s.t.}\: n(x)\cdot v < 0 \:\:\mbox{and}\:\:\forall s \in [0,t(x,v)], \:x+sv \notin \Omega_{1/p}  \right\},$$
with $\Omega_{1/p}$ being defined by $\eqref{Ul}$.

\bigskip
$\Gamma_p(x)$ gives us the values of scalar products between a normal on the boundary and all the directions that create a characteristic trajectory which stays at a distance less than $1/p$ from the boundary in between two distinct rebounds (see Fig.$\ref{fig:anglegrazing}$). This is exactly what we would like to control uniformly on the boundary.
\par We remark that $\Gamma_p(x)$ is not empty because $\Omega$ and, thus, $\Omega_{1/p}$ are convex and by the geometric theorem of Hahn-Banach we can separate $\Omega_{1/p}$ and a disjoint convex ball containing $x$. It is also straightforward, a mere Cauchy-Schwartz inequality, that $\Gamma_p(x)$  is bounded from above by $1$. Therefore we can define, for all $p$ in $\N^*$,

$$\function{h_p}{\partial\Omega}{\R^+}{x}{\sup\Gamma_p(x).}$$

\bigskip
We are going to prove that $(h_p)_{p\in\N^*}$ satisfies the following properties: it is a decreasing sequence of functions, $h_p$ is continuous in $x$ for each $p\geq 1$ and for all $x$ in $\partial\Omega$ $(h_p(x))_{p\in\N^*}$ converges to $0$.
\par The fact that $(h_p)$ is decreasing is obvious.
\par In order to prove the continuity of $h_p$ we take an $x$ on the boundary and $v$ in $\mathbb{S}^{d-1}$ such that $\abs{n(x)\cdot v}$ is in $\Gamma_p(x)$. We have that for all $s$ in $[0,t(x,v)]$

$$ d(x+sv,\partial\Omega) < 1/p.$$

\bigskip
The distance to the boundary is a continuous function and $[0,t(x,v)]$ is compact so there exists $s(x,v)$ in the latter interval such that $d(x+s(x,v)v,\partial\Omega)$ is maximum. Because $\Omega$ is convex we have that $\Omega_{1/p}$ is convex and therefore

$$\forall s \in [0,t(x,v)],\quad B\left(x+sv,\frac{d(x+s(x,v)v,\Omega_{1/p})}{2}\right)\cap \Omega_{1/p} = \emptyset.$$

\bigskip
Then for all $x'$ on the boundary such that $\norm{x-x'}\leq d(x+s(x,v)v,\Omega_{1/p})/2$ we have that for all $s$ in $[0,t(x',v)]$, $x'+sv$ is not in $\Omega_{1/p}$. Lemma $\ref{lem:grazinggeo}$ gives us that if $x'$ is close to $x$ then $t(x',v) >0$ and thus $v$ is not tangential at $x'$ either. Moreover $\Omega$ is $C^2$ so the outward normal to the boundary is continuous and therefore for $x'$ even closer to $x$ we have that $v$ is such that $\abs{n(x')\cdot v}$ is also in $\Gamma_p(x')$. To conclude, we notice that the scalar product is continuous and therefore for all $\eta >0$ we obtain

$$-\eta \leq \Big|\abs{n(x')\cdot v}-\abs{n(x)\cdot v}\Big| \leq \eta,$$
when $x'$ is close enough to $x$.

\bigskip
The same arguments with the same constants (since our continuous functions act on compact sets and therefore are uniformly continuous) if $x'$ is close to $x$ then taking $\abs{n(x')\cdot v}$ in $\Gamma_{1/p}(x')$ we have $\abs{n(x)\cdot v}$ in $\Gamma_{1/p}(x)$ and the same inequality as above. This gives us the continuity of $h_p$ at $x$. Indeed, we showed that for all $x'$ close to $x$ and for all element $u$ in $\Gamma_{1/p}(x)$ we can find an element $u'$ in $\Gamma_{1/p}(x')$ that is close to $u$.
\par Finally, it remains to show that for $x$ on the boundary we have that $h_p(x)$ tends to $0$ as $p$ tends to $+\infty$.
\\One can notice that the vector $-n(x)$ is the maximum possible in $\Gamma_p(x)$ and is exactly the direction of the diametre in $\Omega$ passing by $x$. Hence, simple convexity arguments lead to the fact that if all the segments of the form $[x,x-t(x,-n(x))n(x)]$ intersect $\Omega_{1/p}$ then we have that for all $x$ on the boundary, there exists $v_p(x)$ in $\mathbb{S}^{d-1}$ such that $n(x)\cdot v_p(x) = -h_p(x)$. Moreover, the segment $[x,x+t(x,v_p(x))v_p(x)]$ is tangent to $\Omega_{1/p}$ and we denote by $x_p$ its first contact point (see Fig.$\ref{fig:anglegrazing}$). The convexity of $\Omega$ and $\Omega_{1/p}$ shows that, as $p$ increases, $x_p$ gets closer to $x$ and to the boundary ($\Omega$ is convex). Therefore $v_p(x)$ tends to a tangent vector of the boundary at $x$. This shows that

$$\lim\limits_{p \to +\infty}h_p(x) = 0$$
in the case where all the segments of the form $[x,x-t(x,-n(x))n(x)]$ intersect $\Omega_{1/p}$.

\bigskip
We now come to the case where the segments of the form $[x,x-t(x,-n(x))n(x)]$ do not all intersect $\Omega_{1/p}$. If for all $p$, this segment does not intersect $\Omega_{1/p}$ this implies by convexity of $\Omega$ that $[x,x-t(x,-n(x))n(x)]$ is included in $\partial\Omega$. But then $-n(x)$ is not only a normal vector to the boundary at $x$ but also a tangential one at $x$. Geometrically this means that $x$ is a corner of $\partial\Omega$ and $n(x)$ is ill-defined. This is impossible for $\Omega$ being $C^2$. Hence, for all $x$ on the boundary, it exists $p(x)$ such that the segment at $x$ intersect $\Omega_{p(x)}$. However, $\Omega$ is $C^2$ and we also have Lemma $\ref{lem:grazinggeo}$. Those two facts implies that $p(x)$ is continuous on $\partial\Omega$ which is compact and therefore $p(x)$ reaches a maximum. Let us call this maximum $P$. For all $p\geq P$, all the segments of the form $[x,x-t(x,-n(x))n(x)]$, $x$ in $\partial\Omega$, intersect $\Omega_P$ and we conclude thanks to the previous case.

\bigskip
Thanks to these three properties and the fact that $\partial\Omega$ is compact, we are able to use Dini's theorem. We therefore find that $(h_p)_{p\in\N^*}$ converges uniformly to $0$. By taking $p_\eps$ big enough we have that for a segment of a characteristic trajectory joining two points on the boundary to be outside $\Omega_{p_\eps}$ we must have $\Gamma_{p_\eps} \leq \eps$ for any $x$ on the boundary (see Fig.$\ref{fig:anglegrazing}$).

\bigskip
\textbf{Step $2$: more general trajectories}. We take $x$ in $\partial\Omega$ and $v$ such that $v_m\leq \norm{v} \leq v_M$ and we suppose that for a given $t>0$

$$\forall s \in [0,t], \quad X_s(x,v) \notin \Omega_{1/p_{(\eps/2N_{max})}},$$
$N_{max}$ to be define later.

\bigskip
We are about to find a uniformly small time such that trajectories having at least two collisions against the boundary do not undergo an important evolution of velocity. This will be achieved thanks to the facts that $\norm{v}\leq v_M$ and that the maximum of the scalar product is attained at a critical vector and which is the only one that needs to be controlled.
\par Thanks to Proposition $\ref{prop:nbrebounds}$, $(X_s(x,v))_{s}$ has countably many rebounds against the boundary (almost surely a finite number in fact). We denote by $(t_i)_{(i\in\N)}$ the sequence of times between consecutive collisions and by $(l_i)_{i\in\N}$ the distance travelled during these respective times. We have that
$$\forall i \in \N, \quad l_i = \abs{v}t_i \quad\mbox{and}\quad v_m t\leq \sum\limits_{i\in \N} l_i \leq v_M t.$$

Therefore, for all $\eta >0$, there exists $N_{\eta}(x,v)$ in $\N$ such that
\begin{equation}\label{controlsmalltimes}
\sum\limits_{i>N_{\eta}(x,v)} t_i \leq \eta.
\end{equation}

By continuity of $t(x,v)$, see Lemma $\ref{lem:grazinggeo}$, and the fact that $t(x,v)=0$ if and only if $n(x)\cdot v \geq 0$, we have that for $\eta$ small enough $\eqref{controlsmalltimes}$ yields
\begin{equation}\label{controlsmallangles}
\sum\limits_{i>N_{\eta}(x,v)} \abs{n(x_i)\cdot v_i} \leq \eps/4,
\end{equation}
where $v_i$ is the velocity after the $i^{th}$ rebound and $x_i$ is the $i^{th}$ footprint.
\par $t(x,v)$ is uniformly continuous on the compact $\partial\Omega \times \{\abs{v} = v_M\}$ (see Lemma $\ref{lem:grazinggeo}$) therefore the footprints of $(X_s(x,v))_{s\in[0,t]}$ are uniformly continuous and therefore there exists $\alpha^{(1)}_X > 0$ and $N_{max}$ in $\N$ such that
\begin{equation}\label{Nmax}
\forall x,x' \in \partial\Omega \:\:\:\mbox{s.t.}\:\:\:\norm{x-x'}\leq \alpha_X, \:\forall v_m \leq \abs{v} \leq v_M, \quad N_{\eta}(x,v) \leq N_{max}-1.
\end{equation}
We have now defined $N_{max}$.

\bigskip
The first property to notice is that  if $(X_s(x,v))_{s\in[0,t]}$ has at least two rebounds against the boundary, then at each of them the scalar product between the incoming velocity and the outward normal is less than $\eps/2N_{max}$.
\par Secondly, $\Omega$ is $C^2$ and therefore $n(x)$ is uniformly continuous on the boundary. Thus, the specular reflection operator $\mathcal{R}_x$ is uniformly continuous on $\partial\Omega \times B(0,v_M)$:
\begin{equation}\label{controlreflection}
\exists \alpha^{(2)}_X > 0,\:\forall x,x' \in \partial\Omega \:\:\:\mbox{s.t.}\:\:\:\norm{x-x'}\leq \alpha_X, \quad \norm{\mathcal{R}_x-\mathcal{R}_{x'}}\leq \eps/4N_{max}.
\end{equation}

\bigskip
We want to be sure that straight trajectories stay in our domain of uniformity so we consider
$$t \leq t_\eps(v_M) = \max\left(\frac{\alpha_X}{v_M},\frac{1}{p_{\eps/2N_{max}}v_M}\right),$$
where $\alpha_X = \min(\alpha^{(1)}_X,\alpha^{(2)}_X)$ defined in $\eqref{Nmax}$ and $\eqref{controlreflection}$.
To conclude, thanks to $\eqref{Nmax}$ and $\eqref{controlsmallangles}$, if $(X_s(x,v))_{s\in[0,t]}$ collides at least twice with the boundary then

$$\forall s \in [0,t], \quad \norm{v-V_s(x,v)}\leq 2\sum_{i\in\N}\abs{n(x_i)\cdot v_i} \leq 2 \sum_{i \leq N_{max}-1} \frac{\eps}{4N_{max}} + 2\frac{\eps}{4} = \eps.$$

\bigskip
Roughly speaking we do not allow the velocities near the critical direction to bounce against the wall and for the grazing ones we run them for a short time, preventing them from escaping a small neighbourhood where the collisions behave almost the same everywhere (see Fig.$\ref{fig:anglegrazing}$).

\bigskip
To conclude our proof, it only remains to find $l\leq 1/p_{\eps/2N_{max}}$ that prevents trajectories staying in $\Omega_{l}$ but go through only one rebound with a scalar product greater than $\eps/2$ from happening. This is easily achieved by taking $l$ small enough such that not a single trajectory with a scalar product greater than $\eps/2N_{max}$ can stay inside $\Omega_{l}$ during a time $\tau$. Indeed, one part of these trajectories will overcome a straight line of lenght at least $v_m \tau /2$ and making a scalar product greater than $\eps/2N_{max}$. The distance from the boundary of the extremal point of these straight lines is therefore, by convexity, uniformly bounded from below (e.g. in dimension $2$ it is bounded by $v_m\tau\eps/4N_{max}$. Taking $l_{\eps}(v_m,\tau)$ being the minimum between this lower bound and $1/p_{\eps/2N_{max}}$ gives us the required distance from the boundary.

\bigskip
\begin{figure}[!h]
\begin{center}
\includegraphics[scale=0.68]{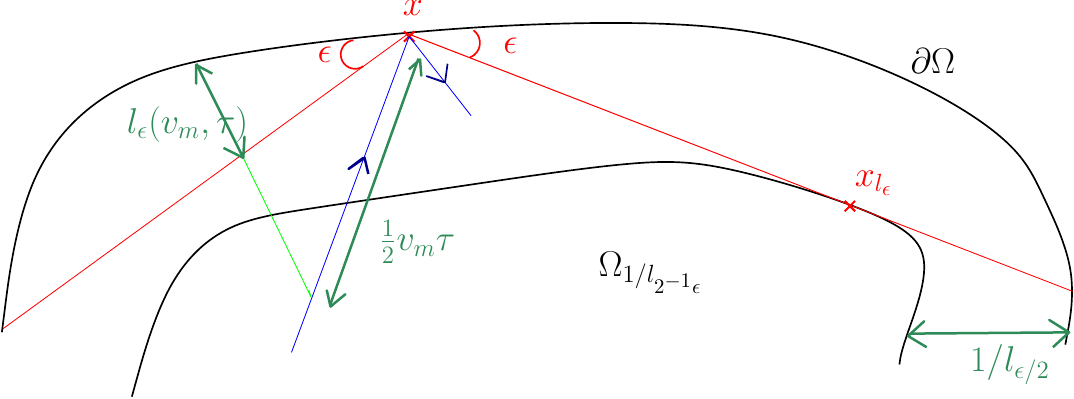}
\end{center}
\caption{\footnotesize Control on grazing trajectories}
\label{fig:anglegrazing}
\end{figure}
\bigskip
\end{proof}
\bigskip

\bigskip
\begin{remark}\label{rem:constructivegrazing}
In the case of $\Omega$ is a strictly convex $C^3$ domain, the proof of Proposition $\ref{prop:grazinggeo}$ can be easily made constructive thanks to the tools developed by Guo \cite{Gu6}.
\par In that case we have the existence of $\func{\xi}{\R^d}{\R}$ to be $C^3$ such that 
$$\Omega = \{x\in\R^d,\quad \xi (x)<0\}$$
and such that $\nabla\xi \neq 0$ on $\partial\Omega$ and there exists $C_\xi>0$ such that 
$$\partial_{ij}\xi (x)v_iv_j \geq C_\xi \norm{v}^2$$
for all $x$ in $\bar{\Omega}$ and all $v$ in $\R^d$. It allows us to define the following bounded functional along a characteristic trajectories $(X_s,V_s)$,
$$\alpha(s) = \xi^2(X_s) + \left[V_s\cdot\nabla\xi (X_s)\right]^2 - 2 \left[V_s\cdot\nabla^2\xi (X_s)\cdot V_s\right]\xi (X_s) \geq 0. $$
The latter functional satisfies that if $X_{s_0}$ is on $\partial\Omega$ then 
$$\alpha(s_0) = \left[V_{s_0}\cdot\nabla\xi (X_{s_0})\right]^2 = \left[V_{s_0}\cdot n(X_{s_0})\right]^2\abs{\nabla\xi (X_{s_0})}^2.$$
$\alpha$ thus encodes the evolution of the scalar product between the velocity of the trajectory and the normal to $\Omega$ at the footprints of the characteristic. If the characteristic trajectory starts with a velocity $v$ such that $v_m \leq \norm{v} \leq v_M$, as in Proposition $\ref{prop:grazinggeo}$, Lemma $1$ and Lemma $2$ of \cite{Gu6} shows that in between two consecutive collision with the boundary at time $s_1$ and $s_2$ we have the existence of $C_\xi >0$ such that
\begin{eqnarray}
\abs{s_1-s_2} &\geq& C_\xi \frac{\sqrt{\alpha(s_1)}}{v_M^2\abs{\nabla\xi(X_{s_0})}}, \label{timeinterval}
\\e^{C_\xi(v_m+1)s_1}\alpha(s_1) &\leq& e^{C_\xi(v_M+1)s_2}\alpha(s_2), \label{ineqalpha1}
\\e^{-C_\xi(v_M+1)s_1}\alpha(s_1) &\geq& e^{-C_\xi(v_m+1)s_2}\alpha(s_2). \label{ineqalpha2}
\end{eqnarray}
With $\eqref{timeinterval}$ we can control the minimum time between two consecutive collisions with the boundary and therefore the minimum lenght of a segment between two consecutive collisions, uniformly in $x$ and $v$ (since $\nabla\xi$ is bounded from below on $\partial\Omega$ and non-vanishing) . We therefore obtain a uniform maximum number of collisions during the given time $T$. Finally, $\eqref{ineqalpha1}$ and $\eqref{ineqalpha2}$ bounds uniformly the evolution of the scalar product between two consecutive collision and therefore the maximum evolution of $V_s(x,v)$ on the whole trajectory for a given time $T$. Plugging those constructive constants into the study we just made gives explicit constants in Proposition $\ref{prop:grazinggeo}$.
\end{remark}
\bigskip

Now that we understand how grazing trajectories behave geometrically we can turn our attention to their effects combined with the spreading property of the Boltzmann $Q^+$ operator.


\subsection{Spreading effect along grazing trajectories}\label{subsec:spreadgrazing}

In order to use the geometrical behaviour of grazing characteristic trajectories, one needs to consider velocities that are bounded from below. However, we would like to spread a lower bound up to ball centred at $0$ where a lower bound on the norm of velocities is impossible. We shall overcome this problem using the flexibility of the spreading property of the $Q^+$ operator, Lemma $\ref{lem:Q+spread}$, which allows us to extend the radius of the ball from $0$ up to $\sqrt{2}$ times the initial radius.
\par The idea is to spread the initial lower bound by induction as long as the origin is strictly outside, where we are allowed to use the geometrical property of grazing characteristics. Finally, a last iteration of the spreading property, not requiring any \textit{a priori} knowledge on characteristics, will include $0$ in the lower bound.
\par In Corollary $\ref{cor:centredballfar}$ we can fix a special time $\tau_1$ of crossing the frontier of some $\Omega_l$ allowing us to derive a lower bound for our function in this special case. The second case of grazing trajectories is dealt with Proposition $\ref{prop:grazinggeo}$ where we can find an $l$ for $\Omega_l$ to control the evolution of the velocity. Our goal now will be to find all the constants that are still free and to finally find a time of collision small enough that it will remain the same during all the iteration scheme.

\bigskip
We now fix all the constants that remain to be fixed in Corollary $\ref{cor:centredballfar}$ thanks to Proposition $\ref{prop:grazinggeo}$.
\par Let

\begin{equation}\label{DeltaT}
\Delta_T = \min\left(\delta_T, t_{\delta_V/4}(3R_{min})\right).
\end{equation}

Next we define, for $\xi$ in $(0,1)$,

\begin{equation} \label{rni}
\left\{\begin{array}{rl} \displaystyle{r_0(\xi) }&\displaystyle{= \delta_V} \vspace{2mm}\\\vspace{2mm} \displaystyle{r_{n+1}(\xi) }&\displaystyle{= \sqrt{2}(1-\xi)r_n(\xi)-\frac{\delta_V}{4}.} \end{array}\right.
\end{equation}

We have that $\left(r_n(1/2-5/(8\sqrt{2}))\right)_{n\in\N}$ is a strictly increasing sequence. Therefore, it exists $N_{max}$ such that

$$r_{N_{max}}\left(\frac{1}{2}-\frac{5}{8\sqrt{2}}\right) \geq 2R_{min}.$$

Now we fix $N$ in $\N^*$ greater than $N_{max}$. With this $N$ and Proposition $\ref{prop:upheaval}$ at $\tau_0=0$, we construct $v_1,\dots,v_{N_X}$.
\par For $i$ in $\{1,\dots,N\}$ we take $\xi^{(i)}$ in $(0,1/4-5/(8\sqrt{2})]$ and we define $N_{max}(i)$ to be such that $0 \notin B\left(v_i,r_n(\xi^{(i)})\right)$ for all $n < N_{max}(i)$ and $0 \in B\left(v_i,r_{N_{max}(i)}(\xi^{(i)})\right)$. We can in fact take $\xi^{(i)}$ such that $0 \in \mbox{Int}\left(B\left(v_i,r_{N_{max}(i)}(\xi^{(i)})\right)\right)$.
\par Therefore we have that for all $i$ in $\{1,\dots,N_X\}$,

$$\delta_i = \norm{v_i} - r_{N_{max}(i)-1}(\xi^{(i)}) \geq 0,$$
which is strictly positive if and only if $N_{max}(i) >0$. We consider 

\begin{equation}\label{v_m}
v_m = \min\limits_{i \in \{1,\dots,N_X\}}\{\delta_i; \: \delta_i >0\}.
\end{equation}

We can now define:

\begin{eqnarray}
\forall \: 0<\tau \leq \Delta_T,&& \quad R(\tau) = \max\left(3R_{min}, \:\frac{2\delta_X}{\tau}+1\right),\label{R}
\\ &&\quad \tau_1(\tau) = \tau - \frac{2\delta_X}{R(\tau)} >0 , \label{tau1}
\\ &&\quad \tilde{t}(\tau) = \tilde{t}(l(\tau),\tau_1(\tau),\Delta_T). \label{alpha}
\end{eqnarray}

Finally, we define  $l(\tau)$ 

\begin{eqnarray}
\forall \: 0<\tau \leq \Delta_T, &&\quad \tau_2(\tau) = \min\left(\Delta_T,\frac{\delta_X}{R(\tau)}\right),\label{tau2}
\\ &&\quad l(\tau) = \min\left(\delta_X, l_{\delta_V/4}\left(v_m,\tau_2(\tau)\right)\right).\label{l}
\end{eqnarray}

\bigskip
We also build up the following sequence, where $R$, $l$ and $\tau_1$ depend on $\tau$,

\begin{equation}\label{bn}
\left\{\begin{array}{rl} \displaystyle{b^{(i)}_0(\tau,\Delta_T) }&\displaystyle{= a_0e^{-(\Delta_T-\tau) C_L  \langle R\rangle^{\gamma^+}} }\vspace{2mm}\\\vspace{2mm} \displaystyle{b^{(i)}_{n+1}(\tau,\Delta_T) }&\displaystyle{= \min\left( C_Q r_n^{d+\gamma}(\xi^{(i)})^{d/2-1}\frac{\delta_X}{2^{n+2}R}e^{-\tau C_L  \langle R\rangle^{\gamma^+}}b^{(i)}_n(\tau,\Delta_T)^2; a(l,\tau_1,\Delta_T)\right)}\end{array}\right.
\end{equation}
 $\xi^{(i)}$ was defined above and $a(l,\tau,\Delta_T)$ was defined in Corollary $\ref{cor:centredballfar}$.

\bigskip
We are now ready to state the next Proposition which is the complement of Proposition $\ref{prop:spreadfar}$ in the case when the trajectory stays close to the boundary. We remind the reader that $0<\tilde{t}(\tau) <\tau_1(\tau)$.

\bigskip
\begin{prop}\label{prop:spreadgrazing}
Let $f$ be the mild solution of the Boltzmann equation described in Theorem $\ref{theo:boundcutoff}$ and suppose that $f$ satisfies Proposition $\ref{prop:upheaval}$ with $\tau_0=0$.
\\Consider $0<\tau \leq \Delta_T$ and take $i$ in $\{1,\dots,N_X\}$ such that $N_{max}(i)>1$.
\\For all $n$ in $\{0,\dots,N_{max}(i) -1\}$ we have that for all $t$ in $[\tau - \delta_X/(2^nR(\tau)),\Delta_T]$, all $x$ in $B(x_i,\delta_X/2^n)$ and all $v$ in $B(0,R(\tau))$, if
$$\forall s \in [0,t-\tilde{t}(\tau)], \quad X_{s,t}(x,v) \notin \Omega_{l(\tau)}$$
then
$$f(t,x,v) \geq b^{(i)}_n(\tau,\Delta_T) \mathbf{1}_{B\left(v_i,r_n(\xi^{(i)})\right)}(v),$$
all the constants being defined in $\eqref{DeltaT}, \eqref{rni}, \eqref{l}, \eqref{R}, \eqref{tau1}, \eqref{alpha}$ and $ \eqref{bn}$.
\end{prop}
\bigskip

\bigskip
\begin{proof}[Proof of Proposition $\ref{prop:spreadgrazing}$]
We are going to use the same kind of induction we used to prove Proposition $\ref{prop:spreadfar}$. So we start by fixing $i$ such that $N_{max}(i) >1$.

\bigskip
\textbf{Step $1$: Initialization}. The initialisation is simply Proposition $\ref{prop:upheaval}$ and the first term in the Duhamel formula $\eqref{mildCO}$ starting at $\tau$, with the control from above on $L$ thanks to Lemma $\ref{lem:L}$.

\bigskip
\textbf{Stef $2$: Proof of the induction}. We consider the case where the Proposition is true at $n \leq N_{max}(i)-2$.
\\We take $t$ in $[\tau - \delta_X/(2^{n+1}R(\tau)),\Delta_T]$, $x$ in $B(x_i,\delta_X/2^{n+1})$ and all $v$ in $B(0,R(\tau))$. 
\par We suppose now that for all $s \in [0,t-\tilde{t}(\tau)]$ we have that $X_{s,t}(x,v)$ does not belongs to $\Omega_{l(\tau)}$.

\bigskip
To shorten notation we will skip the dependence in $\tau$ of the constant.
\par We use the definition of $f$ being a mild solution to write $f(t,x,v)$ under its Duhamel form $\eqref{mildCO}$ where both parts are positive. As in the proof of Proposition $\ref{prop:spreadfar}$, we control, uniformly on $t$, $x$ and $v$, the $L$ operator from above. This yields

\begin{equation}\label{1ststepgrazing}
f(t,x,v) \geq e^{-C_L \tau \langle R \rangle^{\gamma^+}} \int_{t-\frac{\delta_X}{2^{n+1}R}}^{t-\frac{\delta_X}{2^{n+2}R}}  Q^+ \left[f(s,X_{s,t}(x,v),\cdot),f(s,X_{s,t}(x,v),\cdot)\right]\left(V_{s,t}(x,v)\right)\:ds, 
\end{equation}
where we used $\norm{V_{s,t}(x,v)} = \norm{v} \leq R$. We also emphasize here that this inequality holds true thanks to the definition of $\eqref{R}$:

$$t-\frac{\delta_X}{2^{n+1}R} \geq \tau - \frac{\delta_X}{R} > 0.$$

\bigskip
The goal is now to apply the induction to the triplet $(s,X_{s,t}(x,v),v_*)$, where $v^*$ is the integration parameter inside the $Q^+$ operator, with $\norm{v_*} \leq R$.
\par We notice first that for all $s$ in $[t-\delta_X/(2^{n+1}R),t-\delta_X/(2^{n+2}R)]$

\begin{eqnarray*}
\norm{x_i - X_{s,t}(x,v)} &\leq& \frac{\delta_X}{2^{n+1}} + \norm{x-X_{s,t}(x,v)}  
\\ &\leq& \frac{\delta_X}{2^{n+1}} + (t-s)R \leq \frac{\delta_X}{2^{n}},
\end{eqnarray*}
so that for all $s$ in $[t-\delta_X/(2^{n+1}R),t-\delta_X/(2^{n+2}R)]$, $X_{s,t}(x,v)$ belongs to $B(x_i,\delta_X/2^n)$.

\bigskip
We also note that

$$\left[t-\frac{\delta_X}{2^{n+1}R},t-\frac{\delta_X}{2^{n+2}R}\right] \subset \left[\tau-\frac{\delta_X}{2^{n}R},\Delta_T\right].$$

\bigskip
We have two different cases to consider for $(X_{s',s}(X_{s,t}(x,v),v_*))_{s'\in [0,s-\tilde{t}]}$.
\par Either for some $s'$ in $[0,s-\tilde{t}]$, $X_{s',s}(X_{s,t}(x,v),v_*)$ belongs to $\Omega_l$ and then we can apply Corollary $\ref{cor:centredballfar}$:

\bigskip
\begin{eqnarray}
f(s,X_{s,t}(x,v),v_*) &\geq& a(l,\tau_1,\Delta_T)\mathbf{1}_{B(0,2R_{min})}(v_*)\nonumber
\\ &\geq& b^{(i)}_n(\tau,\Delta_T) \mathbf{1}_{B\left(v_i,r_n(\xi^{(i)})\right)}(v) \label{grazing1},
\end{eqnarray}
\bigskip
since $v_i$ is in $B(0,R_{min})$.
\par Or for all $s'$ in $[0,s-\tilde{t}] \subset [0,\tau_2]$, $X_{s',s}(X_{s,t}(x,v),v_*)$ does not belong to $\Omega_l$ and then we can apply our induction property at rank $n$ and we reach the same lower bound $\eqref{grazing1}$.

\bigskip
Plugging $\eqref{grazing1}$ into $\eqref{1ststepgrazing}$ implies, thanks to the spreading property of $Q^+$, Lemma $\ref{lem:Q+spread}$ with $\xi = \xi^{(i)}$,

\begin{eqnarray}
&&\quad\quad f(t,x,v) \geq \label{2ndstepgrazing}
 \\&& C_Q r_n^{d+\gamma}(\xi^{(i)})^{d/2-1}e^{-\tau C_L  \langle R\rangle^{\gamma^+}}(b^{(i)}_n)^2 \int_{t-\frac{\delta_X}{2^{n+1}R}}^{t-\frac{\delta_X}{2^{n+2}R}}  \mathbf{1}_{B(v_i,\sqrt{2}(1-\xi^{(i)})r_n(\xi^{(i)}))}\left(V_{s,t}(x,v)\right)\:ds. \nonumber
\end{eqnarray}

\bigskip
To conclude we use the fact that for all $s$ in $[0,t-\tilde{t}]$ we have that $X_{s,t}(x,v)$ does not belong to $\Omega_{l}$ and that $t-\tilde{t} >\tau_2$. Moreover, $n+1 \leq N_{max}(i)-1$ and so if $v$ belongs to $B\left(v_i,r_n(\xi^{(i)})\right)$ we have that $v_m \leq \norm{v}$. We apply Proposition $\ref{prop:grazinggeo}$, raising

$$\forall s \in \left[t-\frac{\delta_X}{2^{n+1}R},t-\frac{\delta_X}{2^{n+2}R}\right], \quad \norm{v-V_{s,t}(x,v)}\leq \frac{\delta_V}{4}.$$

Therefore, if $v$ belongs to $B\left(v_i,r_{n+1}(\xi^{(i)})\right)$ we have that $V_{s,t}(x,v)$ belongs to $B(v_i,\sqrt{2}(1-\xi^{(i)})r_n(\xi^{(i)}))$ for all $s$ in $[t-\delta_X/(2^{n+1}R),t-\delta_X/(2^{n+2}R)]$.
\par Therefore if $v$ belongs to $B\left(v_i,r_{n+1}(\xi^{(i)})\right)$ we can compute explicitly $\eqref{2ndstepgrazing}$ and obtain the expected induction. 
\end{proof}
\bigskip

Thanks to Proposition $\ref{prop:spreadgrazing}$, we can build, for all $x$ and all $v$, a lower bound that will contain $0$ in its interior after another use of the spreading property of the $Q^+$ operator. The next Corollary is the complement of Corollary $\ref{cor:centredballfar}$.

\bigskip
\begin{cor}\label{cor:centredballgrazing}
Let $f$ be the mild solution of the Boltzmann equation described in Theorem $\ref{theo:boundcutoff}$ and suppose that $f$ satisfies Proposition $\ref{prop:upheaval}$ with $\tau_0=0$.
\\Let $\Delta_T$ be defined by $\eqref{DeltaT}$.
\\ There exists $r_V>0$ such that for all $\tau \in (0,\Delta_T]$ there exists $b(\tau)>0$ such that for all $t$ in $[\tau,\Delta_T]$
\par If, for $\tilde{t}(\tau)$ and $l(\tau)$ being defined by $\eqref{alpha} - \eqref{l}$, 
$$\forall s \in [0,t-\tilde{t}(\tau)], \quad X_{s,t}(x,v) \notin \Omega_{l(\tau)}.$$

Then

$$f(t,x,v) \geq b(\tau) \mathbf{1}_{B(0,r_V)}(v).$$
\end{cor}
\bigskip

\bigskip
\begin{proof}[Proof of Corollary $\eqref{cor:centredballgrazing}$]
We are going to use the spreading property of $Q^+$ one more time.
\bigskip
We recall that we chose $N \geq N_{max} \geq N_{max}(i)$ for all $i$. By definition of $N_{max}(i)$,

$$\forall i \in \{1,\dots,N_X\}, \quad 0 \in \mbox{Int}\left(B\left(v_i,r_{N_{max}(i)}(\xi^{(i)})\right)\right).$$

We define

$$r_V = \min\left\{r_{N_{max}(i)}(\xi^{(i)}) - \norm{v_i}; \: i \in \{1,\dots,N_X\} \right\},$$
which only depends on $\delta_V$ and $(v_{i})_{i \in \{1,\dots,N_X\}}$. By construction we see that

\begin{equation}\label{inclusioncentredball}
\forall i \in \{1,\dots,N_X\}, B(0,r_V) \subset B\left(v_i,r_{N_{max}(i)}(\xi^{(i)})\right).
\end{equation}

\bigskip
Now we take $\tau$ in $(0,\Delta_T]$ and we take $t$ in $[\tau,\Delta_T]$, $x$ in $B(x_i,\delta_X/2^{N})$ and $v$ in $B(0,R(\tau))$ such that

$$\forall s \in [0,t-\tilde{t}(\tau)], \quad X_{s,t}(x,v) \notin \Omega_{l(\tau)},$$

We have that $t$ is in $[\tau - \delta_X/(2^{N_{max}(i)-1}R(\tau)),\Delta_T]$ and $x$ in $B(x_i,\delta_X/2^{N_{max}(i)-1})$ ($N \geq N_{max}(i)$). By the same methods we reached $\eqref{2ndstepgrazing}$, we obtain for $n=N_{max}(i)$

\begin{eqnarray}
&& \quad\quad f(t,x,v) \geq \label{corspreading}
\\&& C_Q r_n^{d+\gamma}(\xi^{(i)})^{d/2-1}e^{-\tau C_L  \langle R\rangle^{\gamma^+}}(b^{(i)}_n)^2 \int_{t-\frac{\delta_X}{2^{n+1}R}}^{t-\frac{\delta_X}{2^{n+2}R}}  \mathbf{1}_{B(v_i,\sqrt{2}(1-\xi^{(i)})r_n(\xi^{(i)}))}\left(V_{s,t}(x,v)\right)\:ds. \nonumber
\end{eqnarray}

\bigskip
This time the conclusion is different because we cannot bound the velocity from below since our lower bound contains $0$. However, $\eqref{inclusioncentredball}$ allows us to bound from below the integrand in $\eqref{corspreading}$ by a function depending only on the norm. Moreover, $\norm{v} = \norm{V_{s,t}(x,v)}$ along characteristic trajectories (see Proposition $\eqref{prop:characteristics}$). Thus we obtain the expected result by taking

$$b(\tau) = \min\left\{b_{N_{max}(i)}^{(i)}; \: i \in \{1,\dots,N_X\}\right\}.$$
\end{proof}
\bigskip

%% file: boundCutoff_finalproof.tex
\section{Maxwellian lower bound in the cutoff case: proof of Theorem $\ref{theo:boundcutoff}$}\label{sec:cutoff_finalproof}

This section gathers all the results we proved above and proves the main Theorem in the case of a cut-off collision kernel.


\subsubsection{Proof of Proposition $\eqref{prop:centredball}$}\label{subsubsec:proofpropcentred}

By combining Corollary $\ref{cor:centredballfar}$ and Corollary $\ref{cor:centredballgrazing}$ we can deal with any kind of characteristic trajectory. This is expressed by the following lemma.

\bigskip
\begin{lemma}\label{lem:centredball}
Let $f$ be the mild solution of the Boltzmann equation described in Theorem $\ref{theo:boundcutoff}$ and suppose that $f$ satisfies Proposition $\ref{prop:upheaval}$ with $\tau_0=0$.
\\There exists $\Delta_T > 0$ and $r_V>0$ such that for all $0 <\tau \leq \Delta_T$ there exists $a(\tau)$ and

$$\forall t \in [\tau,\Delta_T],\:\mbox{a.e.}\:(x,v) \in \bar{\Omega}\times \R^d, \quad f(t,x,v) \geq a(\tau)\mathbf{1}_{B(0,r_V)}(v). $$
\end{lemma}
\bigskip

\bigskip
\begin{proof}[Proof of Lemma $\ref{lem:centredball}$]
In Corollary $\ref{cor:centredballgrazing}$ we constructed $\Delta_T$ and $r_V$.
\par We now take $\tau$ in $(0,\Delta_T]$ and consider $t$ in $[\tau,\Delta_T]$, $(x,v)$ in $\bar{\Omega} \times \R^d$ where $f$ is a mild solution of the Boltzmann equation.

\bigskip
We remind the reader that $l(\tau)$ and $\tilde{t}(\tau)$ have been introduced in $\eqref{l}$ and $\eqref{alpha}$.
\par Either $(X_{s,t}(x,v))_{s\in[0,t-\tilde{t}(\tau)]}$ meets $\Omega_{l(\tau)}$ and then we use Corollary $\ref{cor:centredballfar}$ to get

$$f(t,x,v) \geq a(l(\tau),\tau_1(\tau),\Delta_T)\mathbf{1}_{B(0,r_V)}(v).$$

\bigskip
Or $(X_{s,t}(x,v))_{s\in[0,t-\tilde{t}(\tau)]}$ stays out of $\Omega_{l(\tau)}$ and then we use Corollary $\ref{cor:centredballgrazing}$ to get

$$f(t,x,v) \geq b(\tau)\mathbf{1}_{B(0,r_V)}(v).$$

\bigskip
We obtain Lemma $\ref{lem:centredball}$ with $a(\tau) = \min\left(a(l(\tau),\tau_1(\tau),\Delta_T),b(\tau)\right)$.
\end{proof}
\bigskip

We now have all the tools to prove Proposition $\ref{prop:centredball}$.

\bigskip
\begin{proof}[Proof of Proposition $\ref{prop:centredball}$]
Let $\tau$ be strictly positive and consider $t$ in $[\tau/2,\tau]$.

\bigskip
\textbf{First case}. We suppose that $f$ satisfies Proposition $\ref{prop:upheaval}$ with $\tau_0=0$.
\par We can compare $t$ with $\Delta_T$ constructed in Lemma $\ref{lem:centredball}$.
\par If $t \leq \Delta_T$ then we can apply the latter lemma and obtain for almost every $(x,v)$ in $\bar{\Omega}\times\R^d$

\begin{equation}\label{1st}
f(t,x,v) \geq a\left(\frac{\tau}{2}\right)\mathbf{1}_{B(0,r_V)}(v).
\end{equation}

If $t \geq \Delta_T$ then we can use Duhamel formula $\eqref{mildCO}$ and bound $f(t,x,v)$ by its value at time $\Delta_T$ (as we did in the first step of the induction in the proof of Proposition $\ref{prop:spreadfar}$) and use Lemma $\ref{lem:centredball}$ at $\Delta_T$. This gives, for $\norm{v}\leq r_V$,

\begin{eqnarray}
f(t,x,v) &\geq& f(\Delta_T, X_{\Delta_T, t}(x,v),V_{\Delta_T,t}(x,v))e^{-(t-\Delta_T)C_L\langle r_V \rangle^{\gamma^+}} \nonumber
\\ &\geq& a(\Delta_T)e^{-(\tau-\Delta_T)C_L\langle r_V \rangle^{\gamma^+}}\mathbf{1}_{B(0,r_V)}(V_{\Delta_T,t}(x,v)) \nonumber
\\&=& a(\Delta_T)e^{-(\tau-\Delta_T)C_L\langle r_V \rangle^{\gamma^+}}\mathbf{1}_{B(0,r_V)}(v). \label{2nd}
\end{eqnarray}

\bigskip
We just have to take the minimum of the two lower bounds $\eqref{1st}$ and $\eqref{2nd}$ to obtain Proposition $\ref{prop:centredball}$.

\bigskip
\textbf{Second case}. We do not assume anymore that $f$ satisfies Proposition $\ref{prop:upheaval}$ with $\tau_0=0$.
\par Thanks to Proposition $\ref{prop:upheaval}$ with $\tau_0=\tau/4$ we have that
$$\forall t \leq 0,\:\forall x \in \bar{\Omega}, \: v\in \R^d, \quad F(t,x,v)=f(t+\tau_0,x,v)$$
is a mild solution of the Boltzmann equation satisfying exactly the same bounds as $f$ in Theorem $\ref{theo:boundcutoff}$ and such that $F$ has the property of Proposition $\ref{prop:upheaval}$ at $0$ (note that all the constants depend on $\tau_0$).
\par Hence, we can apply the first step for $t'$ in $[\tau/4,3\tau/4]$ and $F(t',x,v)$.  This gives us the expected result for $f(t,x,v)$ for $t = t'+\tau_0$ in $[\tau/2,\tau]$.

\end{proof}
\bigskip


\subsection{Proof of Theorem $\ref{theo:boundcutoff}$}\label{subsec:prooftheocutoff}

As was mentioned in Section $\ref{subsec:strategy}$, the main difficulty in the proof is to create a lower bound depending only on the norm of the velocity. This has been achieved thanks to Proposition $\ref{prop:centredball}$. If we consider this proposition as the start of an induction then it leads to exactly the same process developed by Mouhot in \cite{Mo2}, Section $3$. Therefore we will just explain how to go from Proposition $\ref{prop:centredball}$ to Theorem $\ref{theo:boundcutoff}$, without writing too many details.

\bigskip
First of all, by using the spreading property of the $Q^+$ operator once again we can grow the lower bound derived in Proposition $\ref{prop:centredball}$. 

\bigskip
\begin{prop}
Let $f$ be the mild solution of the Boltzmann equation described in Theorem $\ref{theo:boundcutoff}$.
\\ For all $\tau$ in $(0,T)$, there exists $R_0 >0$ such that
$$\forall n \in \N, \forall t \in \left[\tau-\frac{\tau}{2^{n+1}},\tau\right], \forall (x,v) \in \bar{\Omega}\times\R^d, f(t,x,v) \geq a_n(\tau)\mathbf{1}_{B(0,r_n)}(v),$$
with the induction formulae
$$a_{n+1}(\tau) = \emph{\mbox{cst}}\:C_e\frac{a_n^2(\tau)r_n^{d+\gamma}\xi_n^{d/2+1}}{2^{n+1}} \quad \mbox{and} \quad r_{n+1} = \sqrt{2}(1-\xi_n)r_n,$$
where $(\xi_n)_{n\in\N}$ is any sequence in $(0,1)$ and $r_0 = r_V$, $a_0(\tau)$ and $C_e$ only depend on $\tau$, $E_f$ (and $L^{p_\gamma}_f$ if  $\Phi$ satisfies $\eqref{assumptionPhi}$ with $\gamma < 0$).
\end{prop}
\bigskip

Indeed, we take the result in Proposition $\ref{prop:centredball}$ to be the first step of our induction and then, for $n$ in $\N$ and $0<\tau < T$,  the Duhamel form of $f$ gives

\begin{eqnarray*}
&&f(t,x,v) \geq  
\\&& \quad \int^{\tau - \frac{\tau}{2^{n+2}}}_{\tau - \frac{\tau}{2^{n+1}}}e^{-C_L(t-s)\langle v \rangle^{\gamma^+}}Q^+\left(f(s,X_{s,t}(x,v),\cdot),f(s,X_{s,t}(x,v),\cdot)\right)(V_{s,t}(x,v))ds,
\end{eqnarray*}
for $t$ in $[\tau - \tau/2^{n+2},\tau]$.
\par Using the induction hypothesis together with the spreading property of $Q^+$ (Lemma $\ref{lem:Q+spread}$) leads us, as in the proofs of Propositions $\ref{prop:spreadfar}$ and $\ref{prop:spreadgrazing}$, to a bigger ball in velocity, centred at $0$. The only issue is to avoid the $v$-dependence in $\mbox{exp}\left[-C_L(t-s)\langle v \rangle^{\gamma^+}\right]$ which can easily be achieved as shown at the end of the proof of Proposition $3.2$ in \cite{Mo2}. This is exactly the same result as Proposition $3.2$ in \cite{Mo2}, but with the added uniformity in $x$.

\bigskip
As in Lemma $3.3$ in \cite{Mo2}, we can take an appropriate sequence $(\xi_n)_{n\in\N}$ and look at the asymptotic behaviour of $\left(a_n(\tau)\right)_{n\in\N}$. We obtain the following
$$\forall \tau >0,\:\exists \rho_\tau,\theta_\tau >0,\:\forall (x,v) \in \bar{\Omega}\times\R^d, \quad f(t,x,v) \geq \frac{\rho_\tau}{(2\pi\theta_\tau)^{d/2}}e^{-\frac{\abs{v}^2}{2\theta}}.$$
Notice that, again, the result is uniform in space, since the previous one was, and that the constants $\rho_\tau$ and $\theta_\tau$ only depend on $\tau$ and the physical quantities associated to $f$.

\bigskip
To conclude, it remains to make the result uniform in time. As noticed in \cite{Mo2}, Lemma $3.5$, the results we obtained so far do not depend on an explicit form of $f_0$ but just on uniform bounds and continuity that are satisfied at all times, positions and velocities. Therefore, we can do the same arguments starting at any time and not $t=0$. So if we take $\tau >0$ and consider $\tau \leq t < T$ we just have to make the proof start at  $t - \tau$ to obtain Theorem $\ref{theo:boundcutoff}$.

%% file: boundNonCutoff.tex
\section{Exponential lower bound in the non cutoff case: proof of Theorem $\ref{theo:boundnoncutoff}$} \label{sec:noncutoff}

In this section we prove the immediate appearance of an exponential lower bound for solutions to the Boltzmann equation $\eqref{BE}$ in the case of a collision kernel satisfying the non cutoff property.
\par The definition of being a mild solution in the case of a non cutoff collision kernel, Definition $\ref{def:mildnoncutoff}$ and equation $\eqref{noncutoffsplitting}$, shows that we are in fact dealing with an almost cutoff kernel to which we add a non locally integrable remainder.  The strategy will mainly follow what we did in the case of a cutoff collision kernel with the addition of controlling the loss due to the added term.
\par As in the last section, we shall first prove that solutions to the Boltzmann equation can be uniformly bounded from below by a lower bound depending only on the norm of the velocity and then use the proof given for the non cutoff case in \cite{Mo2}. We will do that by proving the immediate appearance of localised ``upheaval points" and spreading them up to the point where we reach a uniform lower bound that includes a ball in velocity centred at the origin. The spreading effect will be done both in the case where the trajectories reach a point far from the boundary and in the case of grazing trajectories. At this point we will spread this lower bound on the norm of the velocity up to the exponential lower bound we expect.

\bigskip
We gather here two lemmas, proved in \cite{Mo2}, which we shall use in this section. They control the $L^\infty$-norm of the linear operator $S_\eps$ and of the bilinear operator $Q^1_\eps$. We first give a property satisfied by the linear operator $S$, $\eqref{noncutoffsplitting}$,  which is Corollary $2.2$ in \cite{Mo2}, where we define
\begin{equation}\label{mb}
m_b = \int_{\mathbb{S}^{d-1}}b\left(\mbox{cos}\:\theta\right)(1-\mbox{cos}\:\theta)d\sigma = \left|\mathbb{S}^{d-2}\right|\int_0^\pi b\left(\mbox{cos}\:\theta\right)(1-\mbox{cos}\:\theta) \mbox{sin}^{d-2}\theta \:d\theta.
\end{equation}

\bigskip
\begin{lemma}\label{lem:S}
Let $g$ be a measurable function on $\R^d$. Then
$$\forall v \in \R^d,\quad \abs{S[g](v)} \leq C_g^S\langle v \rangle^{\gamma^+},$$
where $C_g^S$ is defined by:
\begin{enumerate}
\item If $\Phi$ satisfies $\eqref{assumptionPhi}$ with $\gamma \geq 0$ or if $\Phi$ satisfies $\eqref{assumptionPhimol}$, then
$$C^S_g =\emph{\mbox{cst}}\: m_b C_\Phi e_g.$$
\item If $\Phi$ satisfies $\eqref{assumptionPhi}$ with $\gamma \in (-d,0)$, then
$$C^S_g = \emph{\mbox{cst}}\: m_b C_\Phi \left[e_g+ l^p_g\right],\quad p > d/(d+\gamma).$$
\end{enumerate}
\end{lemma}
\bigskip

We will compare the lower bound created by the cutoff part of our kernel to the remaining part $Q^1_\eps$. To do so we need to control its $L^\infty$-norm. This is achieved thanks to Lemma $2.5$ in \cite{Mo2}, which we recall here.

\bigskip
\begin{lemma}\label{lem:Q1}
Let $B=\Phi b$ be a collision kernel satisfying $\eqref{assumptionB}$, with $\Phi$ satisfying $\eqref{assumptionPhi}$ or $\eqref{assumptionPhimol}$ and $b$ satisfying $\eqref{assumptionb}$ with $\nu \in [0,2)$. Let $f,g$ be measurable functions on $\R^d$.
\\Then
\begin{enumerate}
\item If $\Phi$ satisfies $\eqref{assumptionPhi}$ with $2+\gamma \geq 0$ or if $\Phi$ satisfies $\eqref{assumptionPhimol}$, then
$$\forall v \in \R^d,\quad \abs{Q^1_b(g,f)(v)} \leq \emph{\mbox{cst}}\: m_b C_\Phi \norm{g}_{L^1_{\tilde{\gamma}}}\norm{f}_{W^{2,\infty}}\langle v \rangle ^{\tilde{\gamma}}.$$
\item If $\Phi$ satisfies $\eqref{assumptionPhi}$ with $2+\gamma < 0$, then
$$\forall v \in \R^d,\quad \abs{Q^1_b(g,f)(v)} \leq \emph{\mbox{cst}}\: m_b C_\Phi \left[\norm{g}_{L^1_{\tilde{\gamma}}}+\norm{g}_{L^p}\right]\norm{f}_{W^{2,\infty}}\langle v \rangle ^{\tilde{\gamma}}$$
with $p > d/(d+\gamma+2)$.
\end{enumerate}
\end{lemma}
\bigskip


\subsection{A lower bound only depending on the norm of the velocity}\label{subsec:centredballNCO}

In this section we prove the following proposition, which is exactly Proposition $\ref{prop:centredball}$ in the non-cutoff framework.

\bigskip
\begin{prop}\label{prop:centredballNCO}
Let $f$ be the mild solution of the Boltzmann equation described in Theorem $\ref{theo:boundnoncutoff}$.
\\For all $0 < \tau < T$ there exists $a_0(\tau)>0$ such that
$$\forall t \in [\tau/2,\tau],\:\forall(x,v) \in \bar{\Omega}\times \R^d, \quad f(t,x,v) \geq a_0(\tau)\mathbf{1}_{B(0,r_V)}(v), $$
$r_V$ and $a_0(\tau)$ only depending on $E_f$, $E'_f$, $W_f$ (and $L^{p_\gamma}_f$ if  $\Phi$ satisfies $\eqref{assumptionPhi}$ with $\gamma < 0$).
\end{prop}
\bigskip

\bigskip
\begin{proof}[Proof of Proposition $\ref{prop:centredballNCO}$]
As before, we would like to create localised ``upheaval points" (as the ones created in Proposition $\ref{prop:upheaval}$) and then extend them. Both steps are done, as in the cutoff case, by induction along the characteristics.
\par We have the following inequality
\begin{equation}\label{ineqNCO}
Q^+_\eps(f,f) + Q^1_\eps(f,f) \geq Q^+_\eps(f,f) -\abs{ Q^1_\eps(f,f)}.
\end{equation}
From the definition of being a mild solution in the non-cutoff case (Definition $\ref{def:mildnoncutoff}$), for any $0< \eps < \eps_0$,

\begin{eqnarray}
\label{mildNCO1} &&
\\&&f(t,X_t(x,v),V_t(x,v)) = f_0(x,v)\mbox{exp}\left[-\int_0^t \left(L_{\eps} + S_{\eps}\right)[f(s,X_s(x,v),\cdot)](V_s(x,v))\:ds\right]\nonumber 
\\ &&\quad\quad + \int_0^t \mbox{exp}\left(-\int_s^t \left(L_{\eps} + S_{\eps}\right)[f(s',X_{s'}(x,v),\cdot)](V_{s'}(x,v))\:ds'\right)\nonumber
\\ && \quad\quad\quad\quad\quad \left(Q^+_{\eps} + Q^1_{\eps}\right)[f(s,X_s(x,v),\cdot), f(s,X_s(x,v),\cdot)](V_s(x,v))\: ds.\nonumber
\end{eqnarray}
Due to Lemmas $\ref{lem:L}$, $\ref{lem:S}$ and $\ref{lem:Q1}$ we find that

\begin{equation}\label{lemmas1}
L_{\eps}[f] \leq C_f n_{b^{CO}_{\eps}}\langle v\rangle^{\gamma^+}, \quad S_{\eps}[f] \leq C_f m_{b^{NCO}_{\eps}}\langle v\rangle^{\gamma^+}
\end{equation}
and 

\begin{equation}\label{lemmas2}
\abs{Q^1_{\eps}(f,f)} \leq C_f m_{b^{NCO}_{\eps}}\langle v\rangle^{(2 + \gamma)^+}
\end{equation}
where $C_f >0$ is a constant depending on $E_f$, $E'_f$, $W_f$ (and $L^{p_\gamma}_f$ if  $\Phi$ satisfies $\eqref{assumptionPhi}$ with $\gamma < 0$).

\bigskip
The proof of Proposition $\ref{prop:centredballNCO}$ is divided into three different inductions that are dealt with in the same way as in the proof of Proposition $\ref{prop:centredball}$. Each induction represents a step in the proof: one to create localised initial lower bounds (Lemma $\ref{lem:positivity1}$), another one to deal with non-grazing trajectories (Proposition $\ref{prop:spreadfar}$) and the final one for grazing trajectories (Proposition $\ref{prop:spreadgrazing}$). Therefore, we will just point out below the only changes we need to make those inductions work in the non-cutoff case.

\bigskip
In all the inductions in the cutoff case, the key point of the induction was to control at each step quantities of the form
\begin{eqnarray*}
f(t,x,v) &\geq& \int_{t^{(1)}_n}^{t^{(2)}_{n}}\mbox{exp}\left(-\int_s^t \left(L_{\eps} + S_{\eps}\right)[f(s',X_{s'}(x,v),\cdot)](V_{s'}(x,v))\:ds'\right)
\\&&\quad\quad\left(Q^+_{\eps} + Q^1_{\eps}\right)[f(s,X_s(x,v),\cdot), f(s,X_s(x,v),\cdot)](V_s(x,v))\: ds,
\end{eqnarray*}
where $(t^{(1)}_n)_{n\in \N}$, $(t^{(2)}_n)_{n\in \N}$ are defined differently for grazing and non-grazing trajectories (see proofs of Propositions $\ref{prop:spreadfar}$ and $\ref{prop:spreadgrazing}$).

\bigskip
Much like those previous induction, and using $\eqref{ineqNCO}$, $\eqref{mildNCO1}$ and $\eqref{lemmas1}-\eqref{lemmas2}$, if $f(t,x,v) \geq a_n\mathbf{1}_{B(\bar{v},r_n)}$ then

$$f(t,x,v) \geq \int_{t^{(1)}_n}^{t^{(2)}_{n}}e^{-C_f^\eps(R)}\left(a_n^2 Q^+_{\eps}[\mathbf{1}_{B(\bar{v},r_n)},\mathbf{1}_{B(\bar{v},r_n)}] - C_f m_{b^{NCO}_{\eps}}\langle R\rangle^{(2 + \gamma)^+}\right)(V_s(x,v))\: ds,$$
which leads to

\begin{eqnarray}
&& \quad f(t,x,v) \geq \int_{t^{(1)}_n}^{t^{(2)}_{n}}e^{-C_f^\eps(R)} \label{finalineqNCO} 
\\ &&\quad\quad\left(a_n^2\mbox{cst}\: l_{b^{CO}_{\eps}} c_\Phi r_n^{d+\gamma} \xi_n^{\frac{d}{2}-1}\mathbf{1}_{B\left(\bar{v},r_n\sqrt{2}(1-\xi_n)\right)}- C_f m_{b^{NCO}_{\eps}}\langle R\rangle^{(2 + \gamma)^+}\right)(V_s(x,v))\: ds,\nonumber
\end{eqnarray}
due to the spreading property of $Q^+_{\eps}$ (see Lemma $\ref{lem:Q+spread}$) and using the shorthand notation $C^\eps_f(R) = C_f(n_{b^{CO}_{\eps}}+m_{b^{NCO}_{\eps}})\langle R\rangle^{\gamma^+}$.

\bigskip
To conclude we notice that, thanks to the definitions $\eqref{lb}$, $\eqref{nb}$ and $\eqref{mb}$,
$$l_{b^{CO}_{\eps}} \geq l_b$$
and
\begin{equation}\label{eqv1}
n_{b^{CO}_{\eps}} \mathop{\sim}\limits_{\eps \to 0} \frac{b_0}{\nu}\eps^{-\nu}, \quad m_{b^{NCO}_{\eps}} \mathop{\sim}\limits_{\eps \to 0} \frac{b_0}{2-\nu}\eps^{2-\nu}
\end{equation}
if $\nu$ belongs to $(0,2)$ and
\begin{equation}\label{eqv2}
n_{b^{CO}_{\eps}} \mathop{\sim}\limits_{\eps \to 0} b_0\abs{\mbox{log}\eps}, \quad m_{b^{NCO}_{\eps}} \mathop{\sim}\limits_{\eps \to 0} \frac{b_0}{2}\eps^{2}
\end{equation}
for $\nu=0$.
\par Thus, at each step of the inductions we just have to redo the proofs done in the cutoff case and choose $\eps = \eps_n$ small enough such that
\begin{equation}\label{choiceepsn}
C_f m_{b^{NCO}_{\eps_n}}\langle R\rangle^{(2 + \gamma)^+} \leq \frac{1}{2} a_n^2\mbox{cst}\: l_b c_\Phi r_n^{d+\gamma} \xi_n^{\frac{d}{2}-1}.
\end{equation}

\bigskip
Proposition $\ref{prop:centredballNCO}$ follows directly from these choices plugged into the study of the cutoff case.
\end{proof}
\bigskip


\subsection{Proof of Theorem $\ref{theo:boundnoncutoff}$}\label{subsec:prooftheononcutoff}

Now that we proved the immediate appearance of a lower bound depending only on the norm of the velocity we can spread it up to an exponential lower bound. As in Section $\ref{subsec:prooftheocutoff}$, we thoroughly follow the proof of Theorem $2.1$ of \cite{Mo2}. The proof in our case is exactly the same induction, starting from Proposition $\ref{prop:centredballNCO}$. Therefore we only briefly describe how to construct the expected exponential lower bound. For more details we refer the reader to \cite{Mo2}, Section $4$.

\bigskip
We start by spreading the initial lower bound (Proposition $\ref{prop:centredballNCO}$) by induction where, at each step, we use the spreading property of the $Q^+_{\eps_n}$ operator and fix $\eps_n$ small enough to obtain a strictly positive lower bound (see $\eqref{choiceepsn}$).
\par There is, however, a subtlety in the non-cutoff case that we have to deal with. Indeed, at each step of the induction we choose an $\eps_n$ of decreasing magnitude, but at the same time in each step the action of the operator $-(Q^-_\eps + Q^2_\eps)$ behaves like (see $\eqref{finalineqNCO}$)
$$\mbox{exp}\left[-C_f\left(m_{b^{NCO}_{\eps_n}}+n_{b^{CO}_{\eps_n}}\right)(t^{(1)}_{n}-t^{(2)}_n)\langle v \rangle^{\gamma^+}\right].$$
By $\eqref{eqv1}-\eqref{eqv2}$, as $\eps_n$ tends to $0$ we have that $n_{b^{CO}_{\eps_n}}$ goes to $+\infty$ and so the action of $-(Q^-_\eps + Q^2_\eps)$ seems to decrease the lower bound to $0$ exponentially fast. The idea to overcome this difficulty is to find a time interval $t^{(1)}_{n}-t^{(2)}_n = \Delta_n$ at each step to be sufficiently small to counterbalance the effect of $n_{b^{CO}_{\eps_n}}$.

\bigskip
More precisely, by starting from Proposition $\ref{prop:centredballNCO}$ as the first step of our induction, taking $$t^{(1)}_{n} = \left(\sum_{k=0}^{n+1}\Delta_k\right)\tau, \quad t^{(2)}_{n} = \left(\sum_{k=0}^n\Delta_k\right)\tau$$
in $\eqref{finalineqNCO}$ and  fixing $\eps_n$ by $\eqref{choiceepsn}$ we can prove the following induction property

\bigskip
\begin{prop}
Let $f$ be the mild solution of the Boltzmann equation described in Theorem $\ref{theo:boundnoncutoff}$.
\\ For all $\tau$ in $(0,T)$ and any sequence $\left(\Delta_n\right)_{n\in\N}$ such that $\sum_{n\geq 0}\Delta_n =1$,
$$\forall n \in \N, \forall t \in \left[\left(\sum_{k=0}^n\Delta_k\right)\tau,\tau\right], \forall (x,v) \in \bar{\Omega}\times\R^d, f(t,x,v) \geq a_n(\tau)\mathbf{1}_{B(0,r_n)}(v),$$
with the induction formulae
$$a_{n+1}= \emph{\mbox{cst}}\: \Delta_{n+1} \emph{\mbox{exp}}\left[-[\tilde{C}_f a_n^2 r_n^{d+\gamma-\tilde{\gamma}} \xi_n^{d/2-1}]^{-\frac{\nu}{2-\nu}}\left(\sum_{k\geq n+1}\Delta_k\right)r_n^{\gamma^+}\right]a_n^2 r_n^{\gamma+d}\xi_n^{d/2+1}$$
if $\nu$ is in $(0,2)$,
$$a_{n+1}= \emph{\mbox{cst}}\: \Delta_{n+1} \emph{\mbox{exp}}\left[-\emph{\mbox{cst}}\:\emph{\mbox{log}}[\tilde{C}_f a_n^2 r_n^{d+\gamma-\tilde{\gamma}} \xi_n^{d/2-1}]\left(\sum_{k\geq n+1}\Delta_k\right)r_n^{\gamma^+}\right]a_n^2 r_n^{\gamma+d}\xi_n^{d/2+1}$$
if $\nu =0$ and
$$r_{n+1} = \sqrt{2} r_n(1-\xi_n),$$
where $(\xi_n)_{n\in\N}$ is any sequence in $(0,1)$ and $r_0 = r_V$, $a_0(\tau)$ and $\tilde{C}_f$ depend only on $\tau$, $E_f$, $E'_f$, $W_f$ (and $L^{p_\gamma}_f$ if  $\Phi$ satisfies $\eqref{assumptionPhi}$ with $\gamma < 0$).
\end{prop}
\bigskip

We emphasize here that the induction formulae are obtained thanks to the use of equivalences $\eqref{eqv1}$ and $\eqref{eqv2}$ inside the exponential term
$$e^{-C_f\left(m_{b^{NCO}_{\eps_n}}+n_{b^{CO}_{\eps_n}}\right)(t^{(1)}_{n}-t^{(2)}_n)\langle R \rangle^{\gamma^+}} \geq e^{-C_f\left(m_{b^{NCO}_{\eps_n}}+n_{b^{CO}_{\eps_n}}\right)\left(\sum_{k\geq n+1}\Delta_k\right)\langle R \rangle^{\gamma^+}}$$
(see step $2$ of proof of Proposition $4.2$, Section $4$ in \cite{Mo2}).

\bigskip
As we obtain exactly the same induction formulae as in \cite{Mo2}, the asymptotic behaviour of the coefficients $a_n$ is the same. Thus, by choosing an appropriate sequence $(\Delta_n)_{n \in \N}$, as done in \cite{Mo2}, we can construct the expected exponential lower bound independently of time.

%% file: annexe1_transport.tex
\section{The free transport equation: proof of Theorem $\ref{theo:transport}$}\label{appendix:transport}

In this section, we study the transport equation with a given initial data and boundary condition in a bounded domain $\Omega$. We will only consider the case of purely specular reflections on the boundary $\partial \Omega$. Those kind of interaction cannot occur for all velocities at the boundary. Indeed, for a particle to bounce back at the boundary, we need its velocity to come from inside the domain $\Omega$. To express this fact mathematically, we define
$$\Lambda^+ = \left\{(x,v) \in \partial\Omega \times \R^d: \: v\cdot n(x) \geq 0\right\},$$
where we denote by $n(x)$ the exterior normal to $\partial\Omega$ at $x$.

\bigskip
Consider $\func{u_0}{\bar{\Omega} \times \R^d}{\R}$ which is $C^1$ in $x \in \Omega $ and $L^2(\bar{\Omega}\times\R^d) = L^2_{x,v}$. We are interested in the problem stated in Theorem $\ref{theo:transport}$, $\eqref{transporteq}-\eqref{transportBC}$.
\par If $D_x(v)(u)$ denotes the directional derivate of $u$ in $x$ in the direction of $v$ we have, in the case of functions that are $C^1$ in $x$,
$$D_x(v)(u) = v\cdot \nabla_x u.$$
Therefore, instead of imposing that the solution to the transport equation should be $C^1$ in $x$, we reformulate the problem with directional derivatives.

\bigskip
Physically, the free transport equation means that a particle evolves freely in $\Omega$ at a velocity $v$ until it reaches the boundary. Then it bounces back and moves straight until it reaches the boundary for the second time and so on so forth up to time $t$. The method of characteristics is therefore the best way to link $u(t,x,v)$ to $u_0$ by just following the path used by the particle, backwards from $t$ to $0$ (see Fig.$\ref{fig:bouncing}$).

\bigskip
\begin{figure}[h!]
\begin{center}
\includegraphics[scale=0.68]{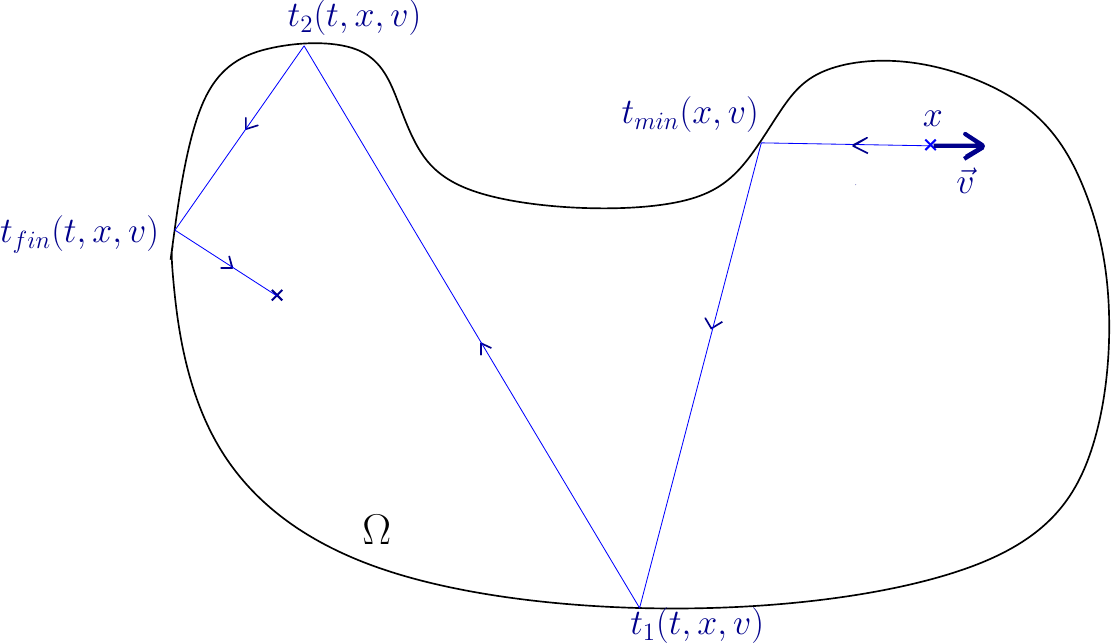}
\end{center}
\caption{\footnotesize Backward trajectory with standard rebounds}
\label{fig:bouncing}
\end{figure}
\bigskip

This method has been used in \cite{Gu5} on the half-line and in \cite{JM}, \cite{Hw}, for instance, in the case of convex media. However, in both articles they only deal with finite, or countably many,  numbers of rebounds in finite time. Indeed, the electrical field in \cite{Gu5} and \cite{Hw} makes the particles always reach the boundary with $v \cdot n(x) >0$ and \cite{JM} has a specular boundary problem with an absorption coefficient $\alpha \in [0,1)$: $u(t,x,v) = \alpha u(t,x,\mathcal{R}_x(v))$. Therefore, in the case the particle arrives tangentially to the boundary, i.e. $v\cdot n(x) = 0$, we have $\mathcal{R}_x(v) =v$ and so $u(t,x,v)=0$. This vanishing property allowed the authors to not care about the special cases where the particle starts to roll on the boundary.
\par Another way of looking at the characteristics method is to study the footprints of the trajectories on the boundary. This problem, as well as the possibility of having infinitely many rebounds in a finite time, has been tackled by Tabachnikov in \cite{Ta}. Tabachnikov only focused on boundary points since the description of the trajectories by only considering their collisions with the boundary holds a symplectic property and a volume-preserving transformation. Such properties allowed him to show that the set of points on the boundary that lead to infinitely many rebounds in finite time is of measure $0$ (\cite{Ta}, Lemma $1.7,1$). Unfortunately, in our case we would like to follow the characteristics and the study of trajectories only via their footprints on the boundary is no longer a volume-preserving transformation.

\bigskip
In our case we need to follow the path of a particle along the characteristics of the equation to know the value of our function at each step. If the particle starts to roll on the boundary (see Fig.$\ref{fig:rolling}$) we require to know for how long it will do so. The major issue is the fact that $v \cdot n(x) =0$ does not tell us much about the geometry of $\partial\Omega$ at $x$ and the possibility, or lack of, for the particle to keep moving tangentially to the boundary. Moreover, some cases lead to non physical behaviour since the sole specular collision condition implies that some pairs $(x,v) \in \partial\Omega\times \R^d$ can only be starting points, they cannot be generated by any trajectories (see Fig.$\ref{fig:stop}$). This case is mentioned quickly in the first chapter of  \cite{Ta1} but not dealt with.

\bigskip
\begin{figure}[!h]
\begin{center}
\includegraphics[scale=0.68]{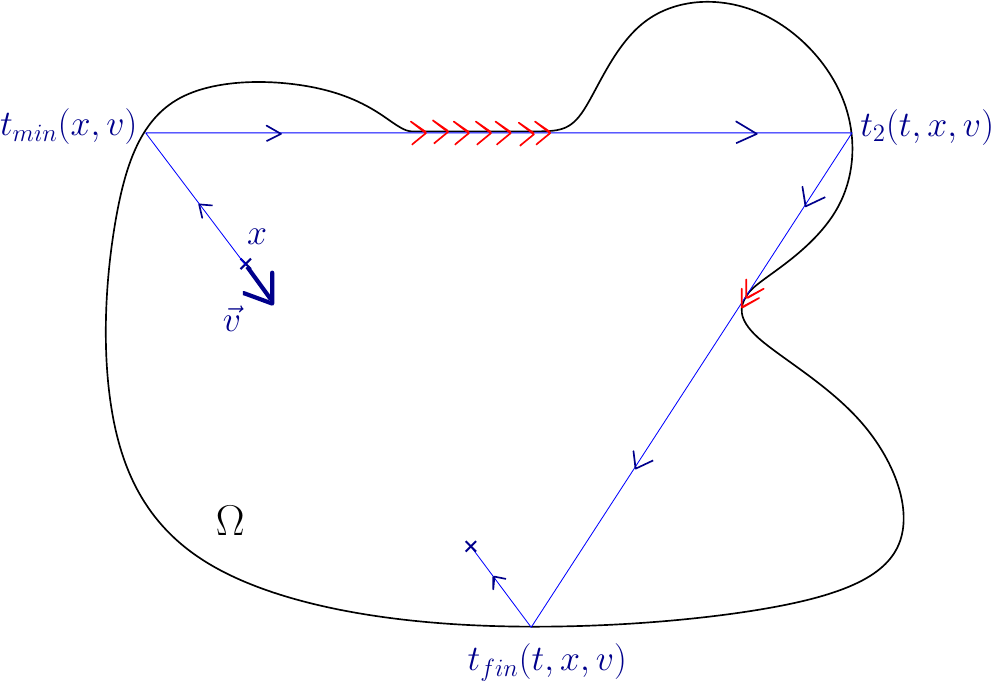}
\end{center}
\caption{\footnotesize Backward trajectory rolling on the boundary}
\label{fig:rolling}
\end{figure}
\begin{figure}[!h]
\begin{center}
\includegraphics[scale=0.68]{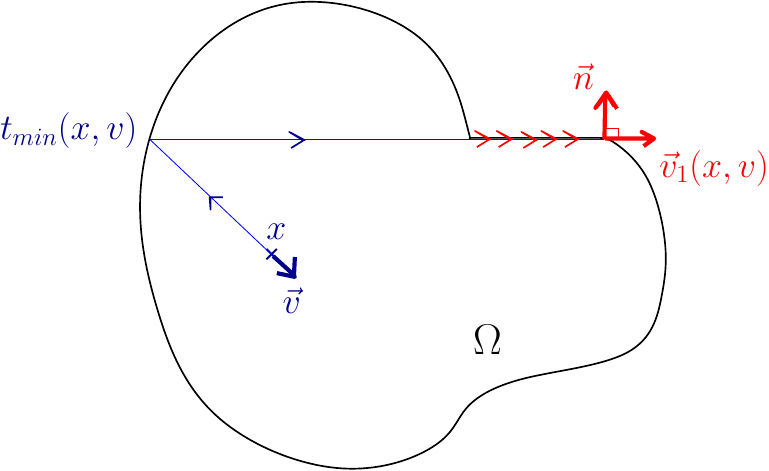}
\end{center}
\caption{\footnotesize Backward trajectory that reaches an end}
\label{fig:stop}
\end{figure}
\bigskip

Therefore, in order to prove the well-posedness of the transport equation $\eqref{transporteq}-\eqref{transportBC}$, we follow the ideas developed in \cite{Gu5} and \cite{Hw}, which consist of studying the backward trajectories that can lead to a point $(t,x,v)$, combined with the idea of countably many collisions in finite time used in \cite{JM}. However, we have to deal with the issues described above and to do so we introduce a new classification of possible interactions with the boundary (see Definition $\ref{def:rebounds}$). We also extend the result of \cite{Ta}, in terms of pair $(x,v)$ leading to infinitely many rebounds in finite time, to the whole domain $\Omega$ (Proposition $\ref{prop:nbrebounds}$). To do so we link up the study on the boundary made in \cite{Ta} with the Lebesgue measure on $\Omega$ by artificially creating volume on $\partial\Omega$ thanks to time and a foliation of the domain by parallel trajectories.

\bigskip
The section is divided as follows. First of all we shall describe and classify the collisions with the boundary in order to describe very accurately the backward trajectories of a point $(x,v)$ in $\partial\Omega\times \R^d$. We will name trajectory or characteristic any solution $(X(t,x,v),V(t,x,v))$ satisfying the initial condition $(X(0,x,v),V(0,x,v)) = (x,v)$, the boundary condition $\eqref{transportBC}$ and satisfying, in $\Omega$,
$$\left\{\begin{array}{rl} \displaystyle{\frac{dX}{dt} }&\displaystyle{= V} \vspace{2mm}\\\vspace{2mm} \displaystyle{\frac{dV}{dt} }&=\displaystyle{0.} \end{array}\right.$$
This will give us an explicit form for the characteristics and allow us to link $u(t,x,v)$ with $u_0(x^*,v^*)$, for some $x^*$ and $v^*$. Finally, we will show that the function we constructed is, indeed, a solution to the transport equation with initial data $u_0$ and specular boundary condition and that such a solution is unique.


\subsection{Study of rebounds on the boundary}\label{subsec:rebounds}

As mentionned in the introduction of this section, when a particle reaches a point at the boundary with a velocity $v$ it can bounce back (Fig.$\ref{fig:bouncing}$), keep moving straight (Fig.$\ref{fig:rolling}$) or stop moving because the specular reflection does not allow it to do anything else (Fig.$\ref{fig:stop}$), which is physically unexpected. The next definition gives a partition of the points at the boundary which takes into account those properties.

\bigskip
\begin{defi}\label{def:rebounds}
We define here a partition of $\partial\Omega\times\R^d$ that focuses on the outcome of a collision in each of the sets.
\begin{itemize}
\item The set coming from a rebound without rolling
$$\Omega_{rebounds} = \left\{(x,v) \in \partial\Omega\times\R^d: \:v\cdot n(x) <0 \right\}.$$
\item The set coming from rolling on the boundary
$$\Omega_{rolling} = \left\{(x,v) \in \partial\Omega\times\R^d: \:v\cdot n(x) = 0 \:\mbox{and}\:\: \exists \delta>0, \forall t\in [0,\delta], \:x-vt \in \bar{\Omega} \right\}.$$
\item The set of only starting points
$$\Omega_{stop} = \left\{(x,v) \in \partial\Omega\times\R^d: \:v\cdot n(x) = 0 \:\mbox{and}\:\: \forall \delta>0, \exists t\in [0,\delta], \:x-vt \notin \bar{\Omega} \right\}.$$
\item The set coming from straight line
$$\Omega_{line} = \left\{(x,v) \in \partial\Omega\times\R^d: \:v\cdot n(x) > 0 \right\}.$$
\end{itemize}
\end{defi}

\bigskip
One has to notice that any point of $\Omega_{line}$ indeed comes from a straight line arriving at $x$ with direction $v$ since $\Omega$ is open and is $C^1$ (so there is no cusp).
\\ In order to understand the behaviour expected at $\Omega_{stop}$ we have the following proposition. The proof of it gives insight into the nature of specular reflections.

\bigskip
\begin{prop}\label{prop:stop}
If we have $(x,v)$ in $\Omega_{stop}$ then there is no trajectory with specular boundary reflections that leads to $(x,v)$.
\end{prop}
\bigskip

\bigskip
\begin{proof}[Proof of Proposition $\ref{prop:stop}$]
Let us assume the contrary, that is to say $(x,v)$ is in $\Omega_{stop}$ comes from a trajectory with specular boundary reflection.
\par We have that $(x,v)$ belongs to $\partial\Omega\times\R^d$ and so if $(x,v)$ comes from a straight line it can only be (by definition of trajectories) a line containing $x$ with direction $v$ which means that $(x,v)$ comes from $\{(x-vt,v), \: t\in [0,T] \}$, for some $T>0$. But the trajectory is necessarily in $\bar{\Omega}$ and this is in contradiction with the definition of $\Omega_{stop}$.
\par Therefore, $(x,v)$ must come from a rebound after a straight line trajectory. But again we obtain a contradiction because the velocity before the rebound is $\mathcal{R}_x(v) = v$ and the backward trajectory is the one studied above.
\end{proof}
\bigskip

Now we have our partition of points on the boundary of $\Omega$, we are able to generate the backward trajectory associated to a starting point $(x,v)$ in $\bar{\Omega}\times\R^d$. The first step towards its resolution is to find the first point of real collision (if it exists) that generates $(x,v)$ (see Fig.$\ref{fig:bouncing}$). The next proposition-definition proves mathematically what the figure shows.

\bigskip
\begin{prop}\label{prop:tmin}
Let $\Omega$ be an open, bounded and $C^1$ domain in $\R^d$.
Let $(x,v)$ be in $\bar{\Omega}\times \R^d$, then we can define
$$t_{min}(x,v) = \max\left\{t\geq 0: \: x-vs \in \bar{\Omega}, \: \forall \:0\leq s \leq t\right\}.$$
Moreover we have the following properties:
\begin{enumerate}
\item if there exists $t$ in $(0,t_{min}(x,v))$ such that $x-vt$ hits $\partial\Omega$ then $(x-vt,v)$ belongs to $\Omega_{rolling}$.
\item $t_{min}(x,v) = 0$ if and only if $(x,v)$ belongs to $\Omega_{stop} \cup \Omega_{rebounds}$.
\item $\left(x-vt_{min}(x,v), v \right)$ belongs to $\Omega_{stop} \cup \Omega_{rebounds}$.
\end{enumerate}
\end{prop}
\bigskip

Property $(1)$ emphasises the fact that if, on the straight line between $x$ and $x-vt_{min}(x,v)$, the particle hits the boundary it will not be reflected and so just rolls on. Then property $(2)$ tells us than $t_{min(x,v)}$ is always strictly positive except if $(x,v)$ does not come from any trajectory of a particle or if it is the outcome of a rebound without rolling. Finally, property $(3)$ finishes the study since at $x-vt_{min}(x,v)$ the particles either come from a reflection (case $\Omega_{rebounds}$), and we can keep tracking backwards, or started its trajectory at $x-vt_{min}(x,v)$ (case $\Omega_{stop}$).

\bigskip
\begin{proof}[Proop of Proposition $\ref{prop:tmin}$]
First of all we have that $\Omega$ is bounded and so there exists $R$ such that $\bar{\Omega} \subset B(0,R)$, the ball of radius $R$ in $\R^d$.
\par Then we notice that $0$ belongs to 
$$A(x,v) = \left\{t\geq 0: \: x-vs \in \bar{\Omega}, \: \forall \:0\leq s \leq t\right\}.$$
Therefore $A(x,v)$ is not empty. Moreover, this set is bounded above by $2R/\norm{v}$ since for all $t$ in $A(x,v)$
$$R > \norm{x-vt} \geq t\norm{v} - \norm{x}.$$
Therefore we can talk about the supremum $t_{min}(x,v)$ of $A(x,v)$. Let $(t_n)_{n\in \N}$ be increasing  sequence in $A(x,v)$ that tends to $t_{min}(x,v)$. As $\bar{\Omega}$ is closed we have that $x-vt_{min}(x,v)$ belongs to $\bar{\Omega}$. Then, if $0\leq s < t_{min}(x,v)$ there exists $n$ such that $0 \leq s \leq t_n$ and so, by the property of $t_n$, $x-vs$ is in $\bar{\Omega}$. This conclude the fact that $t_{min}(x,v)$ belongs to $A(x,v)$ and so is a maximum.

\bigskip
We now turn to the proof of properties.
\par Let $(x,v)$ be in $\bar{\Omega}$ and $0 < t < t_{min}(x,v)$ such that $x-vt$ belongs to $\partial\Omega$. Then for all $0<t_1 <t < t_2 < t_{min}(x,v)$, $x-v t_1$ and $x-vt_2$ are in $\bar{\Omega}$ and so, by the definition of an exterior normal to a surface we have
$$\left[(x-vt)-(x-vt_1)\right]\cdot n(x-vt) \geq 0 \:\:\mbox{and}\: \left[(x-vt)-(x-vt_2)\right]\cdot n(x-vt) \geq 0,$$
which gives $v\cdot n(x-vt) =0$.
\\Moreover, since $t_2$ belongs to $A(x,v)$, for all $s$ in $[0,t_2-t]$, $(x-vt) -vs$ is in $\bar{\Omega}$, which means that $(x-vt,v)$ belongs to $\Omega_{rolling}$.
\par Property $(2)$ is direct since if $t_{min}(x,v)=0$ then for all $t>0$, there exists $0<s\leq t$ such that $x-vs$ does not belong to $\bar{\Omega}$ and then $v \cdot n(x) \leq 0$. So $(x,v)$ belongs to $\Omega_{rebounds}$, if $v \cdot n(x)>0$, or to $\Omega_{stop}$.
\par Finally, property $(3)$ is straightforward since $x-vt_{min}(x,v)$ is in $\partial\Omega$ (because $\Omega$ is open) and since for all $0\leq t \leq t_{min}(x,v)$, $x-vt$ is in $\bar{\Omega}$. Thus $\left[(x-vt_{min}) - (x-vt)\right]\cdot n(x-vt_{min}(x,v)) \geq 0$, which yields $v \cdot n(x-vt_{min}(x,v)) \leq 0$.Then, by the definition of $A(x,v)$ and the fact that $t_{min}(x,v)$ is its maximum, we have that either $(x-vt_{min}(x,v),v)$ belongs to $\Omega_{rebounds}$ or belongs to $\Omega_{stop}$. 
\end{proof}
\bigskip

Up to now we focused solely on the case of the first possible collision with the boundary. In order to conclude the study of rebounds for any given characteristics we have to, in some sense, count the number of rebounds without rolling that can happen in finite time. This is the purpose of the next proposition.

\bigskip
\begin{prop}\label{prop:nbrebounds}
Let $\Omega$ be a $C^1$ open, bounded domain in $\R^d$ and let $(x,v)$ be in $\bar{\Omega}\times \R^d$.
\\Then for all $t \geq 0$ the trajectory finishing at $(x,v)$ after a time $t$ has at most a countable number of rebound without rolling.
\par Moreover, this number is finite almost surely with respect to the Lebesgue measure on $\bar{\Omega}\times \R^d$
\end{prop}
\bigskip

\bigskip
\begin{proof}[Proof of Proposition $\ref{prop:nbrebounds}$]
The fact that there is countably many rebounds without rolling comes directly from the fact that $t_{min}(x,v)>0$ except if $(x,v)$ is a starting/stopping point (and then did not move from $0$ to $t$) or if $(x,v)$ is the outcome of a rebound (and so comes from $(x,\mathcal{R}_x(v))$ which belongs to $\Omega_{line}$, implying that $t_{min}(x,\mathcal{R}_x(v)) >0$).

\bigskip
Now we shall prove that the set of points in $\bar{\Omega}\times \R^d$ which lead to an infinite number of rebounds in a finite time is of measure $0$. To do so, we first need some definitions. The measure  $\mu$ in $\bar{\Omega}\times \R^d$  is the one induced by the Lebesgue measure and we denote by $\lambda$ the measure on $\partial\Omega\times\R^d$ (see section $1.7$ of \cite{Ta}).
\\ We will also denote
\begin{eqnarray*}
\Omega &=& \left\{(x,v) \in \Omega\times \left(\R^d-\{0\}\right) \:\:\mbox{coming from an infinite number of rebounds}\right\},
\\ \Omega_\partial &=& \left\{(x,v) \in \partial\Omega\times \left(\R^d-\{0\}\right) \:\:\mbox{coming from an infinite number of rebounds}\right\}.
\end{eqnarray*}
We know (\cite{Ta} Lemma $1.7.1$) that $\lambda(\Omega_\partial) = 0$ and we are going to establish a link between the measure of $\Omega$ and the one of $\Omega_{\partial}$. Those two sets do not live in the same topology nor same dimension and so we build a function that artificially recreates them via time.
\par Because $\Omega$ is bounded we can find time $T_M>0$ such that for all $x$ in $\bar{\Omega}$ and $v$ in $\R^d-\{0\}$, $\left(x-T_M v/\norm{v}\right)$ does not belong to $\bar{\Omega}$. Furthermore, in the same way as for $t_{min}(x,v)$, we can define, for $(x,v)$ in $\bar{\Omega}\times \R^d$, 
$$T(x,v) = \left\{\begin{array}{rl} &\displaystyle{\min\{t >0 :\: x+vt \in \partial\Omega\}\:\mbox{if}\: (x,v) \in \Omega \cup \Omega_{rebounds}} \vspace{2mm} \\\vspace{2mm} & \displaystyle{0 \:\mbox{otherwise}} \end{array}\right..$$
We define the following function which is clearly $C^1$.
$$\function{F}{[0,T_M]\times\R^d\times\left(\R^d-\{0\}\right)}{\R^d\times\left(\R^d-\{0\}\right)}{(t,x,v)}{(x+\frac{v}{\norm{v}}t,v).}$$
We also define the set
$$B = \left\{(t,x,v): \: x \in \partial\Omega, \: v \in (\R^d-\{0\}), \: t \in [0,T(x,v))\right\}.$$
and claim that  $F$ is injective on the set $B$. Indeed, if $(t,x,v)$ and $(t^*,x^*,v^*)$ are in $B$ such that $F(t,x,v)=F(t^*,x^*,v^*)$ then $v=v^*$ and $x + tv/\norm{v} = x^* + t^*v/\norm{v}$.
\\Let assume that $t^*>t$, therefore we have that 
$$x = x^*+ (t^*-t)\frac{v}{\norm{v}} \in \partial\Omega$$
and thus $t^*-t \geq T(x^*,v)$. However, $t^* \leq T(x^*,v)$ so we reach a contradiction and $t^*\leq t$. By symmetry we have $t=t^*$ and then $x=x^*$. We also notice that $[0,T_M]\times\Omega_{stop}$ and $[0,T_M]\times\Omega_{rolling}$ do not intersect $B$.
\par Finally we have that $\Omega = F\left(B\cap \left([0,T_M]\times\Omega_{\partial}\right)\right)$.
Indeed, if $(t,x,v)$ belongs to $B\cap \left([0,T_M]\times\Omega_{\partial}\right)$ then $F(t,x,v) = (x+ tv/\norm{v},v)$ and $x+tv/\norm{v}$ is in $\Omega$ and its first rebound backward in time is $(x,v)$ which lead to infinitely many rebounds in finite time. Therefore
$$x+t\frac{v}{\norm{v}} \in \Omega.$$
The converse is direct, by considering the first collision with the boundary of the backward trajectory starting at $(x,v)$ in $\Omega$.

\bigskip
All those properties allow us to compute $\mu(\Omega)$ by a change of variable in $B \cap \Omega_{\partial}$.

\begin{eqnarray*}
\mu(\Omega) &=& \mu(F\left(B\cap \left([0,T_M]\times\Omega_{\partial}\right)\right))
 \\                  &=&  \int_{\bar{\Omega}\times \R^d} \mathbf{1}_{F\left(B\cap \left([0,T_M]\times\Omega_{\partial}\right)\right)}(x,v)dxdvdt
  \\                 &=&  \int_{B\cap \left([0,T_M]\times\Omega_{\partial}\right)}\abs{\mbox{Jac}(F^{-1})} d\lambda(x,v)dt
  \\                 &\leq& T_M \sup\limits_{[0,T_M]\times \bar{\Omega}}\left(\abs{\mbox{Jac}(F^{-1})} \right)\lambda(\Omega_{\partial}) =0.
\end{eqnarray*}
\end{proof}
\bigskip


\subsection{Description of characteristics}\label{subsec:characteristics}

In the previous section we derived all the relevant properties of when, where and how a trajectory can bounce against the boundary of $\Omega$. As was shown, the characteristic starting from a point $(t,x,v)$ in $\R^+\times\bar{\Omega}\times \R^d$ is the backward trajectory satisfying specular boundary reflections that leads to $(x,v)$ in time $t$. Basically, it consists in a straight line as long as it stays inside $\Omega$ or it rolls on the boundary. Then it reaches a boundary point where it does not move any more ($\Omega_{stop}$) or bounces back ($\Omega_{rebounds}$).
\par Thanks to Proposition $\ref{prop:nbrebounds}$ we can generate the countable (and almost surely finite) sequence of collisions with the boundary associated to the future point $(x,v)$. We shall construct it by induction. We consider $(x,v)$ in $\bar{\Omega}\times\R^d$.

\begin{itemize}
\item \textbf{Step $1$: initialisation: } we define
$$\left\{\begin{array}{rl} \displaystyle{x_0(x,v) }&\displaystyle{= x,}\vspace{2mm}\\\vspace{2mm} \displaystyle{v_0(x,v) }&\displaystyle{=v,} \vspace{2mm} \\\vspace{2mm} \displaystyle{t_0(x,v) }&\displaystyle{= 0.}\end{array}\right.$$
\item \textbf{Step $2$: induction: } if $(x_k(x,v),v_k(x,v)) \in \Omega_{stop}$ then we define
$$\left\{\begin{array}{rl} \displaystyle{x_{k+1}(x,v) }&\displaystyle{= x_{k}(x,v),}\vspace{2mm}\\\vspace{2mm} \displaystyle{v_{k+1}(x,v) }&\displaystyle{=v_{k}(x,v),} \vspace{2mm} \\ \vspace{2mm}\displaystyle{t_{k+1}(x,v) }&\displaystyle{= +\infty,}\end{array}\right.$$
if $(x_k(x,v),v_k(x,v)) \notin \Omega_{stop}$ then we define
$$\left\{\begin{array}{rl} \displaystyle{x_{k+1}(x,v) }&\displaystyle{= x_k(x,v)- v_k(x,v)t_{min}(x_k(x,v),v_k(x,v)),} \vspace{2mm}\\\vspace{2mm} \displaystyle{v_{k+1}(x,v) }&\displaystyle{=\mathcal{R}_{x_{k+1}(x,v)}(v_{k}(x,v)),} \vspace{2mm} \\ \vspace{2mm} \displaystyle{t_{k+1}(x,v) }&\displaystyle{= t_k (x,v) +  t_{min}(x_k(x,v),v_k(x,v)).}\end{array}\right.$$
\end{itemize}

\bigskip
\begin{remark}
Let us make a few comments on the accuracy of the sequence we just built.
\begin{enumerate}
\item Looking at Proposition $\ref{prop:tmin}$, we know that at each step (apart from $0$) we necessary have that $(x_k(x,v), v_k(x,v))$ belongs to either $\Omega_{stop}$ or $\Omega_{rebounds}$ and so the characteristic stops for ever (case $1$ in induction)  or bounces without rolling and start another straight line (case $2$). Thus the sequence of footprints defined above captures the trajectories as long as there are rebounds and then becomes constant once the trajectory reach a stopping point.
\item If $t_{min}(x_k(x,v),v_k(x,v))=0$ for some $k>0$ then, by properties $2.$ and $3.$ of Proposition $\ref{prop:tmin}$, we must have $(x_k(x,v),v_k(x,v)) \in \Omega_{stop}$ (since $v_k(x,v)$ is the specular reflection at $x_k(x,v)$ of $v_{k-1}(x,v)$ and $(x_k(x,v),v_{k-1}(x,v))$ is in $\Omega_{rebounds} \cup \Omega_{stop}$). Thus, $(t_k(x,v))_{k \in \N}$ is strictly increasing as long as it does not reach the value $+\infty$, where it remains constant.
\end{enumerate}
\end{remark}
\bigskip

Finally, it remains to connect the time variable to those quantities. In fact, the time will determine how many rebounds can lead to $(x,v)$ in a time $t$. The reader must remember that the backward trajectory can lead to a point in $\Omega_{stop}$ before time $t$.
\par Since the characteristics method helps us to find the value of the solution of the transport equation at a given point using its trajectory, the next definition links a triplet $(t,x,v)$ to the first rebound of the trajectory that leads to $(x,v)$ in a time $t$.

\bigskip
\begin{defi}\label{def:solution}
Let $\Omega$ be an open, bounded and $C^1$ domain in $\R^d$.
\\ Let $(t,x,v)$ be in $\R^+\times\bar{\Omega}\times\R^d$. Then we can define
$$ n(t,x,v) =\left\{\begin{array}{l} \displaystyle{\max\{k\in \N: \: t_k(x,v) \leq t \},} \:\:\mbox{if it exists}, \vspace{2mm}\\\vspace{2mm} \displaystyle{+\infty,} \:\:\mbox{if $(t_k(x,v))_k$ is bounded by $t$.}\end{array}\right.$$
The last rebound is then define by
\begin{itemize}
\item if $n(t,x,v) <+\infty$ and $t_{n(t,x,v)+1}=+\infty$, then 
$$\left\{\begin{array}{l} \displaystyle{x_{fin}(t,x,v) = x_{n(t,x,v)}(x,v),}\vspace{2mm}\\\vspace{2mm} \displaystyle{v_{fin}(t,x,v) = v_{n(t,x,v)}(x,v),}\vspace{2mm} \\\vspace{2mm} \displaystyle{t_{fin}(t,x,v) = t,}\end{array}\right.$$

\item if $n(t,x,v) <+\infty$ and $t_{n(t,x,v)+1}<+\infty$, then 
$$\left\{\begin{array}{l} \displaystyle{x_{fin}(t,x,v) = x_{n(t,x,v)}(x,v),}\vspace{2mm}\\\vspace{2mm} \displaystyle{v_{fin}(t,x,v) = v_{n(t,x,v)}(x,v),} \vspace{2mm} \\ \vspace{2mm} \displaystyle{t_{fin}(t,x,v) = t_{n(t,x,v)}(x,v),}\end{array}\right.$$

\item if $n(t,x,v) =+\infty$,  then
$$\left\{\begin{array}{l} \displaystyle{x_{fin}(t,x,v) = \lim\limits_{k\to +\infty} x_{k}(x,v),} \vspace{2mm}\\\vspace{2mm} \displaystyle{v_{fin}(t,x,v) = \lim\limits_{k\to +\infty}v_{k}(x,v),} \vspace{2mm} \\ \vspace{2mm} \displaystyle{t_{fin}(t,x,v) = \lim\limits_{k \to +\infty}t_{k}(x,v).}\end{array}\right.$$
\end{itemize}
\end{defi}

\bigskip
\begin{remark}\label{remark:endtrajectory}
Let us make a few comments on the definition above and the existence of limits.
\begin{enumerate}
\item After the last rebound, occuring at $t_{n(t,x,v)}$, the backward trajectory can only be a straight line during the time period $t-t_{n(t,x,v)}$ (see Fig.$\ref{fig:bouncing}$). That is why we defined $t_{fin}(t,x,v) = t_{n(t,x,v)}$ if we reached a point on $\Omega_{rebounds}$ and $t_{fin}(t,x,v) = t$ if the last rebound reaches $\Omega_{stop}$ (the trajectory can only start from there).
\item In the last case of the definition, we remind the reader that $(t_k(x,v))_{k\in \N}$ is strictly increasing and so converges if bounded by $t$. But then, because $(\norm{v_{k}(x,v)})_{k\in\N}$ is constant and $x_k(x,v)=x_{k-1}(x,v)-t_{min}(x_k(x,v),v_{k}(x,v))v_{k}(x,v)$, we have that $(x_k(x,v))_{k\in\N}$ is a Cauchy sequence.
\item The last case in Definition $\ref{def:solution}$ almost surely never happens, as proved in Proposition $\ref{prop:nbrebounds}$.
\end{enumerate}
\end{remark}
\bigskip

To conclude this study of the characteristics we just have to make one more comment. We studied the characteristics that go backward in time because it simplifies the construction of a solution to the free transport equation. However, it is easy to prove (just requires the inductive construction of $v_k$ and $x_k$) that the forward trajectory of $(x,v)$ during a period $t$ is the backward trajectory over a period $t$ of $(x,-v)$. This gives the final proposition.

\bigskip
\begin{prop}\label{prop:characteristics}
Let $\Omega$ be an open, bounded and $C^1$ domain in $\R^d$.
Then for all $(x,v)$ in $\bar{\Omega}\times \R^d$ we have existence and uniqueness of the characteristic $(X_t(x,v),V_t(x,v))$ given by, for all $t\geq 0$,

\begin{eqnarray*}
X_t(x,v) &=& x_{fin}(t,x,-v)+(t-t_{fin}(t,x,-v))v_{fin}(t,x,-v),
\\V_t(x,v) &=& -v_{fin}(t,x,-v).
\end{eqnarray*}
Moreover, we have that $V_t(x,v) = O_{t,x,v}(v)$ with $O_{t,x,v}$ an orthogonal transformation, and that for almost every $(x,v)$ in $\bar{\Omega}\times \R^d$ we have the following

\begin{equation}\label{bijectioncharacteristics}
\forall t\geq 0,\quad (x,v) = (X_t(X_t(x,-v),-V_t(x,-v)),V_t(X_t(x,-v),-V_t(x,-v))).
\end{equation}
\end{prop}
\bigskip

\bigskip
\begin{proof}[Proof of Proposition $\ref{prop:characteristics}$]
By construction we have that 
$$O_{t,x,v} = \mathcal{R}_{x_{fin}(t,x,v)} \circ \dots \circ \mathcal{R}_{x_{1}(t,x,v)}.$$

It only remains to show the last equation $\eqref{bijectioncharacteristics}$, but it follows directly from the fact that the backward trajectory of $(x,v)$ is the forward trajectory of $(x,-v)$. 
\par We can reach a point on $\Omega_{stop}$ after a time $t_1$ and so the forward trajectory of that point during a time $t>t_1$ does not come back to the original point (since we stayed in $\Omega_{stop}$ for a period $t-t_1$).
\par However, the set of points that reach $\Omega_{stop}$ belongs to the set of points that bounce infinitely many times in a finite time and this set is of measure zero (see Proposition $\ref{prop:nbrebounds}$).
\end{proof}
\bigskip


\subsection{Existence and uniqueness of solution to $\eqref{transporteq}-\eqref{transportBC}$}\label{subsec:transport}

\subsubsection{Proof of uniqueness}
The uniqueness of a solution with $u_0$ in $C^1_x \cap L^2_{x,v}$ comes directly from the fact that we have a preserved quantity through time, thanks to the specular reflection property. Indeed, let us assume that $u$ is a solution to our free transport equation satisfying specular boundary condition and the initial value problem $u_0$. Then, a mere integration by part gives us
$$\forall t\geq 0, \quad\norm{u(t,\cdot,\cdot)}^2_{L^2_{x,v}}=\norm{u_0}^2_{L^2_{x,v}},$$
which directly implies the uniqueness of a solution, since the transport equation $\eqref{transporteq}$ is linear.


\subsubsection{Construction of the solution}

It remains to construct a function $u$ that will be constant on the characteristic trajectories and check that we indeed obtain a function that is differentiable in $t$ and $x$ which satisfies the transport equation. The first point of Remark $\ref{remark:endtrajectory}$ gives us the answer as we expect the following behaviour 
$$u(t,x,v) = u(t-t_1(x,v),x_1(x,v),v_1(x,v)) =\dots= u(t-t_k(x,v),x_k(x,v),v_k(x,v)),$$
up to the point where there are no more rebound in the time interval $[0,t]$. From there we continue in a straight line.
\par Thus, we define: $\forall (t,x,v) \in \R^+\times\bar{\Omega}\times\R^d,$
$$u(t,x,v) = u_0\left(x_{fin}(t,x,v)-(t-t_{fin}(t,x,v))v_{fin}(t,x,v),v_{fin}(t,x,v)\right).$$


\subsubsection{Boundary and initial conditions}

First of all, $u$ satisfies the initial condition $\eqref{transportCI}$ as $n(0,x,v)=0$ (since $t_{min}(x,v)\geq 0$).
\par $u$ also satifies the specular boundary condition $\eqref{transportBC}$. Indeed, if $(x,v)$ is in  $\Lambda^+$, then either $v\cdot n(x)=0$ and the result is obvious since $\mathcal{R}_x(v)=v$, or $v\cdot n(x) > 0$ and thus $(x,\mathcal{R}_x(v))$ belongs to $\Omega_{rebounds}$ so $t_{min}(x,\mathcal{R}_x(v))=0$ (Proposition $\ref{prop:tmin}$). An easy induction shows 
$$x_{k}(x,v)=x_{k+1}(x,\mathcal{R}_x(v)),\:v_{k}(x,v)=v_{k+1}(x,\mathcal{R}_x(v)),\: t_{k}(x,v)=t_{k+1}(x,\mathcal{R}_x(v)),$$
for all $k$ in $\N$.
\\The last equality gives us that $n(t,x,v)=n(t,x,\mathcal{R}_x(v))-1$ and therefore, combined with the two other equalities,
$$x_{fin}(t,x,v)=x_{fin}(t,x,\mathcal{R}_x(v)), \: v_{fin}(t,x,v)=v_{fin}(t,x,\mathcal{R}_x(v)),$$ $$t_{fin}(t,x,v)=t_{fin}(t,x,\mathcal{R}_x(v)),$$
which leads to the specular reflection boundary condition.


\subsubsection{Time differentiability}

Here we prove that $u$ is differentiable in time on $\R^+$. Let us fix $(x,v)$ in $\Omega\times\R^d$.
\\ By construction, we know that $n(t,x,v)$ is piecewise constant. Since $(t_k(x,v))_{k\in \N}$ is strictly increasing up to the step where it takes the value $+\infty$, for $t_{k}(x,v) <t<t_{k+1}(x,v)$ we have that for all $s \in \R$ such that $t_{k}(x,v) <t+s<t_{k+1}(x,v),$
$$x_{fin}(t,x,v)=x_{fin}(t+s,x,v), \: v_{fin}(t,x,v)=v_{fin}(t+s,x,v),$$ $$t_{fin}(t,x,v)=t_{fin}(t+s,x,v).$$
Therefore, we have that

\begin{eqnarray*}
&&\frac{u(t+s,x,v)-u(t,x,v)}{s} 
\\&&= \frac{u_0(x_{fin}-(t+s-t_{fin})v_{fin},v_{fin})-u_0(x_{fin}-(t-t_{fin})v_{fin},v_{fin})}{s}
\\                                       &&\underset{s \to 0}{\to} -v_{fin}\cdot \left(\nabla_x u_0\right)(x_{fin}-(t-t_{fin})v_{fin},v_{fin}),
\end{eqnarray*}
because $u_0$ is $C^1$ in $x$. So $u$ is differentiable at $t$ if $t$ in strictly between two times $t_k(x,v)$. We thus find that $u$ is differentiable at $t$ and that its derivative is continuous (since $x_{fin}$, $v_{fin}$ and $t_{fin}$ are continuous when $x$ and $v$ are fixed).
\par In the case $t=t_k(x,v)$ we can use what we just proved to show that we have the existence of right (except for $t=0$) and left limits of $\partial_t u (t,x,v)$ as $t$ tends to $t_k(x,v)$. We use the specular reflection boundary condition of $u_0$ together with the fact that it is $C^1$ in $x$ and that $t_k(x,v)=t_{k+1}(x,\mathcal{R}_x(v))$ to obtain the equality of the two limits.


\subsubsection{Space differentiability and solvability of the transport equation}

Here we prove that $u$ is differentiable in $x$ in $\Omega$, which follows directly from the time differentiability. Let us fix $t$ in $\R^+$ and $v$ in $\R^d$, we shall study the differentiability of $u(t,\cdot,v)$ in the direction of $v$.
\par $\Omega$ is open and so
$$\forall x \in \Omega, \:\exists \delta>0,\:\forall s \in [-\delta,\delta],\quad x+sv \in \Omega.$$
Thanks to the inductive construction, one find easily that
$$u(t,x+sv,v) = u(t-s,x,v).$$
Therefore, since $u$ is time differentiable, we have that $u(t,\cdot,v)$ admits a directional derivative in the direction of $v$ and that
$$\mbox{D}_{x}(v)(u)(t,x,v) = - \partial_t u(t,x,v).$$

%% file: Filling_of_vacuum_for_BE.bbl
\begin{thebibliography}{10}

\bibitem{Bri5}
{\sc Briant, M.}
\newblock Instantaneous exponential lower bound for solutions to the boltzmann
  equation with maxwellian diffusion boundary conditions.
\newblock {\em Kin. Rel. Mod. 8}, 2 (June 2015), 281--308.

\bibitem{Ca}
{\sc Carleman, T.}
\newblock Sur la th\'eorie de l'\'equation int\'egrodiff\'erentielle de
  {B}oltzmann.
\newblock {\em Acta Math. 60}, 1 (1933), 91--146.

\bibitem{Ce}
{\sc Cercignani, C.}
\newblock {\em The {B}oltzmann equation and its applications}, vol.~67 of {\em
  Applied Mathematical Sciences}.
\newblock Springer-Verlag, New York, 1988.

\bibitem{Ce1}
{\sc Cercignani, C., Illner, R., and Pulvirenti, M.}
\newblock {\em The mathematical theory of dilute gases}, vol.~106 of {\em
  Applied Mathematical Sciences}.
\newblock Springer-Verlag, New York, 1994.

\bibitem{JM}
{\sc Chen, J., and Yang, M.~Z.}
\newblock Linear transport equation with specular reflection boundary
  condition.
\newblock {\em Transport Theory Statist. Phys. 20}, 4 (1991), 281--306.

\bibitem{DV}
{\sc Desvillettes, L., and Villani, C.}
\newblock On the trend to global equilibrium in spatially inhomogeneous
  entropy-dissipating systems: the linear {F}okker-{P}lanck equation.
\newblock {\em Comm. Pure Appl. Math. 54}, 1 (2001), 1--42.

\bibitem{DV1}
{\sc Desvillettes, L., and Villani, C.}
\newblock On the trend to global equilibrium for spatially inhomogeneous
  kinetic systems: the {B}oltzmann equation.
\newblock {\em Invent. Math. 159}, 2 (2005), 245--316.

\bibitem{Gr1}
{\sc Grad, H.}
\newblock Principles of the kinetic theory of gases.
\newblock In {\em Handbuch der {P}hysik (herausgegeben von {S}. {F}l\"ugge),
  {B}d. 12, {T}hermodynamik der {G}ase}. Springer-Verlag, Berlin, 1958,
  pp.~205--294.

\bibitem{Gu5}
{\sc Guo, Y.}
\newblock Singular solutions of the {V}lasov-{M}axwell system on a half line.
\newblock {\em Arch. Rational Mech. Anal. 131}, 3 (1995), 241--304.

\bibitem{Gu6}
{\sc Guo, Y.}
\newblock Decay and continuity of the {B}oltzmann equation in bounded domains.
\newblock {\em Arch. Ration. Mech. Anal. 197}, 3 (2010), 713--809.

\bibitem{Hw}
{\sc Hwang, H.~J.}
\newblock Regularity for the {V}lasov-{P}oisson system in a convex domain.
\newblock {\em SIAM J. Math. Anal. 36}, 1 (2004), 121--171 (electronic).

\bibitem{HwVel}
{\sc Hwang, H.~J., and Vel{\'a}zquez, J. J.~L.}
\newblock Global existence for the {V}lasov-{P}oisson system in bounded
  domains.
\newblock {\em Arch. Ration. Mech. Anal. 195}, 3 (2010), 763--796.

\bibitem{Mo2}
{\sc Mouhot, C.}
\newblock Quantitative lower bounds for the full {B}oltzmann equation. {I}.
  {P}eriodic boundary conditions.
\newblock {\em Comm. Partial Differential Equations 30}, 4-6 (2005), 881--917.

\bibitem{Po}
{\sc Poritsky, H.}
\newblock The billiard ball problem on a table with a convex boundary---an
  illustrative dynamical problem.
\newblock {\em Ann. of Math. (2) 51\/} (1950), 446--470.

\bibitem{PW}
{\sc Pulvirenti, A., and Wennberg, B.}
\newblock A {M}axwellian lower bound for solutions to the {B}oltzmann equation.
\newblock {\em Comm. Math. Phys. 183}, 1 (1997), 145--160.

\bibitem{Ta}
{\sc Tabachnikov, S.}
\newblock Billiards.
\newblock {\em Panor. Synth.}, 1 (1995), vi+142.

\bibitem{Ta1}
{\sc Tabachnikov, S.}
\newblock {\em Geometry and billiards}, vol.~30 of {\em Student Mathematical
  Library}.
\newblock American Mathematical Society, Providence, RI, 2005.

\bibitem{Vi2}
{\sc Villani, C.}
\newblock A review of mathematical topics in collisional kinetic theory.
\newblock In {\em Handbook of mathematical fluid dynamics, {V}ol. {I}}.
  North-Holland, Amsterdam, 2002, pp.~71--305.

\bibitem{Vi1}
{\sc Villani, C.}
\newblock Cercignani's conjecture is sometimes true and always almost true.
\newblock {\em Comm. Math. Phys. 234}, 3 (2003), 455--490.

\end{thebibliography}
